%% file: main.tex
\documentclass[letterpaper,11pt]{article}

\emergencystretch=\maxdimen
\usepackage{graphicx}
\usepackage{caption}
\usepackage{subcaption}
\usepackage{pdfpages} 
\usepackage[margin=1in]{geometry}
\usepackage[utf8x]{inputenc}
\usepackage{amsfonts}
\usepackage{amsmath}
\usepackage{amssymb}
\usepackage{amsthm}
\usepackage{complexity}
\usepackage{bbm}
\usepackage{mathtools}
\usepackage{thmtools,thm-restate}
\usepackage{mdframed}
\global\mdfdefinestyle{myframe}{leftmargin=.75in,rightmargin=.75in,linecolor=black,linewidth=1.5pt,innertopmargin=10pt,innerbottommargin=10pt}
\usepackage[inline]{enumitem}
\usepackage{graphics}
\usepackage{xspace}
\definecolor{darkgreen}{rgb}{0,0.5,0}
\definecolor{darkblue}{rgb}{0,0,0.5}
\usepackage{hyperref}
\usepackage{nicefrac}
\hypersetup{
    unicode=false,          % non-Latin characters in AcrobatΓÇÖs bookmarks
    colorlinks=true,        % false: boxed links; true: colored links
    linkcolor=darkblue,         % color of internal links (change box color with linkbordercolor)
    citecolor=darkgreen,        % color of links to bibliography
    filecolor=magenta,      % color of file links
    urlcolor=blue           % color of external links
}

\usepackage{algorithm}
\usepackage[noend]{algpseudocode}
\newfloat{protocol}{htbp}{loa}
\floatname{protocol}{Protocol}
\makeatletter\newcommand\l@protocol{\@dottedtocline{1}{1.5em}{2.3em}}\makeatother
\usepackage{tcolorbox}

\usepackage{xfrac}

\usepackage{array}

\newcounter{relctr} %% <- counter for relations
\everydisplay\expandafter{\the\everydisplay\setcounter{relctr}{0}} %% <- reset every eq
\newcommand\labeleq[1]{%
  \begingroup
    \refstepcounter{relctr}%
    \textnormal{(\Alph{relctr})}
    \originallabel{#1}%
  \endgroup
}
\AtBeginDocument{\let\originallabel\label} %% <- store original definition

\newtheorem{theorem}{Theorem}[section]
\newtheorem{lemma}[theorem]{Lemma}

\newtheorem{corollary}[theorem]{Corollary}
\newtheorem{claim}[theorem]{Claim}
\newtheorem*{claim*}{Claim}

\newtheorem{observation}[theorem]{Observation}

\usepackage[capitalize, nameinlink]{cleveref}
\crefname{theorem}{Theorem}{Theorems}
\Crefname{lemma}{Lemma}{Lemmas}
\Crefname{claim}{Claim}{Claims}
\Crefname{fact}{Fact}{Facts}
\Crefname{remark}{Remark}{Remarks}
\Crefname{observation}{Observation}{Observations}
\Crefname{line}{Line}{Lines}
\crefalias{AlgoLine}{line}
\Crefname{protocol}{Protocol}{Protocols}

\makeatletter
\newcounter{algorithmicH}% New algorithmic-like hyperref counter
\let\oldalgorithmic\algorithmic
\renewcommand{\algorithmic}{%
  \stepcounter{algorithmicH}% Step counter
  \oldalgorithmic}% Do what was always done with algorithmic environment
\renewcommand{\theHALG@line}{ALG@line.\thealgorithmicH.\arabic{ALG@line}}
\makeatother

\DeclareMathOperator{\argmax}{argmax}
\DeclareMathOperator{\supp}{supp}

\newcommand{\chainP}{CHAIN\textsubscript{$p$}\xspace}
\newcommand{\chainPn}{CHAIN\textsubscript{$p$}$(n)$\xspace}
\newcommand{\chainOnlyP}[1]{CHAIN\textsubscript{$#1$}\xspace}
\newcommand{\chain}[2]{CHAIN\textsubscript{$#1$}$(#2)$\xspace}

\newcommand{\e}{\varepsilon}
\newcommand{\eps}{\e}

\newcommand{\expected}[1]{\mathbb E \left[ #1 \right]}

\newcommand{\characteristic}{\mathbbm{1}}

\newcommand{\ee}{\expected}

\newcommand{\rb}[1]{\left( #1 \right)}

\newcommand{\cA}{\mathcal{A}}

\newcommand{\cG}{\mathcal{G}}
\newcommand{\cH}{\mathcal{H}}

\newcommand{\cF}{\mathcal{F}}

\newcommand{\Oval}{\textsc{OPT}\xspace}
\newcommand{\Oset}{\mathcal{O}\xspace}

\newcommand{\OAset}{\mathcal{O_A}\xspace}

\newcommand{\ALG}{{PRT}}
\newcommand{\ALGI}{{PRT_{\text{INDEX}}}}
\newcommand{\PRTpI}{{}PRT_{\text{\chainP}}}

\newcommand{\RSet}{{\mathcal{R}}}

\newcommand{\hi}{{i'}}

\newcommand{\maxcard}{Max-Card-$k$\xspace}
\newcommand{\maxcardpdf}{Max-Card-k}

\algnewcommand\myand{\textbf{and} }
\algnewcommand\myor{\textbf{or} }

\newcommand\Greedy{\textsf{Greedy}\xspace}

\renewcommand\vec{\mathbf}
\newcommand{\nnR}{{\mathbb{R}_{\geq 0}}}
\newcommand{\cupdot}{\mathbin{\mathaccent\cdot\cup}}

{\makeatletter
 \gdef\xxxmark{%
   \expandafter\ifx\csname @mpargs\endcsname\relax % in minipage?
     \expandafter\ifx\csname @captype\endcsname\relax % in figure/caption?
       \marginpar{xxx}% not in a caption or minipage, can use marginpar
     \else
       xxx % notice trailing space
     \fi
   \else
     xxx % notice trailing space-
   \fi}
 \gdef\xxx{\@ifnextchar[\xxx@lab\xxx@nolab}
 \long\gdef\xxx@lab[#1]#2{{\bf [\xxxmark #2 ---{\sc #1}]}}
 \long\gdef\xxx@nolab#1{{\bf [\xxxmark #1]}}
 \long\gdef\xxx@lab[#1]#2{}\long\gdef\xxx@nolab#1{}%
}

\title{The One-way Communication Complexity of Submodular Maximization with Applications to Streaming and Robustness}
\author{Moran Feldman\thanks{Email: \href{mailto:moranfe@cs.haifa.ac.il}{moranfe@cs.haifa.ac.il}. Research supported in part by the Israel Science Foundation (ISF) grant no. 1357/16.}\\University of Haifa \and Ashkan Norouzi-Fard\thanks{Email: \href{mailto:ashkannorouzi@google.com}{ashkannorouzi@google.com}.}\\Google Research \and Ola Svensson\thanks{Email: \href{mailto:ola.svensson@epfl.ch}{ola.svensson@epfl.ch}. Research supported by the Swiss National Science Foundation project 200021-184656 ``Randomness in Problem Instances and Randomized Algorithms.''} \\EPFL 
\and Rico Zenklusen\thanks{Email: \href{mailto:ricoz@math.ethz.ch}{ricoz@math.ethz.ch}. Research supported in part by Swiss National Science Foundation grants 200021\_184622 and 200021\_165866. This project has received funding from the European Research Council (ERC) under the European Union's Horizon 2020 research and innovation programme (grant agreement No 817750).}
\\ ETH Zurich}
\date{}

\begin{document}
\maketitle
\begin{abstract}
\input{000-abstract}
\end{abstract}

%\vspace{1cm}
%\textbf{For a shorter read}, we recommend \cref{sec:intro}, \cref{sec:prelim_model}, \cref{ssc:two_players_algs}, and \cref{sec:gen_hardness} up to and including \cref{sec:gen_hardness_intuition}. 
\pagenumbering{Alph}
\thispagestyle{empty}

\newpage
\pagenumbering{arabic}
\setcounter{page}{1}
%\tableofcontents
%\newpage

\input{100-introduction}

\input{200-preliminaries}
\input{600-submodular-expo-time.tex}
%\newpage
\input{400-general-hardness.tex}

%\newpage
\input{500-polytime.tex}
%\newpage
%\input{700-set-cover.tex}

\bibliographystyle{alpha}
\bibliography{ref}

\newpage
\appendix
\input{A5060-applications.tex}
\input{A40-alternative-model.tex}
\input{A20-kIndex.tex}

\input{A30-geometric-grouping.tex}
\input{A-missing-proofs.tex}
\input{A25-fractional-coverage.tex}

\input{A10-poly2players.tex}

\end{document}

%% file: 000-abstract.tex
We consider the classical problem of maximizing a monotone submodular function subject to a cardinality constraint, which, due to its numerous applications, has recently been studied in various computational models.
We consider a clean multi-player model that lies between the offline and streaming model, and study it under the aspect of one-way communication complexity. Our model captures the streaming setting (by considering a large number of players), and, in addition, two player approximation results for it translate into the robust setting.
We present tight one-way communication complexity results for our model, which, due to the above-mentioned connections, have multiple implications in the data stream and robust setting.

Even for just two players, a prior information-theoretic hardness result implies that no approximation factor above $1/2$ can be achieved in our model, if only queries to feasible sets, i.e., sets respecting the cardinality constraint, are allowed. We show that the possibility of querying infeasible sets can actually be exploited to beat this bound, by presenting a tight $2/3$-approximation taking exponential time, and an efficient $0.514$-approximation. To the best of our knowledge, this is the first example where querying a submodular function on infeasible sets leads to provably better results. Through the above-mentioned link to the robust setting, both of these algorithms improve on the current state-of-the-art for robust submodular maximization, showing that approximation factors beyond $1/2$ are possible. Moreover, exploiting the link of our model to streaming, we settle the approximability for streaming algorithms by presenting a tight $1/2+\varepsilon$ hardness result, based on the construction of a new family of coverage functions. This improves on a prior $1-1/e+\varepsilon$ hardness and matches, up to an arbitrarily small margin, the best known approximation algorithm. 

%
%The wide interest in this problems stems from numerous applications in a variety of areas, including combinatorial optimization, machine learning, natural language processing, and information retrieval. While the greedy algorithm is well-known to provide optimal approximations to this problem in a traditional computing model, modern applications naturally raised interest in other models, like map-reduce, streaming, robustness, and most recently also models involving adaptivity.
%

%% file: 100-introduction.tex
\section{Introduction}
\label{sec:intro}
A set function $f\colon 2^W \rightarrow \mathbb{R}$ over a finite ground set $W$ is \emph{submodular}  if 
\begin{align*}
  \label{eq:submodular}
  f(v\mid X) \geq f(v\mid Y)  \qquad \mbox{for all $X \subseteq Y \subseteq W$ and $v\in W \setminus Y$,}
\end{align*}
where, for a subset $S\subseteq W$ and an element $v\in W$, the value $f(v\mid S) = f(S \cup \{v\}) - f(S)$ is the marginal contribution of $v$ with respect to $S$. 
The definition of submodular functions captures the natural property of diminishing returns, and submodular functions have a rich history in optimization with  numerous applications (see, e.g., Schrijver's book~\cite{schrijver-book}). 

Already in 1978,  Nemhauser, Wolsey, and Fisher~\cite{nemhauser1978analysis} analyzed the following algorithm, which we refer to as \Greedy, for selecting the most valuable set $S\subseteq W$ of cardinality at most $k$.
\begin{enumerate}[label=(\roman*), itemsep=0em, topsep=0.2em, parsep=0em]
  \item Initially, let $S = \varnothing$.
  \item For $i=1, \ldots, k$: choose any $v\in \arg \max_{w\in W}  f(w\mid S)$ and set $S = S\cup\{v\}$.
\end{enumerate}
In words, the algorithm  greedily picks in each iteation an element with the largest marginal contribution with respect to the already selected elements $S$.
Assuming that $f$ is non-negative ($f(X) \geq 0$ for all $X\subseteq W$) and monotone ($f(X) \leq f(Y)$ if $X\subseteq Y$), Nemhauser et al.~\cite{nemhauser1978analysis} showed that \Greedy returns a $(1-1/e)$-approximate solution. Moreover, the approximation guarantee of $1-1/e$ is known to be tight~\cite{nemhauser1978best,feige1998threshold}.

In recent years, submodular function maximization has found several applications in problems related to data science and machine learning, including feature selection, sensor placement, and image collection summarization~\cite{SubmodularWWW, golovin2011adaptive, DBLP:conf/nips/Bach10, DBLP:conf/nips/DasDK12, DBLP:conf/icml/DasK11, DBLP:conf/nips/ZhengJCP14,DBLP:conf/acl/BairiIRB15}. 
These applications are often modeled as a maximization of a non-negative and monotone submodular function. 
For example, if we wish to summarize an image collection, we would like to select $k$ images that cover different topics, and this
objective can be modeled as a (non-negative and monotone) submodular function.
While \Greedy gives the best possible guarantee for solving this problem  in traditional computing (when the entire instance is accessible to the algorithm at all times),  the requirements stipulated by modern applications, often involving huge data sets, make such algorithms inadequate. 

This motivates, together with the inherent theoretical interest, the study of submodular function maximization in new models of computation. 
Indeed, in recent years, there has been substantial interest in submodular function maximization with respect to limited memory (so-called data stream algorithms)~\cite{badanidiyuru2014streaming,buchbinder2015online,DBLP:journals/corr/abs-1802-07098, conficmlNorouzi-FardTMZ18, conficml0001MZLK19, alaluf2020optimal}, robustness~\cite{orlin2016robust, DBLP:conf/icml/BogunovicMSC17, pmlr-v70-mirzasoleiman17a, Mitrovic:2017:SRS:3294996.3295209, pmlr-v80-kazemi18a}, parallel computation (in the map-reduce model)~\cite{barbosa15power,DBLP:conf/stoc/MirrokniZ15, DBLP:conf/focs/BarbosaENW16}, and most recently with respect to adaptivity~\cite{DBLP:conf/stoc/BalkanskiS18,balkanski2019exponential,DBLP:conf/stoc/BalkanskiRS19,chen2019unconstrained,chekuri19submodular,ene2019submodular,ene2019matroid,fahrbach2019nonmonotone,DBLP:conf/soda/FahrbachMZ19}.
In each of these models, the central benchmark problem has been the basic cardinality-constrained problem studied in~\cite{nemhauser1978analysis}, namely that of finding a set $S \subseteq W$ of cardinality $k$ that maximizes $f(S)$, where $f$ is a non-negative and monotone submodular function.  
We refer to this problem as \maxcard.  

While tight algorithms are known for \maxcard  (and even for the more general problem where the cardinality constraint is replaced by a matroid) in the map-reduce model~\cite{DBLP:conf/focs/BarbosaENW16, DBLP:conf/soda/LiuV19} and the adaptive model~\cite{balkanski2019exponential,ene2019submodular,DBLP:conf/soda/FahrbachMZ19}, it has remained an open problem  to give tight results for the data stream and robust settings. 
In this paper, we resolve this question for data stream algorithms and make progress on the robust problem. 
These results are obtained by considering the \emph{one-way communication complexity} of \maxcard, the study of  which highlights several interesting aspects of submodular functions.
We first discuss our results in the clean communication model, and then give more detail about the connection to data stream and robust algorithms. 

\subsection{One-way Communication Complexity of \texorpdfstring{\maxcard}{\maxcardpdf}}
We first study the one-way communication complexity of \maxcard in the presence of two players. 
An informal description of the model is as follows (see \cref{sec:prelim_model} for the formal definition).
The first player Alice has only access to a subset $V_A \subseteq W$ of the ground set  and the second player Bob has access to $V_B \subseteq W$ with $V_A\cap V_B=\varnothing$.  
In the first phase, Alice can query the value of a submodular objective function $f$ on any subset of her elements; and then, at the end of the phase, she can send an arbitrary message to Bob based on the information that she has. 
Then, in the second phase, Bob gets the message of Alice and the elements of $V_B$. He can then query $f$ on any subset of the elements, and his objective is to produce a subset of $V_A \cup V_B$ of size at most $k$ that approximately maximizes $f$ among all such subsets.

A trivial protocol that allows Bob to always  output the  optimal solution is for Alice to send all the elements in $V_A$; and for Bob to then output $\arg \max_{S \subseteq V_A \cup V_B: |S| \leq k} f(S)$. 
While this protocol has an optimal approximation guarantee of $1$, it has a very large communication complexity since it requires Alice to send all the elements of $V_A$, which may be as many as $N  = |W|$ elements.

A protocol of lower communication complexity is for Alice to calculate her ``optimal'' solution $S_A = \arg\max_{S \subseteq V_A: |S| \leq k} f(S)$ and send only those (at most $k$ many) elements in $S_A$ to Bob. 
Bob then outputs either his ``optimal'' solution $S_B = \arg\max_{S \subseteq V_B: |S| \leq k} f(S)$ or $S_A$, whichever set attains the larger value. 
It is not hard to see that this protocol has an approximation guarantee of $1/2$, i.e., that $\max(f(S_A), f(S_B))\geq 1/2 \cdot \max_{S \subseteq V_A \cup V_B: |S| \leq k} f(S)$  for any $V_A, V_B \subseteq W$.

The above examples indicate a natural trade-off between the amount of communication and the approximation guarantee, with the central question being to understand the optimal relationship between these two quantities. 
If one further restricts Alice to only  query the submodular function on sets of cardinality at most $k$, then the hardness result in~\cite{conficmlNorouzi-FardTMZ18} for streaming algorithms  implies that, in any (potentially randomized) protocol with an approximation guarantee of $(1/2+\varepsilon)$, Alice must send a message of length $\Omega\rb{\varepsilon N/k}$. 
In other words, under this restriction on Alice, the two basic protocols described above achieve the optimal trade-off up to lower order terms (in this, as in previous work, we think of $k \ll N$).

Restricting Alice to evaluate $f$ only on sets of cardinality at most $k$ may appear like a mere technical assumption to make the arguments in~\cite{conficmlNorouzi-FardTMZ18} work; especially since, to the best of our knowledge, there are no known examples where querying the submodular function on infeasible sets leads to provably better results. 
Perhaps surprisingly, we prove this intuition wrong and give a protocol that crucially exploits the possibility to query the value of infeasible sets.

\begin{theorem} \label{thm:two_player_protocol_simplified}
    There exists a two-player protocol for \maxcard with an approximation guarantee of $2/3$ in which Alice sends a message consisting of $O(k^2)$ elements.
\end{theorem}
We present this protocol in \cref{ssc:two_players_algs}, where we also show that we can further reduce the message size of Alice down to $O(k \log (k)/\varepsilon)$ elements while still obtaining an approximation guarantee of $2/3-\varepsilon$.  
By allowing Alice to query $f$ on sets of cardinality larger than $k$, we can thus improve the approximation guarantee of $1/2$ to $2/3$ while still maintaining an (almost) linear-sized message in $k$.
In \cref{sec:two-player-sub-hardness}, we further show that the guarantee of $2/3$ is tight in the following strong sense. In any protocol that achieves a better guarantee, Alice must send a message of roughly the same size as the trivial protocol mentioned above that achieves an approximation guarantee of $1$. 
\begin{theorem} \label{thm:two_player_hardness_simplified}
    In any (potentially randomized) two-player protocol for \maxcard that has an approximation guarantee of $2/3+\varepsilon$ for $\varepsilon > 0$, Alice sends a message of length at least $\Omega\rb{\varepsilon N/k}$.
\end{theorem}
%\ola{I don't know if we should mention the polytime somewhere here or directly under the discussion of robustness applications}
%
%In our most technically involved result\footnote{\ola{not sure if polytime is not more involved.}}
The fact that we can beat the approximation guarantee of $1/2$ using little  communication in the presence of two players, gives hope that a similar result may hold for many players, and more generally in the streaming model, which can be thought of as having one player per element. 
%We now turn our attention to the more general $p$ player setting with $p\geq 2$. 
The definition of the $p$-player setting, with $p \geq 2$, is the natural generalization of the two-players setting.
Informally (again, see~\cref{sec:prelim_model} for the formal definition), the $i$-th player receives a private subset $V_i \subseteq W$ of the ground set and, upon reception of a message from the previous player, she computes and sends a message to the following player. Finally, the last player's task is to  output a subset of $V_1 \cup V_2 \cup \cdots \cup V_p$  of cardinality at most $k$ that approximately maximizes $f$ among all such subsets.  

One can observe that the  streaming algorithm of~\cite{conficml0001MZLK19}  yields, for any integer $p\geq 2$, a $p$-player protocol that has an approximation guarantee of $(1/2 - \varepsilon)$ and where each player sends a message consisting of at most $O(k/\varepsilon)$ elements.
Moreover, the obtained protocol only queries $f$ on sets of size at most $k$ and is, thus, tight with respect to such protocols (even in the two-player setting).
Similar to the two-player case, the naturally arising question is whether protocols querying $f$ on sets of cardinality larger than $k$ can improve over this approximation guarantee.
Our most technical result shows that this is not the case as $p$ (and $k$)  tends to infinity.
\begin{theorem} \label{thm:many_player_hardness_simplified}
    \label{thm:intro_gen_hardness}
    For every $\varepsilon > 0$, there is an integer $p_0\geq 2$  such that the following holds for any (potentially randomized) $p$-player protocol for \maxcard with $k=p \geq p_0$. If the protocol has an approximation guarantee of $1/2 + \varepsilon$, then one of the players sends a message of length at least $\Omega\rb{\varepsilon N/p^3}$.
\end{theorem}
The proof of the above theorem is given in \cref{sec:gen_hardness}. It is  based on a new construction of a family of coverage functions that hides the optimal solution while guaranteeing that no solution that does not contain elements from the optimal solution can provide an approximation significantly better than $1/2$.
This result immediately implies a tight hardness result in the streaming model that we explain in the next section.

\subsection{Applications to Data Stream and Robustness}
One key reason for the success of communication complexity is that results for the models it motivates, which are on their own right interesting models capturing the essence of trade-offs involving message sizes, are often widely applicable to other models of computation. 
This is also the case for submodular functions. As we show below, our results yield both  new hardness  and algorithmic results in the context of data streams and robustness.
%In particular, our results give both hardness results and new  tight hardness results for  streaming algorithm  \cref{thm:intro_gen_hardness} gives tight hardness results 

\paragraph{Data stream algorithms.} We first discuss the well-known and direct connection to data stream algorithms.  
In the data stream model, the elements of the (unknown) ground set  arrive one element at a time, rather than being available all at once, and the algorithm is restricted to only use a small amount of memory. A semi-streaming algorithm is an algorithm for this model whose memory size has only a nearly-linear dependence on the parameter $k$ (the output size) and at most a logarithmic dependence on the size   of the ground set.
The goal is to output, at the end of the stream, a subset of the elements in the stream of cardinality at most $k$ that approximately maximizes $f$ among all such subsets. 

The first result in this setting was given by Chakrabarti and Kale~\cite{chakrabarti2014submodular}, who described a semi-streaming algorithm for \maxcard with an approximation guarantee of $1/4$.  Badanidiyuru et al.~\cite{badanidiyuru2014streaming} proposed later a different semi-streaming algorithm which provides a better approximation ratio of $(1/2 - \varepsilon)$ and maintains at most $O\rb{k \log (k)/\varepsilon}$ elements in  memory. 
This memory footprint was recently improved by~\cite{conficml0001MZLK19}, who obtained the same approximation guarantee while only maintaining at most  $O\rb{k/\varepsilon}$ elements in memory.

The two last algorithms share the approximation guarantee of $1/2-\varepsilon$. This was improved to $1-1/e - \varepsilon$ by Agrawal et al.~\cite{journals/corr/abs-1809-05082}, but only under the assumption that the elements of the stream arrive in a uniformly random order. In contrast, Huang et al. \cite{huang2020approximability} showed that without this assumption one cannot obtain an approximation ratio better than $2 - \sqrt{2} \approx 0.586$ (improving over a previous inapproximability result of $1-1/e$ due to~\cite{mcgregor2019better}). Hence, prior to the current work, it remained an open question whether the approximation guarantee of $1/2-\varepsilon$ obtained by the state-of-the-art algorithms is optimal.\footnote{We recall that the hardness result of~\cite{conficmlNorouzi-FardTMZ18} only applies to the restricted case when the value of the submodular function is queried on sets of cardinality at most $k$.} However, a direct consequence of \cref{thm:intro_gen_hardness} settles this question. Specifically, any algorithm that achieves a better approximation guarantee than $1/2$ must (up to a polynomial factor in $k$) essentially store  all the elements of the stream.

\begin{restatable}{theorem}{thmstreaminghardness}
\label{thm:streaming-hardness}
    For any $\varepsilon >0$,
    a data stream algorithm for \maxcard with an approximation guarantee of $1/2 + \varepsilon$ must use memory $\Omega\rb{\varepsilon s / k^3}$, where $s$ denotes the number of elements in the stream.
\end{restatable}
The proof of the above theorem is almost immediate given \cref{thm:intro_gen_hardness} and the well-known connection between data stream algorithms and one-way communication. Thus, it is deferred to \cref{app:streaming-hardness}.

\paragraph{Robust submodular function maximization.}
The work on algorithms for \maxcard has been partially motivated by the desire to extract small summaries of huge data sets. In many settings, the extracted summary is also required to be robust. That is, the quality of the summary should degrade by as little as possible when some elements of the ground set are removed. Such removals may arise for many reasons, such as failures of nodes in a network, or user preferences which the model failed to account for; they could even be adversarial in nature. Recently, this topic has attracted special attention due to its importance in privacy and fairness constraints. The robust summaries enable us to remove sensitive data without incurring much loss in performance, giving us the ability to protect personal information (the right to be forgotten) and avoid biases (e.g., gender, measurement, and design biases).

The first attempts to design algorithms that generate robust summaries assumed that the summary is simply a set of size $k$, and the algorithm should guarantee that the value of this set is competitive against the best possible such set even when some elements are deleted (from both the ground set and the solution set). Naturally, this objective makes sense only when the number $d$ of deleted elements is significantly smaller than $k$. Accordingly, \cite{orlin2016robust} provided the first constant ($0.387$) factor approximation result to this problem for $d = o( \sqrt{k} )$, and Bogunovic et al.~\cite{DBLP:conf/icml/BogunovicMSC17} improved the restriction on number of deletions to $d = o(k)$ while keeping the approximation guarantee unchanged.

More recent works studied a more general variant of the above problem where an algorithm consists of two procedures: a summary procedure and a query procedure.
The summary procedure first generates a summary $M\subseteq W$ of the ground set $W$ with few, but typically more than $k$ elements, without knowing the elements $D\subseteq W$ to be deleted; after this, the set $D$ is revealed, and the query procedure returns a solution set $S_D\subseteq M\setminus D$ with $|S_D|\leq k$. The goal is for the final output set $S_D$ to be competitive against the best subset of size $k$ in the ground set without $D$, for any (worst-case) choice of $D$. More formally, such a robust algorithm is said to have an approximation guarantee of $\alpha$ if
\begin{align*}
   \expected{f(S_D)} \geq \alpha \cdot \max_{Z \subseteq W\setminus D, |Z| \leq k} \mspace{-18mu} f(Z)
   \qquad \forall D\subseteq W \text{ with } |D|\leq d\enspace.
\end{align*}
This problem is usually referred to as \emph{robust submodular maximization}.

The state-of-the-art result for robust submodular maximization is a $(1/2-\eps)$-approximation algorithm due to Kazemi et al.~\cite{pmlr-v80-kazemi18a}, whose summaries contain $O(k + d \log k/\eps^2)$ elements. This result improved over previous results by Mirzasoleiman et al.~\cite{pmlr-v70-mirzasoleiman17a} and Mitrovic et al.~\cite{Mitrovic:2017:SRS:3294996.3295209}. It should be noted that all these results enjoy a semi-streaming summary procedure, and by \cref{thm:streaming-hardness}, the approximation ratio of $1/2 - \eps$ guaranteed by some of them is basically the best possible as long as the summary procedure remains a semi-streaming algorithm.

We present the first algorithms for robust submodular maximization whose approximation guarantee is better than $1/2$. We do this via the following theorem, which shows that one can convert most natural two-player protocols for \maxcard into algorithms for robust submodular maximization. The proof of this theorem is based on a technique of~\cite{pmlr-v70-mirzasoleiman17a}, and we defer both this proof and a fully formal statement of the theorem to \cref{app:robustness}.
\begin{theorem} \label{thm:robust_reduction_simplified}
Assume we are given a two-player protocol $\cP$ for \maxcard obeying some natural properties. Then, there exists an algorithm $\cA$ for robust submodular maximization such that
\begin{enumerate}[label=(\roman*),itemsep=0em,parsep=0em,topsep=0.2em]
	\item the approximation guarantee of $\cA$ is at least as good as the approximation guarantee of $\cP$;
	\item the number of elements in the summary of $\cA$ is larger than the communication complexity of $\cP$ (in elements) only by an $O(d)$ factor;
	\item if $\cP$ runs in polynomial time, then so is $\cA$.
\end{enumerate}
\end{theorem}

As all the protocols we use to prove our results in this paper obey the natural properties required by \cref{thm:robust_reduction_simplified}, one can combine this theorem with \cref{thm:two_player_protocol_simplified} to get the following corollary.
\begin{corollary} \label{cor:robust_exponential}
There exists an algorithm for robust submodular maximization returning a $2/3$-approximate solution and using summaries of $O(dk^2)$ elements.
\end{corollary}
Analogous to \cref{thm:two_player_protocol_simplified}, one can reduce the summaries to $O(dk \log(k) /\varepsilon)$ many elements while guaranteeing an approximation factor of $2/3-\varepsilon$. 
%sacrificing a $1-\varepsilon$ to the approximation factor (see~\cref{thm:uppertwoplayersub}).
%
Unfortunately, the protocol used to prove \cref{thm:two_player_protocol_simplified} uses exponential time, and thus, \cref{cor:robust_exponential} is mostly of theoretical value. Nevertheless, we show that even when requiring efficient procedures, the factor of $1/2$ can be beaten, while only using linear message size.
\begin{theorem} \label{thm:two_player_polynomial_simplified}
There exists a polynomial time two-player protocol for Max-Card-k with an approximation guarantee of $0.514$ in which Alice sends a message consisting of $O(k)$ elements.
\end{theorem}
The last theorem is proved in Section~\ref{sec:poly2players}. Combining this theorem with \cref{thm:robust_reduction} yields the following result for robust submodular maximization.
\begin{corollary}
There exists a polynomial time $0.514$-approximation algorithm for robust submodular maximization using summaries of $O(dk)$ elements.
\end{corollary}

%% file: 200-preliminaries.tex
\section{Formal Statement of Model and Results}
\label{sec:prelim_model}

In this section we formally present the model that we assume in this paper, and restate in a formal way the results that we prove for this model. We begin by discussing the model for the two-player setting. It is natural to formulate a simple model for this setting in which Alice forwards some elements to Bob, and then Bob can access only these elements and the elements he receives directly. All the protocols we present fit into this simple model. However, one could imagine more involved protocols in which Alice passes coded information about the elements she received, rather than simply forwarding a subset of these elements. To make our impossibility results apply also to protocols of this kind, we formulate below a somewhat more involved model in which the message sent from Alice to Bob is an arbitrary string of bits. We note that there is no unique ``right'' way to cast the problem we consider into a model, and one can think of multiple natural ways to do so, each corresponding to a different intuitive viewpoint. Fortunately, it seems that our results are mostly independent of the particular formulation used (up to minor changes in the exact bounds), and thus, we chose a model that we believe is both intuitive and allows for a nice presentation of the results. Nevertheless, for completeness, we present in \cref{app:alternative_model} a sketch of an alternative model that we also found attractive.

An instance of our model consists both of global information known upfront to both Alice and Bob, and private information that is available only to either Alice or Bob. The global information includes the upper bound $k$ on the size of the solution (which is a positive integer), a ground set $W$ of elements and a partition of $W$ into two disjoints sets $W_A$ and $W_B$. One should think of the sets $W_A$ and $W_B$ as all elements that Alice and Bob, respectively, could potentially get. We denote by $V_A \subseteq W_A$ the set of elements that Alice actually gets, and by $V_B \subseteq W_B$ the set of elements that Bob actually gets. Both these sets are private information available only to their respective players. Finally, the instance also includes a non-negative monotone submodular function $f\colon 2^W \to \nnR$ defined over all the subsets of $W$. Alice has access to this function through an oracle that can evaluate $f$ on any set $S \subseteq W_A$ (in other words, given such a set $S$, the oracle returns $f(S)$). Bob, in contrast, has access to $f$ through a more powerful oracle that can evaluate $f$ on any subset of $W$. Intuitively, the reason for the difference between the powers of the oracles is that Alice only needs to evaluate sets consisting of elements that she might get, while Bob must also be able to evaluate $f$ on subsets that include elements sent by Alice (nevertheless, one can observe that the oracle of Bob does not leak information about the elements of $W_A$ that Alice actually got, i.e., the elements that ended up in $V_A$). The objective of Alice and Bob is to find a set $S \subseteq V_A \cupdot V_B$ maximizing $f$ among all such sets of size at most $k$.

A \emph{communication protocol} $\cP = (\cA_A, \cA_B)$ for this model consists of two (possibly randomized) algorithms for Alice and Bob. The protocol proceeds in two phases. In the first phase, the algorithm $\cA_A$ of Alice computes a message $m$ for Bob based on the global information and the private information available to Alice. Then, in the second phase, the algorithm $\cA_B$ of Bob computes an output set based on
\begin{enumerate*}[label=(\roman*)]
\item the global information,
\item the private information available to Bob, and
\item the message $m$ received from Alice.
\end{enumerate*}
Formally, the \emph{communication complexity} of protocol $\cP$ is the maximum length in bits of the message $m$, where the maximum is taken over all the possible inputs and the randomness of the algorithms. However, since the message $m$ in our protocols consists mostly of elements that Alice sends to Bob, we state the communication complexity of these protocols, for simplicity, in elements instead of bits. The real communication complexity of these protocols in bits is larger than the stated bound in elements, but only by a logarithmic factor.

We can now restate our results for the two-player model in a more formal way. Note that the first of these theorems uses the $\widetilde{O}$ notation, which suppresses poly-logarithmic terms, and the second of these theorems refers by $N$ to the size of the ground set $W$.

\let\oldthetheorem\thetheorem
\renewcommand{\thetheorem}{\getrefnumber{thm:two_player_protocol_simplified}}
\begin{restatable}{theorem}{thmuppertwoplayersub} \label{thm:uppertwoplayersub}
For every $\varepsilon > 0$, there exists a two-player protocol for \maxcard with an approximation guarantee of  $(\nicefrac{2}{3} - \varepsilon)$ whose communication complexity is $\widetilde{O}(k/\varepsilon)$ elements. Moreover, there exists such a protocol achieving an approximation guarantee of $\nicefrac{2}{3}$ whose communication complexity is $O(k^2)$ elements.
\end{restatable}

\renewcommand{\thetheorem}{\getrefnumber{thm:two_player_hardness_simplified}}
\begin{restatable}{theorem}{thmlowertwoplayersub} \label{thm:lowertwoplayersub}
For every $\varepsilon \in (0, 1/4)$, any two-player (randomized) protocol with an approximation guarantee of $(\nicefrac{2}{3} + \varepsilon)$ must have a communication complexity of $\Omega(\frac{N\varepsilon}{k})$ bits in the regime $k \geq \varepsilon^{-1}$.
\end{restatable}

\renewcommand{\thetheorem}{\getrefnumber{thm:two_player_polynomial_simplified}}
\begin{restatable}{theorem}{thmtwoplayerpolynomial} \label{thm:twoplayerpolynomial}
There exists a two-player protocol for \maxcard with an approximation guarantee of $0.514$ whose communication complexity is $O(k)$ elements, and furthermore, both algorithms in this protocol run in polynomial time.
\end{restatable}

Let us now explain how the above model can be generalized to the $p$-player setting for $p \geq 2$. In this setting, the ground set $W$ is partitioned into $p$ disjoint sets $W_1, W_2, \dotsc, W_p$, rather than just two; and the global information available to all the players is again
\begin{enumerate*}[label=(\roman*)]
\item the upper bound $k$ on the size of the solutions,
\item the ground set $W$, and
\item the partition of this ground set.
\end{enumerate*}
Every player also has private information. In particular, the private information available to player $i \in [p]$ (recall that $[p]$ is a shorthand for the set $\{1, 2, ..., p\}$) is a subset $V_i \subseteq W_i$ and an oracle that can evaluate the objective function $f$ on every subset of $\bigcup_{j = 1}^i W_i$. The objective of the players is to find a set $S \subseteq \bigcup_{i = 1}^p V_i$ maximizing $f$ among all such sets of size at most $k$.

A communication protocol $\cP = (A_1, A_2, \dotsc, A_p)$ for this $p$-player model consists of $p$ (possibly randomized) algorithms for the $p$ players. The protocol proceeds in $p$ phases. In the first phase the algorithm $A_1$ of the first player computes a message $m_1$ based on the global information and the private information available to this player. The next $p - 2$ phases are devoted to players $2$ up to $p - 1$. In particular, in phase $i \in \{2, 3, \dotsc, p - 1\}$, the algorithm $A_i$ of
player $i$ computes a message $m_i$ based on the global information, the private information available to this player, and the message $m_{i - 1}$ produced by the previous player. Finally, in the last phase, the algorithm $A_p$ of the last player computes an output set based on the global information, the private information available to this player, and the message $m_{p - 1}$ produced by the penultimate player. The communication complexity of the protocol $\cP$ is the maximum length in bits of any one of the messages $m_1, m_2, \dots, m_{p - 1}$, where, like in the two-player model, the maximum is taken over all the possible inputs and the randomness of the algorithms.

We can now restate our result for the $p$-player model in a more formal way. Recall that $N = |W|$, and for $p\in \mathbb{Z}_{\geq 0}$, let $H_p = 1 + \frac{1}{2} + \frac{1}{3} + \ldots + \frac{1}{p}$ be the $p$-th harmonic number.
\renewcommand{\thetheorem}{\getrefnumber{thm:many_player_hardness_simplified}}
\begin{restatable}{theorem}{thmgenhardness} \label{thm:gen_hardness}
For every $\varepsilon > 0$, any $p$-player (randomized) protocol  for \maxcard with an approximation guarantee of
\begin{align*}
    \frac{p+ (H_p)^2}{2p - H_p} \cdot (1+\varepsilon)
\end{align*}
must have a communication complexity of $\Omega\left(\frac{N\varepsilon}{p^3}\right)$. Furthermore, this is true even in the special case in which the objective function $f$ is a coverage function and $k=p$.
\end{restatable}
\let\thetheorem\oldthetheorem
\addtocounter{theorem}{-4}

\section{Preliminaries: The INDEX and \texorpdfstring{\chainP}{CHAINp} Problem} \label{sec:index-model}

The impossibility results that we prove in this paper are based on reductions from problems which are known to require high communication complexity. The first of these problems is the well-known INDEX problem. In this two-player problem, Alice gets a string $x \in \{0, 1\}^n$ of $n$ bits, and can then send a message to Bob. Bob gets the message of Alice and an index $t\in [n]$, and based on these two pieces of information alone should output the value of $x_t$. Clearly, Bob can produce the correct answer with probability $1/2$ by outputting a random bit. However, it is known that Bob cannot guarantee any larger constant probability of success, unless the message he gets from Alice is of linear (in $n$) size (see, e.g.,~\cite{bar2002information, jayram2008one}).

The second problem we reduce from is \chainPn, a multi-player generalization of INDEX recently introduced by Cormode et al.~\cite{cormode_2019_independent}, which is closely related to the Pointer Jumping problem (see~\cite{chakrabarti_2007_lower}). In \chainPn, the index $p$ indicates the number of players and $n$ is a parameter that regulates the size of the bit string given to each player.\footnote{In~\cite{cormode_2019_independent}, the problem was simply named \chainP, keeping the parameter $n$ implicit.} The definition is as follows.
There are $p$ players $P_1, P_2, \dotsc, P_p$. For every $i \in [p - 1]$, player $P_i$ has as input a bit string $x^{i} \in \{0,1\}^n$ of length $n$, and, for every $i \in \{2, 3, \dotsc, p\}$, player $P_i$ (also) has as input an index $t^i\in \{1, 2, \dotsc, n\}$ (note that the convention in this terminology is that the superscript of a string/index indicates the player receiving it). Furthermore, it is promised that either $x^{i}_{t^{i + 1}} = 0$ for all $i \in [p-1]$ or $x^{i}_{t^{i + 1}} = 1$ for all these $i$ values. We refer to these cases as the $0$-case and $1$-case, respectively. The objective of the players in \chainPn is to decide whether the input instance belongs to the $0$-case or the $1$-case.

In~\chainPn, we are interested in the communication complexity of a one-way protocol that guarantees a success probability of at least $2/3$. Such a protocol $\cP = (\cA_1, \cA_2, \dotsc, \cA_p)$ consists of $p$ (possibly randomized) algorithms corresponding to the $p$ players. The protocol proceeds in $p$ phases. In phase $i \in [p - 1]$, the algorithm $\cA_i$ of player $i$ computes a message $m_i$ based on the input of this player and the message $m_{i - 1}$ computed by $\cA_{i - 1}$ in the previous phase (unless $i = 1$, in which case the computation done by $\cA_1$ depends only on the input of player $1$). In the last phase, algorithm $\cA_p$ of player $p$ decides between the $0$-case and the $1$-case based on the input of player $p$ and the message $m_{p - 1}$. The communication complexity of the protocol is defined as the maximum size (in bits) of any one of the messages $m_1, m_2, \dotsc, m_{p - 1}$, where the maximum is taken over all the possible inputs and the randomness of the protocol's algorithms. Furthermore, the success probability of the protocol is the probability that the case indicated by $A_p$ matches the real case of the input instance.

Note that \chainPn is indeed a generalization of the INDEX problem since the last problem is equivalent to \chain{$2$}{$n$}. In~\cite{cormode_2019_independent}, the following communication complexity lower bound was shown for \chainPn.

\begin{theorem}[\cite{cormode_2019_independent}]
Any protocol for \chainPn with success probability of at least $2/3$ must communicate at least $\Omega(n/p^2)$ bits in total.
\end{theorem}
Moreover, the following stronger result, for a restricted range of $p$, was announced in~\cite{cormode_2019_independent} without proof.
\begin{theorem}[\cite{cormode_2019_independent}]
There is a constant $C>0$ such that any protocol for \chainPn, where $p\leq C\cdot (\frac{n}{\log n})^{1/4}$, with success probability of at least $2/3$ must communicate at least $\Omega(n/p)$ bits in total.
\end{theorem}
We highlight that the above lower bounds are both for the \emph{total} number of bits communicated and not the maximum message size. Because there are $p$ messages, this immediately translates to lower bounds on the maximum message size of $\Omega(n/p^3)$ and $\Omega(n/p^2)$, respectively. 
For completeness, we show in~\cref{app:pindex} how proofs of standard results for the INDEX problem can get the following impossibility result for \chainPn, which provides a lower bound of $\Omega(n/p^2)$ on the maximum message size without restrictions on the range of $p$.
\begin{restatable}{theorem}{thmpindex}
    For any positive integers $n$ and $p\geq 2$, any (potentially randomized) protocol for \chainPn with success probability of at least $2/3$ must have a communication complexity of at least $n/(36 p^2)$.
    \label{thm:pindex_hardness}
\end{restatable}

%% file: 600-submodular-expo-time.tex
\section{Two Player Submodular Maximization} \label{sec:twoplayerexpo}
In this section we consider \maxcard in the two-player model, while ignoring the computational cost, i.e., we are only interested here in the relationship between the communication complexity and the approximation guarantee that can be obtained for this problem. Below, we restate the formal theorems that we prove in the section. The proofs of these theorems can be found in \cref{ssc:two_players_algs,sec:two-player-sub-hardness}, respectively. In a nutshell, the two theorems show together that an approximation guarantee of $\nicefrac{2}{3}$ is tight for the problem under a natural assumption on the communication complexity.
\thmuppertwoplayersub*

\thmlowertwoplayersub*

\input{610-submodular-expo-algo.tex}
\input{620-submodular-expo-hardness.tex}

%% file: 610-submodular-expo-algo.tex
\subsection{Algorithms for Two Players} \label{ssc:two_players_algs}

In this section we prove \cref{thm:uppertwoplayersub}. For that purpose, let us present \cref{alg:many_sizes}, which is a protocol for \maxcard in the two-player model that uses exponential computation. In this protocol, Alice finds for every $i \in \{0, 1, \dotsc, 2k\}$ the maximum value subset $S_i$ of $V_A$ of size at most $i$, and forwards all the sets she has found to Bob. Then, Bob finds the best solution over the elements that Alice has sent and $V_B$. 
\begin{protocol}
\caption{Repeated Solving with Varying Sizes} \label{alg:many_sizes}
\textbf{Alice's Algorithm}
\begin{algorithmic}[1]
\For{$i = 0$ to $2k$} 
    \State	Let $S_i$ be the set maximizing $f$ among all subsets of $V_A$ of size at most $i$.
\EndFor
\State Send all the sets $\{S_i\}_{i = 0}^{2k}$ as the message to Bob (we note that $S_0$ is always the empty set, so sending it is technically redundant. However, it allows us to treat all $i$ values in the same way in the analysis).
\end{algorithmic}
\textbf{Bob's Algorithm}
\begin{algorithmic}[1]
\State Let $\widehat{S}$ be the subset of $V_B \cup \bigcup_{i=0}^{2k} S_{i}$ of size at most $k$ maximizing $f$.\\
\Return{$\widehat{S}$}.
\end{algorithmic}
\end{protocol}

It is easy to see that \cref{alg:many_sizes} always outputs a feasible set; and moreover, the number of elements Alice sends to Bob is $O(k^2)$ because she sends $2k + 1$ sets of size at most $2k$ each. Thus, to prove that \cref{alg:many_sizes} obeys all the properties guaranteed by the second part of \cref{thm:uppertwoplayersub}, it remains to show that it produces a $\nicefrac{2}{3}$-approximation, which is our main objective in the rest of this section.

Let us denote by $\Oset$ a subset of $V_A \cup V_B$ of size at most $k$ maximizing $f$ among all such subsets, and let $\Oval = f(\Oset)$. Also, let $M = V_B \cup \bigcup_{i = 1}^{2k} S_{i}$ be the set of elements that Bob either receives from Alice or receives directly. Note that $M$ is also the set of elements in which Bob looks for $\widehat{S}$. Using this notation, we can now describe the intuitive idea behind our first observation.

Our analysis of \cref{alg:many_sizes} is based on two sets $S_{k - |\Oset \cap M|}$ and $S_{2(k - |\Oset \cap M|)}$. Observe that one candidate for $S_{k - |\Oset \cap M|}$ is the part of $\Oset$ that Alice got and did not forward to Bob. Thus, we know that $S_{k - |\Oset \cap M|}$ is as valuable as $\Oset \setminus M$. The following observation formalizes this fact.
\begin{observation} \label{obs:remaining_OPT_approximation}
$f(S_{k - |\Oset \cap M|}) \geq f(\Oset \setminus M)$.
\end{observation}
\begin{proof}
The set $\Oset \setminus M$ is a subset of $(V_A \cup V_B) \setminus M \subseteq V_A$ of size $|\Oset \setminus M| = |\Oset| - |\Oset \cap M| \leq k - |\Oset \cap M|$. Thus, the observation follows from the choice of $S_{k - |\Oset \cap M|}$ by \cref{alg:many_sizes}.
\end{proof}

Despite the fact that $S_{k - |\Oset \cap M|}$ is as valuable as $\Oset \setminus M$, it is not clear to what extent the values of the two sets ``overlap'' (more formally, how far is the value of their union from the sum of their individual values). If the overlap is large, then this means that $S_{k - |\Oset \cap M|}$ is a good replacement for $\Oset \setminus M$, and thus, Bob can construct a good solution by combining $S_{k - |\Oset \cap M|}$ with $\Oset \cap M$. In contrast, if the overlap between $S_{k - |\Oset \cap M|}$ and $\Oset \setminus M$ is small, then they can be combined into a single large value set, which guarantees a large value for $S_{2(k - |\Oset \cap M|)}$. Thus, there is a trade-off between the values of the sets $\hat{S}$ and $S_{2(k - |\Oset \cap M|)}$. \cref{lem:large_set_guarantee} formally captures this trade-off.

\begin{lemma} \label{lem:large_set_guarantee}
$f(S_{2(k - |\Oset \cap M|)}) \geq \Oval + f(S_{k - |\Oset \cap M|}) - f(\widehat{S})$.
\end{lemma}
\begin{proof}
To prove the lemma, we have to show that $V_A$ includes a set of size at most $2(k - |\Oset \cap M|) \leq 2k$ whose value is at least $\Oval + f(S_{k - |\Oset \cap M|}) - f(\widehat{S})$. In particular, we will show that the set $(\Oset \setminus M) \cup S_{k - |\Oset \cap M|}$ has these properties. (Note that this set is a subset of $V_A$ because $V_B \subseteq M$.) Clearly, the size of this set is at most $(|\Oset| - |\Oset \cap M|) + (k - |\Oset \cap M|) \leq 2(k - |\Oset \cap M|)$.

Our next goal is to lower bound the value of the above set. Towards this goal, we note that $(\Oset \cap M) \cup S_{k - |\Oset \cap M|}$ is a subset of $M$ of size at most $k$, and thus, $f(\widehat{S}) \geq f((\Oset \cap M) \cup S_{k - |\Oset \cap M|})$ by the definition of $\widehat{S}$. Using the last inequality, we get
\begin{align*}
	f((\Oset \setminus M) \cup S_{k - |\Oset \cap M|})
	\geq{} & 
	f((\Oset \setminus M) \cup S_{k - |\Oset \cap M|}) + f((\Oset \cap M) \cup S_{k - |\Oset \cap M|}) - f(\widehat{S})\\
	\geq{} &
	f(\Oset \cup S_{k - |\Oset \cap M|}) + f(S_{k - |\Oset \cap M|}) - f(\widehat{S})\\
	\geq{} &
	f(\Oset) + f(S_{k - |\Oset \cap M|}) - f(\widehat{S})\enspace,
\end{align*}
where the second inequality follows from the submodularity of $f$, and the last inequality from its monotonicity.
\end{proof}

If $f(\widehat{S})$ is large, then we are done. Otherwise, the previous lemma guarantees that $S_{2(k - |\Oset \cap M|)}$ is a very valuable set. While this set might be infeasible (unless $|\Oset \cap M| \geq k / 2$), its value can be exploited by adding half of this set to $\Oset \cap M$. The following lemma gives the lower bound on $f(\widehat{S})$ that can be obtained in this way.
\begin{lemma} \label{lem:secondary_relationship}
$2 f(\widehat{S}) \geq f(\Oset \cap M) + f(S_{2(k - |\Oset \cap M|)})$.
\end{lemma}
\begin{proof}
Let us define $S^1_{2(k - |\Oset \cap M|)}$ and $S^2_{2(k - |\Oset \cap M|)}$ as an arbitrary disjoint partition of the set $S_{2(k - |\Oset \cap M|)}$ into two subsets of size at most $k - |\Oset \cap M|$ each. Then, the submodularity of $f$ implies
\begin{align*}
	\sum_{h = 1}^2 f((\Oset \cap M) \cup S^h_{2(k - |\Oset \cap M|)})
	\geq{} &
	f(\Oset \cap M) + f((\Oset \cap M) \cup S_{2(k - |\Oset \cap M|)})\\
	\geq{} &
	f(\Oset \cap M) + f(S_{2(k - |\Oset \cap M|)}) \enspace,
\end{align*}
where the second inequality follows from the monotonicity of $f$.
The lemma now follows by the definition of $\widehat{S}$ and the observation that both $(\Oset \cap M) \cup S^1_{2(k - |\Oset \cap M|)}$ and $(\Oset \cap M) \cup S^2_{2(k - |\Oset \cap M|)}$ are subsets of $M$ of size at most $k$.
\end{proof}

We are now ready to prove the approximation guarantee of \cref{alg:many_sizes} (and thus, complete the proof of the second part of \cref{thm:uppertwoplayersub}).
\begin{corollary}
\cref{alg:many_sizes} is a $\nicefrac{2}{3}$-approximation protocol.
\end{corollary}
\begin{proof}
Combining \cref{lem:large_set_guarantee,lem:secondary_relationship}, we get
\begin{align*}
	2 \cdot f(\widehat{S})
	\geq{} &
	f(\Oset \cap M) + f(S_{2(k - |\Oset \cap M|)})\\
	\geq{} &
	f(\Oset \cap M) + \Oval + f(S_{k - |\Oset \cap M|}) - f(\widehat{S}) \enspace.
\end{align*}
Rearranging this inequality, and then plugging into it the lower bound on $f(S_{k - |\Oset \cap M|})$ given by \cref{obs:remaining_OPT_approximation}, yields
\[
	f(\widehat{S})
	\geq
	\frac{f(\Oset \cap M) + \Oval + f(\Oset \setminus M)}{3}
	\geq
	\frac{2}{3} \cdot \Oval
	\enspace,
\]
where the second inequality follows from the submodularity and non-negativity of $f$. Because $\widehat{S}$ is the output of \cref{alg:many_sizes}, this concludes the proof.
\end{proof}

To prove also the first part of \cref{thm:uppertwoplayersub}, we need to reduce the number of elements forwarded from Alice to Bob by \cref{alg:many_sizes}. This can be done by applying geometric grouping to the sizes of the sets in $\{S_i\}_{i=1}^{2k}$. More precisely, Alice only forwards the sets $S_i$ for either $i = 0$, $i= \lfloor (1+\eps)^j \rfloor$, or $i= 2\lfloor (1+\eps)^j \rfloor$ for some integer $ 0 \leq j \leq \log_{1+\eps} k$, where $\eps$ is the parameter from the theorem. This reduces the number of elements forwarded to $\widetilde{O}(k/\eps)$, and it is not difficult to argue that the above analysis of the approximation ratio still works after this reduction, but its guarantee becomes worse by a factor of $1-O(\eps)$. A formal proof of this can be found in \cref{app:geometric_grouping}.

%% file: 620-submodular-expo-hardness.tex
\subsection{Hardness of Approximation for Two Players}\label{sec:two-player-sub-hardness}

In this section we prove the impossibility result stated in \cref{thm:lowertwoplayersub}. We do that by using a reduction from a problem known as the INDEX problem, which is presented in  \cref{sec:prelim_model}. The same section also states an impossibility result for a generalization of this problem (\cref{thm:pindex_hardness}), which in the context of the INDEX problem implies that any protocol guaranteeing a success probability of at least $2/3$ for this problem must have a communication complexity of at least $n / 144$.

Our plan in this section is to assume the existence of a protocol named $\ALG$ for \maxcard in the two-players model with an approximation guarantee of $2/3 + \varepsilon$, and show that this leads to a protocol $\ALGI$ for the INDEX problem whose communication complexity depends on the communication complexity of $\ALG$. This allows us to translate the communication complexity lower bound for protocols for INDEX to a communication complexity lower bound for $\ALG$.

Before getting to the protocol $\ALGI$ mentioned above, let us first present a simpler protocol for the the INDEX problem, which is given as \cref{alg:twoplayersubreduction} and is used as a building block for $\ALGI$. \cref{alg:twoplayersubreduction} refers to $n$ possible objective functions that we denote by $f_1, f_2, \dotsc, f_n$. (Recall that $n$ is the length of the string that Alice receives in the INDEX problem.) To define these functions, we first need to define a set of $n$ other functions. Let $W' = \{w\} \cup \{v_i \mid i\in [n]\}$. For every $i\in [n]$, we define  $g_i\colon 2^{W'} \to \nnR$ as follows, where $S$ is an arbitrary subset of $W'$.
%\[
%    g_i(S)
%    =
%    \begin{cases}
%        \min\{2|S|/3, 1\} & \text{if $w \not \in S$} \enspace,\\
%        \min\{|S|/3, 1\} & \text{if $w \in S$ and $v_i \not \in S$} \enspace,\\
%        1 & \text{if $w, v_i \in S$}\enspace.
%    \end{cases}
%\]
%
\begin{equation*}
g_i(S) = \begin{cases}
\frac{1}{3} &\text{if } S=\{w\}\enspace,\\
1           &\text{if } S=\{w,v_i\}\enspace,\\
\min\left\{\frac{2}{3} |S\setminus \{w\}|, 1\right\} & \text{otherwise}\enspace.
\end{cases}
\end{equation*}
The multilinear extension of $g_i$ is the function $G_i\colon [0, 1]^{W'} \to \nnR$ defined by $G_i(y) = \expected{g_i(\RSet(y))}$, where $\RSet(y)$ is a random subset of $W'$ including every element $v \in W'$ with probability $y_v$, independently.\footnote{The multilinear extension of a set function was first introduced by~\cite{calinescu2011maximizing}.} In the context of $G_i$, given an element $v \in W'$, we occasionally use the notation $\characteristic_v$ to denote the characteristic vector of the singleton set $\{v\}$, i.e., the vector in $[0, 1]^{W'}$ containing $1$ in the $v$-coordinate and $0$ in all other coordinates.

Let us now define the ground set $W = W_A \cupdot W_B$, where $W_A = \{u_i^j \mid i\in [n] \text{ and } j\in [k-1]\}$ and $W_B = \{w\}$. Then, for every $i\in [n]$, the function $f_i\colon 2^W \to \nnR$ is defined as
\[
    f_i(S) = G_i(y^S) \quad \forall\; S \subseteq V \enspace,
\]
where the vector $y^S\in [0,1]^{W'}$ is defined by
\[
    y^S_{v_{i'}}
    =
    \frac{|S \cap \{u_{i'}^j \mid j\in [k-1]\}|}{k - 1}
    \qquad
    \forall\; i'\in [n]
\]
and
\[
    y^S_w
    =
    |\{w\} \cap S|
    \enspace.
\]

\begin{protocol}
\caption{Reduction from INDEX to \maxcard in the Two-Player Model} \label{alg:twoplayersubreduction}
\textbf{Alice's Algorithm}
\begin{algorithmic}[1]
    \State The set of elements Alice of $\ALG$ gets is $V_A = \{u_i^j \mid i\in [n] \text{ with }x_i = 1, \text{ and } j\in [k-1]\}$. Notice that this is indeed a subset of $W_A$, and it intuitively corresponds to the $1$-bits of the vector $x$ given to Alice in the INDEX problem.
    \State The objective function for $\ALG$ is one of the functions $f_1, f_2, \dotsc, f_n$. Since these functions are identical when restricted to $W_A$, the Alice part of $\ALG$ can execute without knowing which one of them is the real objective function.
    \State Send to Bob the same message sent by the Alice of $\ALG$.
\end{algorithmic}
\textbf{Bob's Algorithm}
\begin{algorithmic}[1]
    \State The set of elements Bob of $\ALG$ gets is $V_B = W_B = \{w\}$.
    \State The objective function for $\ALG$ can now be determined to be $f_t$, where $t$ is the index received by Bob.
    \State If $\ALG$ returns a set of value at most $\frac{2k}{3(k - 1)}$, output $x_t = 0$; otherwise, output $x_t = 1$.
\end{algorithmic}
\end{protocol}

We begin the analysis of \cref{alg:twoplayersubreduction} with the following lemma, which shows that the objective function this protocol passes to $\ALG$ has all the necessary properties. The proof of this lemma is simple and technical, and thus, we defer it to \cref{app:missing_proofs}. In a nutshell, it shows by a straightforward case analysis that $g_i$ is non-negative, monotone, and submodular, and then argues that the fact that $g_i$ has these properties implies that $f_i$ has them too. 
\begin{restatable}{lemma}{lemTwoSidesReductionProperties} \label{lem:TwoSidedReductionProperties}
For every $i\in [n]$, the functions $g_i$ and $f_i$ are non-negative, monotone, and submodular.
\end{restatable}

Our next step is analyzing the output distribution of \cref{alg:twoplayersubreduction}.

\begin{lemma} \label{lem:twosidesreduction_oneside}
If $x_t = 0$, where $t$ is the index received by Bob, then \cref{alg:twoplayersubreduction} always produces the correct answer. 
\end{lemma}
\begin{proof}
Let $S$ denote the output of $\ALG$. We need to show that $f_t(S) \leq \frac{2k}{3(k - 1)}$. There are two cases to consider. The first one is when $w \not \in S$. In this case,
\begin{align*}
    f_t(S)
    ={} &
    G_t(y^S)
    \leq
    g_t(\varnothing) + \sum_{v \in W'} [G_t(y^S_v \cdot \characteristic_v) - g_t(\varnothing)]\\
    ={} &
    g_t(\varnothing) + \sum_{v \in W'} y^S_v \cdot g_t(\{v\} \mid \varnothing)
    =
    \frac{2}{3} \cdot \sum_{v \in W'} y^S_v
    \leq
    \frac{2k}{3(k - 1)}
    \enspace,
\end{align*}
where the first inequality holds by submodularity of $g_t$, the second equality holds by the multilinearity of $G_t$, and the last inequality holds since the fact that $S$ contains up to $k$ elements guarantees that the sum of the coordinates of $y^S$ is at most $k / (k - 1)$.

The other case we need to consider is when $w \in S$. In this case,
\begin{align*}
    f_t(S)
    ={} &
    G_t(y^S)
    \leq
    g_t(\{w\}) + \sum_{v \in W'\setminus \{w\}} \mspace{-18mu} [G_t(y^S_v \cdot \characteristic_v + \characteristic_w) - g_t(\{w\})]\\
    ={} &
    g_t(\{w\}) + \sum_{v \in W'\setminus \{w\}} \mspace{-18mu} y^S_v \cdot g_t(\{v\} \mid \{w\})
    =
    \frac{1}{3} + \frac{1}{3} \cdot \sum_{v \in W'\setminus \{w\}} \mspace{-18mu} y^S_v
    \leq
    \frac{2}{3}
    \enspace,
\end{align*}
where the last inequality holds this time since the fact that $S$ contains up to $k - 1$ elements in addition to $w$ guarantees that the sum of all coordinates of $y^S$ except for the $w$-coordinate is at most $1$. To see why the third equality holds as well, note that the fact that $x_t = 0$ implies that none of the elements $u_t^1, u_t^2, \dotsc, u_t^{k - 1}$ belong to $V_A$, and thus, $y^S_{v_t} = 0$.
\end{proof}

\begin{lemma} \label{lem:twosidesreduction_secondside}
For $k \geq \varepsilon^{-1} \geq 4$, if $x_t = 1$, where $t$ is the index received by Bob, then \cref{alg:twoplayersubreduction} always produces the correct answer with probability at least $\varepsilon$.
\end{lemma}
\begin{proof}
We first observe that if $x_t=1$, then the maximum value that the submodular function $f_t$ achieves over subsets of $V_A \cup V_B$ of cardinality $k$ is $1$. Clearly, the function $g_t$ does not take values larger than $1$ for any set, and therefore the same holds for its multilinear extension $G_t$ and the function $f_t$ defined using this multilinear extension. Thus, it remains to show that there exists a set $S \subseteq V_A \cup V_B$ of size at most $k$ with $f_t(S) = 1$. Since $x_t = 1$,  all the elements $u_t^1, u_t^2, \dotsc, u_t^{k - 1}$ belong to $V_A$. Thus, the set $S = \{w\} \cup \{u_t^j \mid j\in [k-1]\}$ is a subset of $V_A \cup V_B$ of size $k$ whose value is
\[
    f_t(S)
    =
    G_t(y^S)
    =
    g_t(\{w, v_t\})
    =
    1
    \enspace.
\]

Let us define now $X$ to be a random variable corresponding to the value of the solution returned by $\ALG$. Since we assumed that $\ALG$ is a $(2/3 + \varepsilon)$-approximation algorithm, and we already proved that the highest value of a feasible set is $1$, we get $\E[X]\geq 2/3 + \varepsilon$. To complete the proof of the lemma, we have to show that the probability $\alpha = \Pr[X > \frac{2k}{3(k-1)}]$, which is the probability that \cref{alg:twoplayersubreduction} correctly returns $x_t=1$, is at least $\varepsilon$.
We upper bound $\E[X]$ through the following simple variation of Markov's inequality:
\begin{align*}
\frac{2}{3} + \varepsilon &\leq \E[X] \leq \Pr\left[X\leq \frac{2k}{3(k-1)}\right]\cdot \frac{2k}{3(k-1)} + \Pr\left[X > \frac{2k}{3(k-1)}\right]\cdot 1\\
 &= (1-\alpha)\cdot \frac{2k}{3(k-1)} + \alpha
 \leq (1-\alpha) \cdot \frac{2}{3(1-\varepsilon)} + \alpha\enspace,
\end{align*}
where in the second inequality we used the fact that $1$ is the largest value $X$ can take, and the last inequality follows from $k\geq \varepsilon^{-1}$. Rearranging the above inequality leads to
\begin{equation*}
(1-3\epsilon)\cdot \epsilon \leq (1-3\epsilon) \cdot \alpha\enspace,
\end{equation*}
which, by using $\epsilon^{-1}\geq 4$, implies $\alpha\geq \epsilon$, as desired.
\end{proof}

At this point we are ready to present the promised algorithm $\ALGI$, which simply executes $\lceil 2\varepsilon^{-1} \rceil$ parallel copies of \cref{alg:twoplayersubreduction}, and then outputs $x_t = 1$ if and only if at least one of the executions returned this answer.

\begin{corollary} \label{cor:twosided_reduction}
For $k \geq \varepsilon^{-1} \geq 4$, $\ALGI$ always answers correctly when $x_t = 0$, and answers correctly with probability at least $2/3$ when $x_t = 1$.
\end{corollary}
\begin{proof}
The first part of the corollary is a direct consequence of \cref{lem:twosidesreduction_oneside}. Additionally, by \cref{lem:twosidesreduction_secondside}, the probability that $\ALGI$ answers $x_t = 0$ when in fact $x_t = 1$ is at most
\[
    (1 - \varepsilon)^{\lceil 2\varepsilon^{-1} \rceil}
    \leq
    (1 - \varepsilon)^{2\varepsilon^{-1}}
    \leq
    e^{-\varepsilon \cdot 2\varepsilon^{-1}}
    =
    e^{-2}
    <
    \frac{1}{3}
    \enspace.
    \qedhere
\]
\end{proof}

Using the last corollary, we can now complete the proof of \cref{thm:lowertwoplayersub}.
\begin{proof}[Proof of \cref{thm:lowertwoplayersub}]
Since \cref{cor:twosided_reduction} shows that $\ALGI$ is an algorithm for the INDEX problem that succeeds with probability at least $2/3$, \cref{thm:pindex_hardness} guarantees that its communication complexity is at least $n/144$. Observe now that the message of $\ALGI$ consists of $\lceil 2 \varepsilon^{-1} \rceil$ messages of \cref{alg:twoplayersubreduction}, and thus, the communication complexity of \cref{alg:twoplayersubreduction} must be of size at least
\[
    \frac{n/144}{\lceil 2 \varepsilon^{-1} \rceil}
    \geq
    \frac{n/144}{3 \varepsilon^{-1}}
    =
    \frac{n\varepsilon}{432}
    \enspace.
\]
We now recall that the message of \cref{alg:twoplayersubreduction} is simply the message generated by $\ALG$ given the instance of \maxcard generated for it by \cref{alg:twoplayersubreduction}. Since this instance has a ground set of size $N = |W| = 1 + n(k - 1)$, the communication complexity of $\ALG$ must be at least
\[
    \frac{n\varepsilon}{432}
    =
    \frac{\varepsilon(N - 1)}{432(k - 1)}
    =
    \Omega\left(\frac{\varepsilon N}{k}\right)
    \enspace.
    \qedhere
\]
\end{proof}

%% file: 400-general-hardness.tex
\section{Hardness for Many Players}
\label{sec:gen_hardness}

In this section we prove that, in the case of many players, any protocol with reasonable communication complexity has an approximation guarantee upper bounded by an expression that tends to $1/2$ as the number of players tends to infinity. 
Specifically, we show the following (where, for $p\in \mathbb{Z}_{\geq 0}$, $H_p = 1 + \frac{1}{2} + \frac{1}{3} + \ldots + \frac{1}{p}$ is the $p$-th harmonic number).

\thmgenhardness*

We highlight that a (weighted) coverage function $f\colon 2^V\to \mathbb{R}_{\geq 0}$ is defined as follows. There is a finite universe $U$ with non-negative weights $a\colon U\to\mathbb{R}_{\geq 0}$, and $V\subseteq 2^U$ is a family of subsets of $U$. Then, for any $S\subseteq V$, we have
$f(S) = \sum_{u\in \cup_{v\in S} v} a(u)$%
.\footnote{In some texts, the term \emph{coverage function} is used for its unweighted version, i.e., $a(u)=1$ for $u\in U$. Our statements and proofs are described in terms of weighted coverage functions. However, this is merely a matter of convenience because any weighted coverage function can be approximated arbitrarily well through a scaled version of an unweighted one.} We also remark that our hardness construction applies to the related maximum set coverage problem. In that problem the stream consists of  $N$ sets $S_1, S_2, \ldots, S_N$ of some universe $U$ and each $S_i$ is encoded as the list of elements in that set. In other words, the submodular function $f$ is given explicitly by the sets of the underlying universe. Prior work showed that, even in this setting, any streaming algorithm with a better approximation guarantee than $(1-1/e)$ requires memory $\Omega(N)$~\cite{mcgregor2019better}. Our techniques also apply to this setting,\footnote{To see that this is the case, it is sufficient to observe that the intuitive description of the family $\cF$ in~\cref{sec:gen_hardness_intuition} is equivalent to the formal definition in~\cref{sec:hard_gen_family} when the underlying universe $U$ is of infinite size; and, from that point of view, it is clear that the algorithm receives no advantage if given the explicit representation of the sets compared to having an oracle access to the coverage function. Furthermore, by standard Chernoff concentration inequalities (see, e.g., the proof of Lemma~$8$ in~\cite{mcgregor2019better}), the family can be approximated up to any desired accuracy for feasible sets of cardinality at most $k$ by selecting $|U| = \Theta(k \log N)$.} 
 and hence we improve the hardness factor for the maximum set coverage problem   to the tight factor $1/2$.
 
The heart of the proof of the above theorem is the construction of a family  $\cF$ of submodular coverage functions on a common ground set $W$, partitioned into  sets $W_1, \ldots, W_p$, one for each player. 
All the sets $W_i$  have the same cardinality, which we denote by $n$, and thus, $N = |W| = n\cdot p$. 
The family $\cF$ contains a weighted coverage function $f_{o_1, \ldots, o_p}$ for  every $o_1\in W_1, o_2 \in W_2, \ldots, o_p\in W_p$.
%and so 
%\begin{align*}
%    \mathcal{F} =  \{f_{o_1, \ldots, o_p} :  o_1 \in V_1, \ldots, o_p \in V_p\}\enspace.
%\end{align*}
The intuition is that $\{o_1 , \ldots, o_p\}$ will be the ``hidden'' optimal solution for $f_{o_1, \ldots, o_p}$ when we set $k=p$. 

For the hardness result, there are two crucial properties that the construction should satisfy: 
\begin{itemize}
    \item \emph{Indistinguishability:} The $i$-th player should not be able to obtain any information about $o_i$ by querying the submodular function on subsets of $W_1 \cup W_2 \cup \cdots \cup W_i$. 
    \item \emph{Value gap:} The value of the solution $\{o_1 , \ldots, o_p\}$ is roughly twice the value of any solution of cardinality $k=p$ that does not contain any of these elements.
\end{itemize}
The first property intuitively ensures that the players must use much communication to identify the special elements $\{o_1, \ldots, o_p\}$; and the second property implies that, if they fail to do so, then the last player can only output a $1/2$-approximate solution.   
The following lemma formalizes these two properties that our family $\cF$ satisfies.
\begin{lemma}
    Let $W$ be partitioned into $p$ sets $W_1, \ldots, W_p$ of cardinality $n$. 
    There is a family $\cF = \{f_{o_1, o_2, \ldots, o_p} \mid o_1 \in W_1, o_2 \in W_2, \ldots, o_p\in W_p$\} of coverage functions on the ground set $W$ that satisfies:
    \begin{itemize}
        \item \emph{Indistinguishability:} For $i\in [p]$, any two functions $f_{o_1, \ldots, o_p}, f_{o'_1, \ldots, o'_p} \in \cF$ with $o_1 = o'_1, \ldots, o_{i-1} = o'_{i-1}$ are identical when restricted to the ground set $W_1 \cup \cdots \cup W_i$.
        \item \emph{Value gap:}  For any $f_{o_1, \ldots, o_p}\in \cF$, we have $f_{o_1, \ldots, o_p}(W) = f_{o_1, \ldots, o_p} (\{o_1,\ldots, o_p\}) \leq 2p$ and
        \begin{align*}
            \max_{\substack{\vspace*{0.5mm}\\S \subseteq W \setminus \{o_1, \ldots, o_p\}\\ |S| \leq p}} \mspace{-27mu} f_{o_1, \ldots, o_p}(S) \leq p + (H_p)^2 \leq \left(\frac{p+ (H_p)^2}{2p - H_p} \right) \cdot f_{o_1, \ldots, o_p} (\{o_1, \ldots, o_p\})\enspace.
        \end{align*}
    \end{itemize}
    \label{lemma:gen_hardness_F_properties}
\end{lemma}
Equipped with the above lemma, we prove \cref{thm:gen_hardness} in \cref{sec:gen_hardness_reduction} by a rather direct reduction from the \chainPn problem. We note that the reduction is similar to the one presented in \cref{sec:two-player-sub-hardness} for the two-player case. 

The core part of this section  is the construction of $\cF$ and the proof of \cref{lemma:gen_hardness_F_properties}. The outline is as follows. We first  give an intuitive description of the main ideas in \cref{sec:gen_hardness_intuition}. The family $\cF$ is then formally defined in \cref{sec:hard_gen_family}. Finally, the value gap and indistinguishability properties of \cref{lemma:gen_hardness_F_properties} are  proved in  \cref{sec:gen_hardness_value_gap} and \cref{sec:gen_hardness_indistinguishability}, respectively.

\input{410-intuition.tex}

\subsection{Construction of Family of Weighted Coverage Functions}
\label{sec:hard_gen_family}
We  formally describe the construction of the family $\cF$ of weighted coverage functions on the common ground set $W$. Recall that the ground set is partitioned into sets $W_1, \ldots, W_p$. Furthermore, each of these sets has cardinality $n$, and thus $N = |W| = n\cdot p$.

In the intuitive description (\cref{sec:gen_hardness_intuition}), we defined the functions in $\cF$ to be coverage functions, where the elements correspond to random subsets of the underlying universe $U$. 
Here we will be more precise and avoid this randomness. 
To this end,  we consider a slight generalization of weighted coverage functions that we call weighted fractional coverage functions. 
This is just done for convenience. In \cref{app:frac-cover}, we show that any such function is indeed a weighted coverage function.

Recall that in a weighted coverage function, every element is a subset of an underlying universe $U$ with non-negative weights $a\colon U \to \mathbb{R}_{\geq 0}$. 
%For notational simplicity, we will consider a slight generalization of this notion that we call \emph{fractional weighted coverage functions}. 
For \emph{fractional weighted coverage functions}, apart from the non-negative weights $a\colon U \to \mathbb{R}_{\geq 0}$, we also associate a function $p_v\colon U \rightarrow [0,1]$ with each element $v\in W$ with the intuition that $p_v(u)$ specifies 
 the ``probability'' that $v\in W$ covers $u\in U$. The value $f(S)$ of a subset $S \subseteq W$ of the elements is then defined by
\begin{equation}\label{eq:fracWeightedCoverageFunction}
   f(S) = \sum_{u\in U} a_u \cdot \Pr[\mbox{an element in $S$ covers $u$}] =  \sum_{u\in U} a_u \cdot \left(1 -  \prod_{v\in S} (1 - p_v(u))\right)\enspace. 
\end{equation}
A function $f\colon 2^W\to \mathbb{R}_{\geq 0}$ as defined above is what we call a \emph{weighted fractional coverage function}. Note that a weighted coverage function is simply the special case of  $\{p_v: v\in W\}$ taking binary values.
%Hence, the above generalization is just for convenience.

%

We are now ready to define our family $\cF$ of weighted fractional coverage functions, which as aforementioned is equivalent to weighted coverage functions, which in turn can be approximated by unweighted coverage functions to any desired accuracy. 

%As aforementioned, the ground set $V$  is partitioned into $p$ disjoint sets $V_1, V_2, \ldots, V_p$ (one set per player), where each  $V_j$ consists of $n$ elements. We refer to elements in $V_j$ by   $\{v^j_1, \ldots, v^j_n\}$. 

%Each function in our family $\cF$ is indexed by $p$ elements $o_1\in V_1, o_2 \in V_2, \ldots, o_p\in V_p$ and so 
%\begin{align*}
%    \mathcal{F} =  \{f_{o_1, \ldots, o_p} :  o_1 \in V_1, \ldots, o_p \in V_p\}\enspace.
%\end{align*}
%The intuition is that $\Oset = \{o_1 , \ldots, o_p\}$ will be an optimal solution to $f_{o_1, \ldots, o_p}$  and any solution of cardinality $k=p$ that does not contain any of these elements will have a significantly smaller value.
The underlying universe $U$ of the coverage functions in $\cF$ consists of $p$ points $U=\{u_1,\ldots, u_p\}$, where the weight $a_{u_j} \in \mathbb{R}_{\geq 0}$, for $j\in [p]$, will be fixed later in \cref{sec:gen_hardness_sel_a}. For notational convenience, we use the shorthand $a_j$ for $a_{u_j}$ and let $A_{\geq j}  = \sum_{i=j}^p a_i$. The family $\cF$ now contains a weighted fractional coverage function $f_{o_1, \ldots, o_p}$ for every $o_1 \in W_1, \ldots, o_p \in W_p$ that is defined as follows.
%is now  defined to be the following fractional coverage function, where $A_{\geq j}= \sum_{i=j}^p a_i$ for every $j \in [p]$.
\begin{itemize}
    \item Element $o_j$ covers  $\{u_j\}$, i.e., $p_{o_j}(u_j) = 1$ and $p_{o_j}(u) = 0$ for $u\in U\setminus \{u_j\}$.
    \item For every other element $v \in W_j \setminus \{o_j\}$, 
    \begin{align*}
        p_v(u) = \begin{cases}
        \frac{a_j}{A_{\geq j}} & \mbox{if $u \in \{u_j, u_{j+1},\ldots, u_{p}\}$\enspace,} \\
        0 & \mbox{otherwise\enspace.}
        \end{cases}
    \end{align*}
\end{itemize}
Note that, by interpreting the $p_v$ functions as probabilities, the above definition equals the intuitive description in \cref{sec:gen_hardness_intuition}: the ``hidden'' optimal elements $\{o_1, \ldots, o_p\}$ form a disjoint cover of the universe, and every other element in $W_j$ corresponds to a random subset of the (now weighted) universe disjoint from the subsets corresponding to $o_1, \ldots, o_{j-1}$.
Finally, by definition, for every $S \subseteq W$
\begin{align*}
 f_{o_1, \ldots, o_p}(S) = \sum_{j=1}^p a_j \left( 1- \prod_{v\in S} \left( 1- p_v(u_j)\right) \right)\enspace,
 \end{align*}
 which can be written as
 \begin{equation} \label{eq:gen_hardness_def_f}
%   \sum_{j=1}^p a_j \cdot  \left( 1 - \prod_{i=1}^j \left(1 - \frac{a_i}{A_{\geq i}} \right)^{|S\cap (W_i \setminus \{o_i\})|} + \mathbf{1}\{o_j \in S\}\cdot\prod_{i=1}^j \left(1 - \frac{a_i}{A_{\geq i}} \right)^{|S\cap (W_i \setminus \{o_i\})|}  \right)\enspace. 
%    
   f_{o_1,\ldots, o_p}(S) = \sum_{j=1}^p a_j \cdot  \left( 1 - \mathbbm{1}\{o_j\not\in S\} \prod_{i=1}^j \left(1 - \frac{a_i}{A_{\geq i}} \right)^{|S\cap (W_i \setminus \{o_i\})|} \right)\enspace, 
\end{equation}
where $\mathbbm{1}\{E\}$ indicates whether the event $E$ holds.

%%% \begin{align}
%%% \begin{aligned}
%%%  f_{o_1, \ldots, o_p}(S) &= \sum_{j=1}^p \sum_{u\in U_j} \left( 1- \prod_{e\in V} \left( 1- p_e(u)\right) \right) \\[5mm]
%%% %  & = \sum_{j=1}^p a_j \cdot \begin{cases}
%%% %     1 & \mbox{if $o_j \in S$,} \\
%%% %     1 - \prod_{i=1}^j \left(1 - \frac{a_i}{A_{\geq i}} \right)^{|S\cap (V_i \setminus \{e^i_{\alpha_i}\})|} &  \mbox{otherwise.}
%%% %     \end{cases} \\
%%%     & = \sum_{j=1}^p a_j \cdot  \left( 1 - \prod_{i=1}^j \left(1 - \frac{a_i}{A_{\geq i}} \right)^{|S\cap (V_i \setminus \{o_i\})|} + \mathbf{1}\{o_j \in V_j\}\cdot\prod_{i=1}^j \left(1 - \frac{a_i}{A_{\geq i}} \right)^{|S\cap (V_i \setminus \{o_i\})|}  \right) 
%%%     \end{aligned}
%%%     \label{eq:gen_hardness_def_f}
%%% \end{align}

\subsubsection[Selection of the Weights \texorpdfstring{$a_1,\ldots, a_p$}{a\_1, a\_2, ..., a\_p}]{Selection of the weights \boldmath{$a_1,\ldots, a_p$}}
\label{sec:gen_hardness_sel_a}
To complete the definition of our family $\cF$, it remains to define the weights $a_1, \ldots, a_p \in \mathbb{R}_{\geq 0}$ of the universe $U$.  
Recall from \cref{sec:gen_hardness_intuition} that we need to set these weights so that, if we let $f_{o_1, \ldots, o_p} \in \cF$ and $v_j \in W_j \setminus \{o_j\}$ for every $j\in [p]$, then
\begin{align*}
    f_{o_1, \ldots, o_p}( v_j \mid \{v_1, \ldots, v_{j-1}\} ) = 1 \qquad \mbox{for every $j\in [p]$.}
\end{align*}
This readily implies that $a_1 = 1$ and, more generally, by~\eqref{eq:gen_hardness_def_f}, one can see that this equals the condition
\begin{align*}
    a_j\prod_{i=1}^{j-1} \left(1 - \frac{a_i}{A_{\geq i}}\right) = 1 \qquad \mbox{for every $j \in [p]$,}
\end{align*}
where, here and later, we interpret the empty product as $1$.

The weights $a_1,\ldots, a_p$ satisfying this condition can be obtained as follows.
First, let $\delta_p = 1 $ and, for $i = p-1, p-2, \ldots, 1$, let $\delta_i$ be the largest solution\footnote{
It can be verified that 
%\begin{align*}
    $\delta_{i} = 1 +  \left(\frac{1+ \sqrt{1+ 4/\delta_{i+1}}}{2}\right)\cdot \delta_{i+1}$,
    but the exact value is not be important to us. 
    %\mbox{ for $i = p-1, p-2, \ldots, 1$}.
%\end{align*}
} of 
\begin{align*}
    \left( 1- \frac{1}{\delta_i}\right) \left(\delta_i - 1\right) = \delta_{i+1}\enspace.
\end{align*}
Now select the weights $a_1, a_2, \ldots, a_p$ to be
\begin{align*}
    a_j =   \prod_{i=1}^{j-1} \left( \frac{\delta_i-1}{\delta_{i+1}}\right) = \prod_{i=1}^{j-1} \frac{1}{1-1/\delta_i} \qquad \mbox{for $j=1, \ldots, p$\enspace.}
\end{align*}
In the next lemma we formally verify that these weights indeed satisfy condition~\eqref{eq:gen_hardness_become_one}. 
By basic calculations, we also show the identity~\eqref{eq:gen_hardness_nice_sum} and the inequalities~\eqref{eq:gen_hardness_bounds}. 
We remark that these are the only properties that we use about these weights in subsequent sections.
\begin{lemma}
    The weights $a_1, \ldots, a_p >0$  satisfy, for every $j=1, \ldots, p$,
  %  \begin{enumerate}
  %      \item $a_j\prod_{i=1}^{j-1} \left(1 - \frac{a_i}{A_{\geq i}}\right) = 1$
  %      \item $\frac{A_{\geq j}}{a_j} = \sum_{i=j}^p \left(2- \frac{a_i}{A_{\geq i}}\right)$
  %  \end{enumerate}
    \begin{align}
        a_j\prod_{i=1}^{j-1} \left(1 - \frac{a_i}{A_{\geq i}}\right)& = 1 \label{eq:gen_hardness_become_one}\enspace, \\
        \sum_{i=j}^p \left(2- \frac{a_i}{A_{\geq i}}\right)& = \frac{A_{\geq j}}{a_j} \label{eq:gen_hardness_nice_sum}\enspace,\text{ and} \\
        2j - H_{j} \leq \frac{A_{\geq p-j+1}}{a_{p-j+1}} &\leq 2j -1\enspace.\label{eq:gen_hardness_bounds}
    \end{align}
    \label{lem:gen_hardness_select_a}
\end{lemma}
\begin{proof}
We start by observing that
\begin{equation}\label{eq:deltaInTermsOfA}
\delta_\ell = \frac{A_{\geq \ell}}{a_\ell} \qquad \forall \ell \in [p]\enspace.
\end{equation}
Indeed,~\eqref{eq:deltaInTermsOfA} trivially holds for $p$ since $\delta_p = 1 = A_{\geq p}/a_p$. Now consider~\eqref{eq:deltaInTermsOfA} for some index $\ell\in [p-1]$, and assume that~\eqref{eq:deltaInTermsOfA} holds when replacing $\ell$ by any larger index, i.e., $\ell+1, \ell+2, \ldots, p$. Then,
\begin{align*}
\frac{A_{\geq \ell}}{a_\ell} = \frac{a_\ell}{a_\ell} + \left( \sum_{j=\ell+1}^p \frac{a_j}{a_{\ell+1}} \right) \frac{a_{\ell+1}}{a_\ell} = 1 + \left( \delta_{\ell+1}\right)  \frac{\delta_{\ell} -1}{\delta_{\ell+1}} = \delta_\ell\enspace.    
\end{align*}
From this, we can see that the first equality of the statement holds.
\begin{align*}
    a_j\prod_{i=1}^{j-1} \left(1 - \frac{a_i}{A_{\geq i}}\right) = \prod_{i=1}^{j-1} \frac{1}{1-1/\delta_i} \cdot \prod_{i=1}^{j-1} \left(1 - \frac{1}{\delta_i}\right) = 1\enspace.
\end{align*}
For the second equality, note that $(1-1/\delta_j) (\delta_j - 1) = \delta_{j+1}$ is equivalent to
%\begin{equation}\label{eq:deltasIncreasingByAtLeast1}
$\delta_j   =\delta_{j+1} +  2-{1}/{\delta_j}$. %\enspace.
%\end{equation}
Hence,
\begin{equation}\label{eq:deltaInTermsOfLargerDeltas}
\delta_j = \delta_p + \sum_{i=j}^{p - 1} \left(2- \frac{1}{\delta_i}\right) = \sum_{i=j}^p \left(2- \frac{1}{\delta_i}\right) \enspace,
\end{equation}
which implies together with~\eqref{eq:deltaInTermsOfA} that
\begin{align*}
    \frac{A_{\geq j}}{a_j} = \delta_j =  \sum_{i=j}^p\left( 2- \frac{1}{\delta_i}\right) = \sum_{i=j}^p \left(2- \frac{a_i}{A_{\geq i}} \right) \enspace.
\end{align*}
Finally, to show~\eqref{eq:gen_hardness_bounds}, we first use~\eqref{eq:gen_hardness_nice_sum} to obtain 
\begin{align}\label{eq:AoveraEval}
\frac{A_{\geq p-j+1}}{a_{p-j+1}} = \sum_{i= p-j+1}^p \left(2 - \frac{1}{\delta_i}\right)\enspace.
\end{align}
The upper bound of~\eqref{eq:gen_hardness_bounds} now follows by observing that $\delta_i \geq 0$ for $i\in [p]$, and $\delta_p=1$, which implies
\begin{align*}
\frac{A_{\geq p-j+1}}{a_{p-j+1}} &= \sum_{i=p-j+1}^p \left(2-\frac{1}{\delta_i}\right)
 \leq \left(\sum_{i=p-j+1}^{p-1} 2\right) + 1 = 2j - 1\enspace.
\end{align*}
Moreover, the lower bound of~\eqref{eq:gen_hardness_bounds} follows from
\begin{align*}
\frac{A_{\geq p-j+1}}{a_{p-j+1}} &= \sum_{i=p-j+1}^p \left(2-\frac{1}{\delta_i}\right)
 = \sum_{i=p-j+1}^p \left(2 - \frac{a_i}{A_{\geq i}}\right)
 \geq \sum_{i=p-j+1}^p \left(2 - \frac{1}{p-i+1}\right) = 2j - H_j\enspace,
\end{align*}
where the first equality comes from~\eqref{eq:AoveraEval}, the second one is due to~\eqref{eq:deltaInTermsOfA}, and the inequality holds because the values $a_i$ are strictly increasing, which implies $A_{\geq i} \geq (p-i+1) \cdot a_{i}$.
\end{proof}

\subsection{Value of Solutions Without any Optimal Elements}
\label{sec:gen_hardness_value_gap}

Consider a function $f_{o_1, \ldots, o_p} \in \cF$. From its definition~\eqref{eq:gen_hardness_def_f}, it is clear that $\{o_1, \ldots, o_p\}$ is an optimal solution of value $f(\{o_1, \ldots, o_p\}) = f(W) =  \sum_{i=1}^p a_i = A_{\geq 1} = A_{\geq 1}/a_1$, which by~\eqref{eq:gen_hardness_bounds} is at least $2p - H_p$ and at most $2p$.  The following lemma, therefore, implies the value gap property of \cref{lemma:gen_hardness_F_properties}. 
%now states that any solution $S \subseteq V \setminus \Oset$ of cardinality $p$ is almost a factor $2$ smaller than the optimum. 
%In this section we first select the numbers $a_1, \ldots, a_p$ and then show that our selection is such that, for any $o_1\in V_1, \ldots, o_p \in V_p$, 
%\begin{align*}
%    \max_{S \subseteq V \setminus \Oset: |S| = p} f_{o_1, \ldots, o_p}(S) \leq f_{o_1, \ldots, o_p}(\Oset) \cdot \frac{1}{2}\,,
%\end{align*}
%where $\Oset = \{o_1, \ldots, o_p\}$. The numbers are selected to satisfy the properties of the following lemma.
\begin{lemma}
    For any $f_{o_1, \ldots, o_p} \in \cF$,
    \begin{align*}
        \max_{S\subseteq W \setminus \{o_1, \ldots, o_p\}: |S|\leq p} f_{o_1, \ldots, o_p} (S) \leq p + (H_p)^2\enspace.
    \end{align*}
    \label{lem:gen_hardness_gap}
\end{lemma}
    The rest of this section is devoted to the proof of the above lemma. 
    Throughout, we let $W' = W \setminus \{o_1, \ldots, o_p\}$ and  
    denote by $f$ the submodular function obtained by restricting $f_{o_1, \ldots, o_p}$  to the ground set $W'$. By definition (see~\eqref{eq:gen_hardness_def_f}), we then have  for every $S \subseteq W'$
    \begin{align*}
        f(S)  = \sum_{j=1}^p a_j \cdot  \left( 1 - \prod_{i=1}^j \left(1 - \frac{a_i}{A_{\geq i}} \right)^{s_i} \right)\,,
    \end{align*}
    where $s_i = |S \cap W_i|$. The value of a set $S \subseteq W'$ is, thus, determined by $s_1 = |S \cap W_1|, \ldots, s_p = |S \cap W_p|$. In the subsequent, we slightly abuse notation and sometimes write $f(s_1, \ldots, s_p)$ for $f(S)$ to highlight that the value only depends on the number of elements from each partition and not on the actual elements.
    
   %  The outline of the proof  is as follows. 
   %  By monotonicity, we may assume that any set $S$ that maximizes the left-hand-side is of cardinality $p$.
    Now assume first that $S$ contains exactly one element from each $W_i$, say $S = \{v_1, \ldots, v_p\}$ where $v_i \in W_i \setminus \{o_i\}$ for every $i \in [p]$. Then, 
    \begin{align*}
        f (S) = \sum_{i=1}^p f(v_i \mid \{v_1, \ldots, v_{i-1}\})\enspace.
    \end{align*}
    Recall that  we selected the weights $a_1, \ldots, a_k$  so that each of the terms in this sum equals $1$ (see \cref{sec:gen_hardness_sel_a}).
    Thus, we have that the value of the set $S$ equals $p$, which  can also be seen from the following basic calculation.
    \begin{align*}
        f(S) = f(1,1,\ldots, 1) & = \sum_{j=1}^p a_j \cdot  \left( 1 - \prod_{i=1}^j \left(1 - \frac{a_i}{A_{\geq i}} \right)  \right) \\
        & =  A_{\geq 1} - \sum_{j=1}^{p}  \left( 1- \frac{a_j}{A_{\geq j}} \right) & \mbox{\small\eqref{eq:gen_hardness_become_one} in \cref{lem:gen_hardness_select_a}}
\\
    & =  A_{\geq 1} - \left( A_{\geq 1} - p\right)  & \mbox{\small $\sum_{j=1}^p  \left(2 - \frac{a_j}{A_{\geq j}} \right) = \frac{A_{\geq 1}}{a_1}$ by \cref{lem:gen_hardness_select_a} }\\
    & = p\enspace.
    \end{align*}
    
    The inequality of \cref{lem:gen_hardness_gap} thus holds in the case when $S$ contains exactly one element from each $W_i$. 
    However, it turns out that such a set is only an approximate maximizer to the left-hand-side of the lemma, and we need  an additional argument to bound the value of any set $S \subseteq W'$ of cardinality at most $p$. 
    We do so by defining a continuous concave version $\widehat{F}$ of the submodular function $f$, which, loosely speaking, can be thought of as being a continuous extension of $f$. 
    By leveraging the concavity of $\widehat{F}$, we can obtain upper bounds through a well-chosen first-order approximation. 
    Specifically, we consider the linear upper bound on the concave function obtained by taking its gradient at the point $\vec{1} = (1, 1, \ldots, 1)$ corresponding to sets which contain exactly one element of each $W_i \setminus \{o_i\}$ (see~\eqref{eq:gen_hardness_upper_bound}).

   %%  \paragraph{Continuous concave function $\widehat{F}$.}    Denote by $f$ the submodular function obtained by restricting $f_{o_1, \ldots, o_p}$  to the ground set $V' = V \setminus \Oset$. By definition (see~\eqref{eq:gen_hardness_def_f}), we then have  for every $S \subseteq V'$
   %%  \begin{align*}
   %%      f(S)  = \sum_{j=1}^p a_j \cdot  \left( 1 - \prod_{i=1}^j \left(1 - \frac{a_i}{A_{\geq i}} \right)^{s_i} \right)\,,
   %%  \end{align*}
   %%  where $s_i = |S \cap V_i|$. The value of a set $S \subseteq V'$ is thus determined by $s_1 = |S \cap V_1|, \ldots, s_p = |S \cap V_p|$. Now, let 
   To define $\widehat{F}$, let us first define $F\colon \mathbb{R}_{\geq 0}^p \rightarrow \mathbb{R}_{\geq 0}$ to be the following continuous proxy for $f$.
   \begin{align*}
        F(s_1, \ldots, s_p) = \sum_{j=1}^p a_j \cdot  \left( 1 - \prod_{i=1}^j \left(1 - \frac{a_i}{A_{\geq i}} \right)^{s_i} \right)\,.
    \end{align*}
    By definition, $F(s_1, \ldots, s_p) = f(s_1, \ldots, s_p)$  for integral vectors, and thus,
    \begin{align*}
        \max_{S \subseteq W'} f(S) = \max_{s_1, \ldots, s_p \in \mathbb{Z}_{\geq 0}: \sum_{i} s_i = p} F(s_1, \ldots, s_p) \leq  \max_{s_1, \ldots, s_p \in \mathbb{R}_{\geq 0}: \sum_{i} s_i = p} F(s_1, \ldots, s_p)\,.
    \end{align*}
    To simplify calculations, we further upper bound $F$ by the function $\widehat{F}$ defined as follows.
    \begin{align*}
        \widehat{F}(s_1, \ldots, s_p) =a_p +  \sum_{j=1}^{p-1} a_j \cdot  \left( 1 - \prod_{i=1}^j \left(1 - \frac{a_i}{A_{\geq i}} \right)^{s_i} \right)\,.
    \end{align*}
    Note that $\widehat{F}$ is a sum of concave functions, and it is thus a concave function on its own. 
    If we let $D = \{ (s_1, \ldots, s_p) \in \mathbb{R}_{\geq 0}^p : \sum_{i=1}^p s_i = p\}$ be the ``feasible'' region, then because of concavity,
\begin{align*}
\max_{\vec{s} \in D} {F}(s) \leq
\max_{\vec{s} \in D} \widehat{F}(s) \leq  
    \widehat{F}(\vec{\bar s}) +  \max_{\vec{s} \in D} \langle \nabla \widehat{F}(\vec{\bar s}), \vec{s} -  \vec{\bar s} \rangle
\end{align*}
for every vector $\vec{\bar s} = (\bar s_1 , \ldots \bar s_p)$.
We select $\vec{\bar s} = \vec{1} = (1 \ldots, 1)$ to be the all-ones vector,  which thus gives the upper bound
\begin{align}
    \max_{S \subseteq V'} f(S) & \leq \widehat{F}(\vec{1}) +  \max_{\vec{s} \in D} \langle \nabla \widehat{F}(\vec{1}), \vec{s} -  \vec{1} \rangle\enspace .
    \label{eq:gen_hardness_upper_bound}
\end{align}
We now consider each of these two terms, starting with $\widehat{F}(\vec{1})$. As $\frac{a_p}{A_{\geq p}} = 1$,
\begin{align*}
    \widehat{F}(\vec{1}) = a_p +  \sum_{j=1}^{p-1} a_j \cdot  \left( 1 - \prod_{i=1}^j \left(1 - \frac{a_i}{A_{\geq i}} \right) \right)  = \sum_{j=1}^{p} a_j \cdot  \left( 1 - \prod_{i=1}^j \left(1 - \frac{a_i}{A_{\geq i}} \right) \right) = F(\vec{1})\enspace . 
\end{align*}
We have already shown that $f(\vec{1}) = f(1,1,\ldots, 1) = p$ and, as $F$ equals $f$ on integral values, we have $\widehat{F}(\vec{1}) = p$.

It remains to bound $\max_{\vec{s} \in D} \langle \nabla \widehat{F}(\vec{1}), \vec{s} -  \vec{1} \rangle$. We have
\begin{align*}
    \max_{\vec{s} \in D} \langle \nabla \widehat{F}(\vec{1}), \vec{s} - \vec{1}\rangle &= p \cdot \max_{\ell \in [p]} \frac{\partial \widehat{F}}{\partial s_\ell} (\vec{1})  - \sum_{\ell =1}^p \frac{\partial \widehat{F}}{\partial s_\ell} (\vec{1}) \\
    & \leq p/2 - \sum_{\ell=1}^p \left( \frac{1}{2} - \frac{H_{p+1-\ell}}{p+1-\ell} \right) \\
    & = \sum_{i=1}^p \frac{H_i}{i} \leq (H_p)^2\enspace,
 \end{align*}
 where the inequality follows from the bounds on the partial derivatives given by  \cref{claim:partial_derivatives}. Assuming that claim, 
we have thus shown the statement of \cref{lem:gen_hardness_gap}, i.e., that any solution of cardinality at most $p$ without any optimal elements has value at most $p + (H_p)^2$.

% \begin{restatable}{theorem}{thmlowertwoplayersub} \label{thm:lowertwoplayersub}
% For every $\varepsilon \in (0, 1/4)$, any two-player (randomized) protocol with an approximation guarantee of $(\nicefrac{2}{3} + \varepsilon)$ must have a communication complexity of $\Omega(\frac{N\varepsilon}{k})$ bits in the regime $k \geq \varepsilon^{-1}$.
% \end{restatable}
\begin{restatable}{claim}{claimpartialderivatives}
    For $\ell = 1, \ldots, p$, 
    \begin{align*}
        \frac{1}{2} - \frac{H_{p+1-\ell}}{p+1-\ell} \leq \frac{\partial \widehat{F}}{\partial s_\ell} (\vec{1}) \leq \frac{1}{2}\enspace.
    \end{align*}
    \label{claim:partial_derivatives}
\end{restatable}
The claim follows from basic calculations and the identities of \cref{lem:gen_hardness_select_a}. As these calculations are mechanical and not very insightful, they can be found in \cref{app:bounding_partial_derivatives}.

\subsection{Players Have No Information About Their Optimum Element}
\label{sec:gen_hardness_indistinguishability}

The second key property of our family $\cF$ is indistinguishability: if one can only query the submodular function $f_{o_1,\ldots, o_p}$ on $W_1 \cup W_2 \cup \cdots \cup W_{\ell}$, then no information can be  obtained about which element in $W_\ell$ is selected to be $o_{\ell}$. 
While this is intuitively clear from the description of $\cF$ in \cref{sec:gen_hardness_intuition}, the following lemma gives the formal proof of the indistinguishability property of \cref{lemma:gen_hardness_F_properties}.

\begin{lemma}
   For $\ell\in [p]$, any two functions $f_{o_1, \ldots, o_p}, f_{o'_1, \ldots, o'_p} \in \cF$ with $o_1 = o'_1, \ldots, o_{\ell-1} = o'_{\ell-1}$ are identical when restricted to the ground set $W_1 \cup \cdots \cup W_\ell$.
%     For $\ell \in [p]$ and $z_1\in V_1, \ldots, z_{\ell-1} \in V_{\ell-1}$, define
%     \begin{align*}
%         \cF' = \{f_{o_1, \ldots, o_p} \in \cF :  o_{1} = z_1 , \ldots, o_{\ell-1} = z_{\ell-1} \}\enspace.
%     \end{align*}
%     Then any two functions in $\cF'$ are identical when restricted to the ground set $V_1 \cup \cdots \cup V_{\ell}$.
    \label{lem:gen_hardness_no_information}
\end{lemma}
\begin{proof}
By definition (see~\eqref{eq:gen_hardness_def_f}),
\begin{equation*}
   f_{o_1,\ldots, o_p}(S) = \sum_{j=1}^p a_j \cdot  \left( 1 - \mathbbm{1}\{o_j\not\in S\} \prod_{i=1}^j \left(1 - \frac{a_i}{A_{\geq i}} \right)^{|S\cap (W_i \setminus \{o_i\})|} \right)\enspace. 
\end{equation*}
As $o_1 = o'_1, \ldots, o_{\ell-1}= o'_{\ell-1}$, the first $\ell-1$ terms in the above sum are identical for $f_{o_1, \ldots, o_p}$ and $f_{o'_1, \ldots, o'_p}$.  Moreover, since we have restricted our ground set to $W_1 \cup W_2 \cup \cdots \cup W_{\ell}$, their values are independent of $o_{\ell+1}, \ldots, o_{p}$ and $o'_{\ell+1}, \ldots, o'_p$, respectively. Therefore, for any $S \subseteq W_1 \cup W_2 \cup \cdots \cup W_{\ell}$,  we can write the difference $f_{o_1, \ldots, o_p}(S) - f_{o'_1, \ldots, o'_p}(S)$ as 
\begin{align*}
f_{o_1,\ldots, o_p}(S) - f_{o'_1,\ldots, o'_p}(S) =    \prod_{i=1}^{\ell-1} \left(1 - \frac{a_i}{A_{\geq i}} \right)^{|S\cap (W_i \setminus \{o_i\})|} \cdot \left( \labeleq{eq:gen_hardness_diff_one} + \labeleq{eq:gen_hardness_diff_two}\right)\,,
\end{align*}
where
\begin{align*}
    \eqref{eq:gen_hardness_diff_one} = a_\ell \cdot \left(\mathbbm{1}\{o'_\ell \not\in S\}\left(1 - \frac{a_\ell}{A_{\geq \ell}} \right)^{|S\cap (W_\ell \setminus \{o'_\ell\})|} - \mathbbm{1}\{o_\ell \not\in S\}\left(1 - \frac{a_\ell}{A_{\geq \ell}} \right)^{|S\cap (W_\ell \setminus \{o_\ell\})|}\right) 
\end{align*}
and 
\begin{align*}
    \eqref{eq:gen_hardness_diff_two} = \left(\sum_{j=\ell+1}^p a_j\right) \cdot \left( \left(1 - \frac{a_\ell}{A_{\geq \ell}} \right)^{|S\cap (W_\ell \setminus \{o'_\ell\})|}
    - \left(1 - \frac{a_\ell}{A_{\geq \ell}} \right)^{|S\cap (W_\ell \setminus \{o_\ell\})|}
    \right)\,.
\end{align*}
We now finish the proof of the lemma by showing that $\eqref{eq:gen_hardness_diff_one} + \eqref{eq:gen_hardness_diff_two} = 0$. This is immediate if 
$\mathbbm{1}\{o_\ell \in S\} = \mathbbm{1}\{o'_\ell \in S\}$ because then $\eqref{eq:gen_hardness_diff_one} = 0$ and $\eqref{eq:gen_hardness_diff_two} = 0$.
To analyze the other case when $\mathbbm{1}\{o_\ell \in S\} \neq \mathbbm{1}\{o'_\ell \in S\}$, 
suppose that $o_\ell \in S$  and $o'_\ell \not \in S$. 
(The other case is symmetric.) Then, if we let $s = |S\cap W_\ell|$, we have
\begin{align*}
    \eqref{eq:gen_hardness_diff_one} = a_\ell \cdot \left(1 - \frac{a_\ell}{A_{\geq \ell}} \right)^{s} 
\end{align*}
and 
\begin{align*}
    \eqref{eq:gen_hardness_diff_two} &= \left( \sum_{j=\ell+1}^p a_j \right) \cdot \left( \left(1 - \frac{a_\ell}{A_{\geq \ell}} \right)^{s}
    - \left(1 - \frac{a_\ell}{A_{\geq \ell}} \right)^{s-1} \right) \\
    &=A_{\geq \ell+1} \cdot \left(1 - \frac{a_\ell}{A_{\geq \ell}} \right)^{s-1}  \cdot \left(-\frac{a_{\ell}}{A_{\geq \ell}}\right)\\
    & = - a_{\ell} \cdot \frac{A_{\geq \ell}-a_{\ell}}{A_{\geq \ell}} \cdot \left(1 - \frac{a_\ell}{A_{\geq \ell}} \right)^{s-1}\\
    & = - a_{\ell}\cdot\left(1 - \frac{a_\ell}{A_{\geq \ell}} \right)^{s}  \enspace,
\end{align*}
thus implying $\eqref{eq:gen_hardness_diff_one} + \eqref{eq:gen_hardness_diff_two} = 0$ as required.
\end{proof}

\subsection{Hardness Reduction from the \texorpdfstring{\chainPn}{CHAINp(n)} Problem}
\label{sec:gen_hardness_reduction}

In this section we present our reduction using the family $\cF$ of (weighted) coverage functions guaranteed by \cref{lemma:gen_hardness_F_properties}.  The arguments are very similar to those of \cref{sec:two-player-sub-hardness}, but instead of reducing from the INDEX problem, we reduce from the multiplayer version $\chainPn$, referred to in \cref{sec:index-model}.  

Recall that, in the \chainPn problem,  there are $p$ players $P_1, \ldots, P_p$. For $i=1, \ldots, p-1$, player $P_i$ has as input a bit string $x^{i} \in \{0,1\}^n$ of length $n$ and for $i=2, \ldots, p$ player $P_i$ (also) has as input an index $t^i\in \{1, \ldots, n\}$,
where we use the convention that the superscript of a string/index indicates the player receiving it. 
The players are promised that either $x^{i}_{t^{i+1}} = 0$ for all $i=1, \ldots, p-1$ (the $0$-case) or $x^{i}_{t^{i+1}} = 1$ for all $i = 1\, \ldots, p-1$ (the $1$-case).  The objective of the players  is to decide whether the input instance belongs to the $0$-case or the $1$-case.

For convenience, we restate the hardness result of \chainPn here (recall that the considered protocols are one-way protocols as defined in \cref{sec:index-model}).
%
%We are interested in the one-way randomized  communication complexity of the above problem to decide between the $0$-case and $1$-case. 
%Recall that a randomized one-way communication protocol is defined by (randomized) algorithms $\cA_1,\ldots, \cA_p$, one for each player. 
%Player $P_1$ calculates a message $m_1 = \cA_1(x^{1})$ as a function of the input $x^{1}$. 
%For $i=2, \ldots, p-1$, player $P_i$ then calculates the message $m_i = \cA_i(m_{i-1}, x^{i}, t^i)$ as a function of the messages  of the previous player and the data  $x^{i}$ and $t^i$ given to the $i$-th player.
%Finally, the last player, $P_p$, calculates $\cA_p(m_{p-1}, t^{p})$ and decides between the $0$-case and the $1$-case. 
%The protocol is said to have a success probability of $s$ if, for any input $x^{1},t^2, x^2,  \ldots, t^{p-1} x^{p-1}, t^p$, player $P_p$ outputs the correct answer with probability at least $s$, where the probability is taken over the randomness of the players.  
%We define the communication complexity of a protocol $\cA_1, \ldots, \cA_p$ to be the maximum bit-length of any message $m_i$ calculated by the players (where the maximum is taken over their randomness and the inputs). 
%
%Note that the $2$-INDEX($n$) problem is simply the INDEX($n$) problem and, by adapting the standard proofs for that problem (see \cref{app:pindex}), we have the following:
%
\thmpindex*
Our plan in the rest of this section is to reduce the \chainPn problem to the $p$-player submodular maximization problem on (weighted) coverage functions on a ground set of cardinality $N=p\cdot n$. As $n$ is clear from context, we simplify notation and refer to the \chainPn problem as the \chainP problem.

Let $PRT$  be a protocol for the $p$-player submodular maximization problem on weighted coverage functions subject to the cardinality constraint $k=p$.
Further assume that $PRT$ has  an 
 approximation guarantee of $\frac{p + (H_p)^2}{2p - H_p}(1 +\varepsilon)$ for $\varepsilon >0$. 
We show below that this leads to a protocol $\PRTpI$ for the $\chainPn$ problem whose message size depends on the message size of $PRT$. 
Since \cref{thm:pindex_hardness} lower bounds the message size of any protocol for \chainPn, this leads to a lower bound also on the message size of $PRT$.

In our reduction, we use  the weighted coverage functions in family $\cF$ whose existence is guaranteed by \cref{lemma:gen_hardness_F_properties}. These functions are defined over a common ground set $W$ that is     partitioned into sets $W_1, \ldots, W_p$ of cardinality $n$ each. For future reference, we let  
\begin{align*}
    W_i = \{v^i_1, v^i_2, \ldots, v^i_n\} \qquad \mbox{for every $i\in [p]$} \enspace,
\end{align*}
and we have $|W| = N = p \cdot n$.

Before getting to the protocol $\PRTpI$ mentioned above, we first present a simpler protocol for the \chainP problem that is used as a building block for $\PRTpI$.  The simpler protocol is presented as  \cref{alg:pplayersubreduction}. 
In words, if the \chainP instance is $x^{1}, x^{2}, \ldots, x^{p-1}, t^2, t^3 \ldots, t^{p}$, \cref{alg:pplayersubreduction} simulates $PRT$ on the following $p$-player \maxcard instance with $k=p$.
\begin{itemize}
    \item Player $i$ receives the subset $V_i = \{v^i_j \in W_i \mid j\in [n] \text{ with }x^i_j = 1\}$ of $W_i$ corresponding to the $1$-bits of $x^i$.
    \item The submodular function the players wish to maximize is $f_{o_1, \ldots, o_p} \in \cF$ with  
    \begin{align*}
        o_1 = v^1_{t^2}, o_2 = v^2_{t^3}, \ldots, o_{p-1}= v^{p-1}_{t^p}
    \end{align*}
    (by the indistinguishability property of \cref{lemma:gen_hardness_F_properties}, the choice of $o_p$ does not matter since these functions are identical).
\end{itemize}
The last player of \cref{alg:pplayersubreduction} then decides between the $0$-case and $1$-case depending on the value of the solution $S \subseteq V_1 \cup V_2 \cup \cdots \cup V_p$ outputted by the last player of $PRT$. Note that this value is informative because the elements $\{o_1, \ldots, o_p\}$ are in $V_1 \cup V_2 \cup \cdots \cup V_p$ if and only if the given \chainP instance is in the $1$-case.  

Some care has to be taken to make sure that the $i$-th player of \cref{alg:pplayersubreduction} can answer the oracle queries made during the simulation of player $P_i$ of $PRT$. 
This is the reason why the message from the $(i-1)$-th player to the $i$-th player of \cref{alg:pplayersubreduction} also contains the indices $t^2, t^3, \ldots, t^{i-1}$. 
Indeed, as player $i$ also receives index $t^i$, she can then infer the selection of $o_1, \ldots, o_{i-1}$, which in turn allows her to calculate the value $f_{o_1, \ldots, o_p}(S)$ of any set $S\subseteq W_1\cup \cdots \cup W_i$ due to the indistinguishability property of \cref{lemma:gen_hardness_F_properties}.

\begin{protocol}
\caption{Reduction from \chainP to \maxcard with $k=p$ in the $p$-Player Model} \label{alg:pplayersubreduction}
% \textbf{Player $P_1$'s Algorithm}
%\begin{algorithmic}[1]
%    \State The set of elements $P_1$ of $PRT$ gets is $V'_1 = \{e_j^1 \in V_1 \mid x^{(1)}_j = 1\}$.
%    \State The objective function for $PRT$ is one of the functions in $\cF$. Note that, by Lemma~\ref{lem:gen_hardness_no_information}, these functions are identical when restricted to $V'_1$. Thus, $P_1$'s part of $PRT$ can execute without knowing which one of them is the real objective function.
%    \State Send to $P_2$ the same message $m_1$ sent by the $P_1$ of $PRT$.
%\end{algorithmic}
\textbf{Player $P_i$'s Algorithm for $i=1, \ldots, p-1$}
\begin{algorithmic}[1]
    \State The set of elements $P_i$ of $PRT$ gets is $V_i = \{v_j^i \in W_i \mid j\in [n] \text{ with }x^{i}_j = 1\}$.
    \State The objective function for $PRT$ is one of the functions $f_{o_1, \ldots, o_p} \in \cF$ with %\begin{align*}
        $o_1 = v^1_{t^2}, \ldots, o_{i-1} = v^{i-1}_{t^i}$. 
    By the indistinguishability property of \cref{lemma:gen_hardness_F_properties}, these functions are identical when restricted to $W_1 \cup \cdots \cup W_i$, and so any oracle query from $P_i$ in $PRT$ can be evaluated without ambiguity. 
%    \State Player $P_i$ of $PRT$ gets oracle access to a function  $f_{o_1, \ldots, o_p} \in \cF$ 
%    %\end{align*} 
%    Note that by the indistinguishability property of \cref{lemma:gen_hardness_F_properties}, these functions are identical when restricted to $W_1 \cup \cdots \cup W_i$. Thus, $P_i$'s part of $PRT$ can execute without knowing which one of them is the real objective function.
    \State Send to $P_{i+1}$ the values $t^2, t^3, \ldots, t^i$ and the same message $m_i$ sent by $P_i$ of $PRT$. 
\end{algorithmic}
\textbf{Player $P_p$'s Algorithm}
\begin{algorithmic}[1]
    \State The set of elements $P_p$ of $PRT$ gets is $V_p = W_p$.
    \State The objective function for $PRT$ can now be determined to be $f_{o_1, \ldots, o_p}$, where $o_i = v^{i}_{t^{i+1}}$ for $i=1, \ldots, p-1$ (the last element $o_p$ does not matter by the indistinguishability property). 
    \State If $PRT$ returns a set of value at most $p+ (H_p)^2$, output ``$0$-case''; otherwise, output ``$1$-case''.
\end{algorithmic}
\end{protocol}

Our next step is analyzing the output distribution of~\cref{alg:pplayersubreduction}.

\begin{lemma} \label{lem:gen_hardness_twosidesreduction_oneside}
In the $0$-case, \cref{alg:pplayersubreduction} always produces the correct answer. 
\end{lemma}
\begin{proof}
Let $S \subseteq V_1 \cup V_2 \cup \cdots \cup V_p$ denote the output of $PRT$, which satisfies $|S|\leq k=p$.  Since we are in the $0$-case,  the set $\{o_1, \ldots, o_p\}$, where $o_i = v^i_{t^{i+1}}$, is disjoint from $V_1 \cup V_2 \cup \cdots \cup V_p$. Therefore, by the value gap property of \cref{lemma:gen_hardness_F_properties}, we must have $f_{o_1, \ldots, o_p}(S) \leq p + (H_p)^2$, and thus, \cref{alg:pplayersubreduction} always correctly decides that we are in the $0$-case.
\end{proof}

The assumption that $PRT$ has an approximation guarantee 
of $\frac{p + (H_p)^2}{2p - H_p}(1 +\varepsilon)$ implies the following success probability in the $1$-case.
\begin{lemma} \label{lem:gen_hardness_twosidesreduction_secondside}
In the $1$-case, \cref{alg:pplayersubreduction} produces the correct answer with probability at least $\varepsilon$.
\end{lemma}
\begin{proof}
Observe that in the $1$-case we have that the elements $o_1, o_2, \ldots, o_p$ all belong to $V_1 \cup V_2 \cup \cdots \cup V_p$. Thus, one solution which $PRT$ could produce is the set $\Oset  = \{o_1, \ldots, o_p\}$, whose value is at least
$ 2p - H_p$ 
by the value gap property of \cref{lemma:gen_hardness_F_properties}.
Since we assumed that $PRT$ has an approximation guarantee of $\frac{p+(H_p)^2}{2p-H_p}(1+\varepsilon)$, the expected value of the solution it produces must be at least $(p+(H_p)^2)(1 + \varepsilon)$. Together with the fact that the value of every solution is at most $2p$, we get  that $PRT$ produces a solution of value at most $p+(H_p)^2$ with probability of at most $1-\varepsilon$ since
\[
   (1 - \varepsilon) \cdot (p+(H_p)^2)  + \varepsilon \cdot (2p) \leq  (1+ \varepsilon) (p+ (H_p)^2) 
    \enspace.
\]
Thus, with probability at least $\varepsilon$, the value of the solution produced by $PRT$ is strictly more than $p+(H_p)^2$, which makes \cref{alg:pplayersubreduction}  output the correct decision with at least this probability.
\end{proof}

At this point we are ready to present the promised protocol $\PRTpI$, which simply executes $\lceil 2\varepsilon^{-1} \rceil$ parallel copies of \cref{alg:pplayersubreduction}, and then determines that the input was in the $1$-case if and only if at least one of the executions returned this answer.

\begin{corollary} \label{cor:gen_hardness_twosided_reduction}
$\PRTpI$ always answers correctly in the $0$-case, and answers correctly with probability at least $2/3$ in the $1$-case.
\end{corollary}
\begin{proof}
The first part of the corollary is a direct consequence of \cref{lem:gen_hardness_twosidesreduction_oneside}. Additionally, by \cref{lem:gen_hardness_twosidesreduction_secondside}, the probability that $\PRTpI$ answers wrongly in the $1$-case is at most
\[
    (1 - \varepsilon)^{\lceil 2\varepsilon^{-1} \rceil}
    \leq
    (1 - \varepsilon)^{2\varepsilon^{-1}}
    \leq
    e^{-(\varepsilon) \cdot 2\varepsilon^{-1}}
    =
    e^{-2}
    <
    \frac{1}{3}
    \enspace.
    \qedhere
\]
\end{proof}

Using the last corollary, we can now complete the proof of \cref{thm:gen_hardness}.
\begin{proof}[Proof of \cref{thm:gen_hardness}]
Since \cref{cor:gen_hardness_twosided_reduction} shows that $\PRTpI$ is a  protocol for the \chainP problem that succeeds with probability at least $2/3$, \cref{thm:pindex_hardness} guarantees that its message size is at least $n/(36p^2)$. Observe now that the messages of $\PRTpI$ consist of $\lceil 2 \varepsilon^{-1} \rceil$ messages of \cref{alg:pplayersubreduction}, and thus, there must be   a message of \cref{alg:pplayersubreduction} of size at least
\[
    \frac{n/(36p^2)}{\lceil 2 \varepsilon^{-1} \rceil}
    \geq
    \frac{n\varepsilon}{108p^2}
    \enspace.
\]
We now recall that each message of \cref{alg:pplayersubreduction} includes only some of the indices $t^2, \ldots, t^p$ and a messages generated by $PRT$ given the instance of $p$-player submodular maximization generated for it by \cref{alg:pplayersubreduction}. The indices are sent using at most $p \lceil \log_2(n) \rceil$ bits. Thus, because the number of elements of this instance is $N = n\cdot p$, $PRT$ must send a message of size  at least
\[
    \frac{n\varepsilon}{108p^2} - p \lceil \log_2(n) \rceil
    =
    \Omega \left(\frac{N \varepsilon}{p^3}\right)
    \enspace.
    \qedhere
\]
\end{proof}

%%% Then player $P_2$, as a function of the received message $m_1$ and the inputs $x^{(2)}}$ and $\alpha_2$, sends a random message $m_2$ to player $P_3$ who, as a function of the received message $m_2$ together with $x^{(3)}}$ and $\alpha_3$, sends a message $m_3$ to $P_4$ and so on. 
%%% At the end player $P_k$ receives message $m_{k-1}$ and $\alpha_k$ and needs to decide whether we 
%%% \begin{itemize}
%%%     \item The protocol is a randomized one-way communication protocol. That is, 
%%% \end{itemize}
%%% Their goal is to design a randomized one-way communication protocol that transmits as few bits as possible while maximizing the probability that $P_k$ correctly decides whether 
%%% 
%%% to send as communicate as few bits as possible in order to maximize the probability  

%In this problem Alice gets $n$ bits: $b_1, b_2, \dotsc, b_m$, and can then send a message to Bob. Bob gets the message of Alice and an index $1 \leq i \leq m$, and based on these two pieces of information alone should output the value of $b_i$. Clearly, Bob can produce the correct answer with probability $1/2$ by outputing a random bit. However, it is known that Bob cannot gurantee any larger constant probability of success, unless the message he gets from Alice is of linear (in $m$) size. In particular, the following result was proved by~\cite{XXX}.

%% file: 410-intuition.tex
\subsection{Intuitive Description of Our Construction}
\label{sec:gen_hardness_intuition}

In this section, we highlight our main ideas for constructing the family $\cF$ satisfying the properties of \cref{lemma:gen_hardness_F_properties}. 
We do so by presenting  three families of coverage functions $\cH$, $\cG$, and finally $\cF$. 
Family $\cH$ is a natural adaptation of coverage functions that have previously appeared in hardness constructions (see, e.g., \cite{mcgregor2019better}). 
We then highlight our main ideas for overcoming issues with those functions by first refining $\cH$ to $\cG$, and then by refining $\cG$ to obtain our final construction $\cF$. 

To convey the intuition, we work with unweighted coverage functions. However, to provide a clean and concise technical presentation later on, we use weighted coverage functions to formally realize the construction plan described here.

\paragraph{The first attempt: family $\cH$.}
The construction of the family $\cH = \{h_{o_1, \ldots, o_p} \mid o_1 \in W_1, \ldots, o_p\in W_p\}$ is inspired by the coverage functions constructed in the NP-hardness result  of~\cite{feige1998threshold}. In those coverage functions every element corresponds to a subset of the underlying universe of size $|U|/p$. Furthermore, the optimal solution $\{o_1, \ldots, o_p\}$ forms a disjoint cover of $U$, whereas any other element behaves like a random subset of the universe of size $|U|/p$.  

Inspired by this, we let  $h_{o_1, \ldots, {o_{p}}} \in \cH$  be  the coverage function where
\begin{itemize}
    \item the subsets of $U$ corresponding to $o_1,\ldots, o_p$ form a partition of equal-sized sets, i.e.,  of size $|U|/p$ each;
    \item every other element corresponds to a randomly selected subset of $U$ of size $|U|/p$.
\end{itemize}
While the above definition is randomized, we assume  for the sake of simplicity in this overview  that the value of a subset equals its expected value.   This can intuitively be achieved by selecting the underlying universe $U$ to be large enough so as to ensure concentration.  For a subset $S \subseteq W \setminus \{o_1, \ldots, o_p\}$, we thus have that $h_{o_1, \ldots, o_p}(S)$ equals the expected number of elements of $U$ covered by $|S|$ random subsets of cardinality $|U|/p$. Hence
\begin{align*}
    h_{o_1, \ldots, o_p}(S) = \left( 1- \left(1- \frac{1}{p} \right)^{|S|} \right) |U| \enspace,
\end{align*}
which is at least $(1-1/e) |U|$ if $|S| =p$. 
This already highlights the first issue of the construction: the value gap between the optimal solution $\{o_1, \ldots, o_p\}$, whose value is $|U|$, and a solution disjoint from this optimal solution is only $1-1/e$; while we need it to approach $1/2$ as $p$ tends to infinity. 
%of since $h_{o_1, \ldots, o_p}(o_1, \ldots, o_p) = |U|$  we have
%\begin{align*}
%     \frac{\max_{S \subseteq W \setminus \{o_1, \ldots, o_p\}: |S| \leq p} h_{o_1, \ldots, o_p}(S)}
%     {h_{o_1, \ldots, o_p}(o_1, \ldots, o_p)} \geq  (1-1/e) 
%\end{align*}
%and thus the value gap is bounded away from our goal of $1/2$ (for any $p$). 

The second and perhaps more significant issue is the indistinguishability. First, we can observe that the value of any subset $S\subseteq W_1$ %equals 
%\begin{align*}
 %   \left( 1- \left(1- \frac{1}{p} \right)^{|S|} \right) |U| 
%\end{align*}
only depends on $|S|$,
and thus the selection of $o_1\in W_1$ is indistinguishable when querying the submodular function restricted to $W_1$. 
However, the same does \emph{not} hold for $o_2$ when querying the submodular function restricted to the set $W_1 \cup W_2$. 
To see this, note that $o_1$ and $o_2$ are the \emph{only} elements of $W_1 \cup W_2$ whose corresponding subsets of $U$ are disjoint. In other words, $\{o_1, o_2\}$ is the unique maximizer to $\max_{S \subseteq W_1 \cup W_2: |S| = 2} h_{o_1, \ldots, o_p}(S)$, and $o_2$ can thus be identified by querying the submodular function on $W_1 \cup W_2$. 
A natural idea for addressing this issue is to make \emph{all} elements in $W_2$, and not only $o_2$, correspond to subsets of $U$ that are disjoint of the subset corresponding to  $o_1$. 
Making this modification for all $W_2, \ldots, W_p$ results in the refined family $\cG$ that we now describe. We note that a similar approach was used in~\cite{kapralov2013better} to guarantee indistinguishability.

\paragraph{The first refinement: family $\cG$.} Motivated by the idea to make every element in $W_i$ correspond to a subset of $U$ disjoint from the subsets of $o_1, \ldots, o_{i-1}$, we define the family $\cG = \{g_{o_1, \ldots, o_p} \mid o_1 \in W_1, \ldots, o_p\in W_p\}$ of coverage functions. Specifically, we let $g_{o_1, \ldots, o_p} \in \cG$ be the coverage function where 
\begin{itemize}
    \item the subsets of $U$ corresponding to $o_1,\ldots, o_p$ form a partition of equal-sized sets, i.e., of size $|U|/p$ each;
    \item for $i=1, \ldots, p$, every element in $W_i \setminus \{o_i\}$ corresponds to a randomly selected subset of $U$ of size $|U|/p$ that is disjoint from the subsets corresponding to $o_1, \ldots, o_{i-1}$.
\end{itemize}
The above description of $\cG$ is given in a way that highlights the changes compared to $\cH$. Another equivalent definition of $g_{o_1, \ldots, o_p}$ is that it is the coverage function where
\begin{itemize}
    \item the elements of $W_1$ form random subsets of $U$ of size $|U|/p$; 
    \item for $i=2, \ldots, p$, every element in $W_i$ corresponds to a randomly selected subset of $U$ of size $|U|/p$ that is disjoint from the subsets corresponding to $o_1, \ldots, o_{i-1}$.
\end{itemize}
From this viewpoint, it is clear that we now have the indistinguishability property of \cref{lemma:gen_hardness_F_properties}. Indeed, for $i\in [p]$, the only subsets of $U$ that depend on $o_i$ in the above construction are those corresponding to elements in $W_{i+1}, \ldots, W_p$. 
It follows that the value of a subset $S \subseteq W_1 \cup \cdots \cup W_i$, which is a function of the subsets of $U$ corresponding to the elements in $S$, is independent of the selection of $o_i$.

Having verified indistinguishability, let us consider the value gap. 
First, note that we still have that the optimal solution $\{o_1, \ldots, o_p\}$ covers the whole universe, and thus has value $|U|$.
Now consider a set $S \subseteq W \setminus \{o_1, \ldots, o_p\}$. 
%For $j\in [p]$, the fraction of the subset of $U$ corresponding to $o_j$ that is covered by the subsets corresponding to the elements in $S$ is then
%\begin{align*}
%    \left( 1- \prod_{j=1}^i \left(1- \frac{(1-(j-1)/p)}
%\end{align*}
It will be instructive to first consider  the case when $S = \{v_1, \ldots, v_p\}$ with $v_i \in W_i$ for all $i \in [p]$, i.e., $S$ contains exactly one element from each of the sets $W_i$.   Abbreviating $g_{o_1, \ldots, o_p}$ by $g$,  we have in this case that
\begin{align*}
 g_{o_1, \ldots, o_p}(S) =  g(v_1) + g(v_2 \mid \{v_1\}) + g(v_3 \mid \{v_1, v_2\}) + \ldots + g(v_p \mid \{v_1, \ldots, v_{p-1}\}) 
\end{align*} 
equals
\begin{align}
\label{eq:gen_hardness_unequal}
 \frac{|U|}{p} \left(1 +  \left( 1- \frac{1}{p}\right) +  \left(1 - \frac{1}{p} \right) \left( 1 - \frac{1}{p-1} \right) + \ldots +  \left(1 - \frac{1}{p} \right) \left( 1 - \frac{1}{p-1} \right) \cdots \left(1 - \frac{1}{2} \right) \right)\enspace,
\end{align}
which in turn solves to 
\begin{align*}
    \frac{|U|}{2} \left( 1+ \frac{1}{p} \right)\enspace.
\end{align*}
Hence, for sets $S \subseteq W \setminus \{o_1, \ldots, o_p\}$ that contain one element from each $W_i$, we have a value gap that approaches the desired constant $1/2$ as $p$ tends to infinity. 
The issue is that there are other subsets of $W \setminus \{o_1, \ldots, o_p\}$ of significantly higher value. 
To see this, note that any element $v\in W_1 \setminus \{v_1\}$ has a marginal value with respect to $\{v_1, \ldots, v_{p-1}\}$ that is much higher than the marginal value of $v_p$ with respect to the same set. 
In particular, the value of a subset $W_1 \setminus \{o_1, \ldots, o_p\}$ of cardinality $p$ is equivalent for functions in $\cG$ and $\cH$, and  is thus at least $(1-1/e)|U|$.  To overcome this issue (i.e., the fact that elements of $W_1$  are more ``valuable'' than other elements), we modify the above construction to let the elements from different $W_i$'s correspond to subsets of different sizes. 

\paragraph{The second and last refinement: family $\cF$.} The family $\cF = \{f_{o_1, \ldots, o_p} \mid o_1 \in W_1, \ldots, o_p\in W_p\}$ is obtained from $\cG$ by selecting subsets of $U$ of \emph{non-uniform sizes}. Specifically, we carefully select numbers $1=a_1 < a_2 < \cdots < a_p$, and make the elements of $W_i$ correspond to subsets of $U$ of size $a_i$; then we let the total size of $U$ be $a_1 + a_2 + \ldots + a_k$.\footnote{We remark that the $a_i$'s do not take integral values, and we think of $U$ as a set of total size $a_1 + a_2 + \ldots + a_p$ consisting of infinitly many infinitesimally small items.} We now let $f_{o_1, \ldots, o_p} \in \cF$ be the coverage function where
\begin{itemize}
    \item the elements of $W_1$ form random subsets of $U$ of size $a_1=1$; 
    \item for $i=2, \ldots, p$, every element in $W_i$ corresponds to a randomly selected subset of $U$ of size $a_i$ that is disjoint from the subsets corresponding to $o_1, \ldots, o_{i-1}$.
\end{itemize}
The family $\cF$ satisfies the indistinguishability property of \cref{lemma:gen_hardness_F_properties} for the exact same reasons $\cG$ satisfies it.
We now explain how the values $a_1, \ldots, a_p$ are selected so as to obtain the value gap.  
Consider a set $S \subseteq W \setminus \{o_1, \ldots, o_p\}$ obeying $S = \{v_1, \ldots, v_p\}$ for some choice of $v_i \in W_i$ for every $i \in [p]$. Abbreviating $f_{o_1, \ldots, o_p}$ by $f$, we thus have
\begin{align*}
    f_{o_1, \ldots, o_p}(S) =  f(v_1) + f(v_2 \mid \{v_1\}) + f(v_3 \mid \{v_1, v_2\}) + \ldots + f(v_p \mid \{v_1, \ldots, v_{p-1}\}) \enspace.
\end{align*}
The numbers $a_1, \ldots, a_p$ are selected so that each term of this sum equals $1$, and hence, $f_{o_1, \ldots, o_p}(S) = p$. Notice that this is in stark contrast to the functions in $\cG$ where the contributions to~\eqref{eq:gen_hardness_unequal} were highly unequal. 
The intuitive reason why we set the numbers so that these marginal contributions are the same is that we want to prove that one cannot form a subset of $W \setminus \{o_1, \ldots, o_p\}$ of cardinality at most $p$ of significantly higher value by increasing the number of elements selected from one of the partitions $W_i$. 
Formally, this is proved in \cref{sec:gen_hardness_value_gap} by considering the linear extension of a concave function at the point corresponding  to such a set $S$ that contains a single element from each $W_i$.
This allows us to upper bound the value of any subset $W \setminus \{o_1, \ldots, o_p\}$ of cardinality at most $p$ by $p + (H_p)^2$. The value gap then follows from basic calculations (see \cref{lem:gen_hardness_select_a}) which show that 
$f_{o_1, \ldots, o_p} (\{o_1, \ldots, o_p\}) = |U| = \sum_{i=1}^p a_i$ is at least  $2p - H_p$ and at most $2p$.

%% file: 500-polytime.tex
\section{Polynomial Time Submodular Maximization for Two Players}\label{sec:poly2players}

In this section, we discuss our result about efficient protocols in the two-player setting. Hence, throughout this section, $f\colon 2^W\to \mathbb{R}_{\geq 0}$ is a monotone submodular function defined on a ground set $W$, and $k\in \mathbb{Z}_{\geq 0}$ is an upper bound on the cardinality of subsets of $W$ that we consider. $V_A\subseteq W$ are the elements that Alice receives and $V_B\subseteq W$ the ones that Bob receives, and we define $V=V_A\cup V_B$. We denote by $\mathcal{O}\subseteq V$ an optimal solution to the (offline) problem $\max\{f(S): S\subseteq V, |S|\leq k\}$.

The well-known \Greedy algorithm by Nemhauser et al.~\cite{nemhauser1978analysis} is a crucial ingredient in our protocol. We remind the reader that \Greedy starts with an empty set $S$, and in each step adds an element $v\in V$ to $S$ with largest marginal gain, i.e.,
$v\in \argmax_{u \in V} f(u \mid S)$.
We recall the following basic performance guarantee of \Greedy, which can readily be derived from its definition (see also~\cite{nemhauser1978analysis}): when running \Greedy for $p$ rounds, a set $S\subseteq V$ with $|S|=p$ is obtained that satisfies, for any $\ell\in \mathbb{Z}_{>0}$,
\begin{align} \label{eq:greedyvalue}
    f(S) \geq \left( 1- \left(1-\frac{1}{\ell}\right)^p \right) \cdot f(O_\ell)\enspace ,
\end{align}
where $O_\ell$ denotes the optimum solution of size $\ell$.

\input{510-polytime-algo.tex}
%\input{520-polytime-hardness.tex}

%% file: 510-polytime-algo.tex
\subsection{Protocol}

We consider the simple deterministic protocol for the two-player setting given as \cref{alg:simple2pAlg}. Without loss of generality we assume that $|\Oset|=k$, $|V_A|\geq 2k$, and $|V_B|\geq k$; for otherwise, these properties can easily be obtained by adding dummy elements of value zero to $V_A$ and $V_B$ (and consequently, also to $W$), and complementing $\Oset$ with such elements to make sure that $|\Oset|=k$.
\begin{protocol}
\caption{Efficient algorithm for the two-player setting beating the approximation ratio $1/2$} \label{alg:simple2pAlg}
\textbf{Alice's algorithm}
\begin{algorithmic}[1]
\State Use \Greedy to select $2k$ elements $a_1,\ldots, a_{2k}\in V_A$, where the numbering corresponds to the order in which \Greedy selected the elements, and send them to Bob.
\end{algorithmic}
\textbf{Bob's algorithm}
\begin{algorithmic}[1]
\State For every $p\in \{0,\ldots, k\}$, use \Greedy to add $k-p$ elements $S_p$ of $V_B$ to the set $\{a_1,\ldots, a_p\}$, leading to a set $X_p = \{a_1,\ldots, a_p\}\cup S_p$ of cardinality $k$.

\State For every $p\in \{0,\ldots, k\}$, compute a set $Y_p$ in two steps. First, apply \Greedy to select $k-p$ elements $Q_p$ from $V_B$. Then, complement $Q_p$ using \Greedy to add $p$ elements from $\{a_1,\ldots, a_{2k}\}$ to obtain $Y_p$.

\State \Return{$\argmax\{f(S): S\in \{X_0,\ldots, X_k\}\cup \{Y_0,\ldots, Y_k\}\}$}.
\end{algorithmic}
\end{protocol}

Our main result here is that \cref{alg:simple2pAlg} has an approximation factor that is strictly better than $1/2$.
\begin{theorem}\label{thm:compRatio1TwoPl}
\cref{alg:simple2pAlg} is a $0.514$-approximation for the two-player setting.
\end{theorem}
%
%The proof of the above statement also allows for highlighting some key aspects why better than $0.5$-approximations can be achieved. We then discuss how a stronger approximation guarantee can be proven for the same algorithm, by exploring worst-case examples through well-chosen linear programs. This leads to the following guarantee, shown in Appendix~\ref{app:poly2players}.
%
%\begin{theorem}\label{thm:compRatio2TwoPl}
%Algorithm~\ref{alg:simple2pAlg} is $0.531$-approximation for the two-player setting.
%\end{theorem}
%
Our focus here is on highlighting some key arguments explaining why one can beat the factor of $1/2$. Thus, we keep our analysis simple instead of aiming for the best approximation ratio achievable.

\subsection{Proof of \texorpdfstring{\cref{thm:compRatio1TwoPl}}{Theorem~\ref*{thm:compRatio1TwoPl}}}

We define $k_A= |\mathcal{O}\cap V_A|$, $k_B = |\mathcal{O}\cap V_B|$, and let $\{a_1,a_2,\ldots, a_{2k}\}\subseteq V_A$ be the $2k$ elements computed by Alice, i.e., they are the first $2k$ elements chosen by \Greedy when maximizing $f$ over $V_A$, numbered according to the order in which they were chosen. For convenience, we define the following notation for prefixes of $\{a_1, a_2,\ldots, a_{2k}\}$. For any $i\in \{0,\ldots, 2k\}$, let
\begin{align*}
G_A^i &= \{a_1,\ldots, a_i\}\enspace,\text{ and}\\
G_A   &= G_A^{k_A}\enspace.
\end{align*}
In particular, $G_A^0=\varnothing$.
Finally, we denote by $G_B\subseteq V_B$ the set $Q_{k_A}$ computed by Bob, i.e., this is a \Greedy solution of the problem $\max\{f(S): S\subseteq V_B, |S|\leq k_B\}$.

To show \cref{thm:compRatio1TwoPl}, we prove that one of the two sets $X_{k_A}$ or $Y_{k_A}$ leads to the desired guarantee, i.e.,
\begin{equation}\label{eq:XorYisGood}
\max\{f(X_{k_A}), f(Y_{k_B})\} \geq 0.514 \cdot f(\Oset)\enspace.
\end{equation}

Without loss of generality, we assume that $f(\Oset\cap V_A) >0$ and $k_A > 0$. For otherwise, $f(\Oset) = f(\Oset\cap V_B)$ and the optimal value can be achieved with a set fully contained in Bob's elements $V_B$. Hence, by the approximation guarantee of \Greedy, we get $f(Y_{k_B}) \geq (1-e^{-1})\cdot f(\Oset)$, and~\eqref{eq:XorYisGood} clearly holds.

\medskip

A key quantity that we use in our analysis is the following.
\begin{equation*}
\Delta_A = \frac{f(\Oset\cap V_A \mid G_A)}{f(\Oset \cap V_A)}\enspace.
\end{equation*}
In words, $\Delta_A$ is a normalized way (normalized by $f(\Oset\cap V_A)$) to measure by how much the greedy solution $G_A$ would further improve when adding all elements of $\Oset\cap V_A$ to it. Notice that the monotonicity and submodularity of $f$ guarantee together $f(\Oset \cap V_A \mid G_A) \in [0, f(\Oset \cap V_A)]$, and thus, after the normalization, we get $\Delta_A\in [0,1]$.

\medskip

As we show in the following, we can exploit both small and large values of $\Delta_A$ to improve over the factor $1/2$. To build up intuition, consider first the following ostensibly natural candidate for a hard instance. Assume that
\begin{equation}\label{eq:optPartInAlice}
\max\left\{ f(S): S\subseteq V_A, |S|=k_A \right\}
\end{equation}
is a submodular maximization instance with maximizer $\Oset\cap V_A$ and the \Greedy solution, which is $G_A$, has value $f(G_A)$ very close to  $(1-e^{-1})\cdot f(\Oset\cap V_A)$, which is the worst-case guarantee for \Greedy.  By looking into the analysis of \Greedy, this happens only when $\Delta_A$ is very close to $e^{-1}$. However, it turns out that in this ostensibly bad case, our algorithm is even about $(1-e^{-1})$-approximate due to the following. The small value of $\Delta_A \approx e^{-1}$ implies that $G_A$ and $\Oset_A$ behave similarly in terms of how the submodular value changes when adding elements from $V_B$. This is important to make sure that Bob can complement $G_A$ to a strong solution through adding elements from $V_B$. More precisely, adding $\Oset\cap V_B$ to $G_A$ increases the submodular value by 
\begin{align*}
f(\Oset \cap V_B \mid G_A) &\geq f(\Oset\cap V_B \mid G_A \cup (\Oset\cap V_A))\\
&\geq f(\Oset) - f(G_A \cup (\Oset \cap V_A))\\
&= f(\Oset) - f(G_A) - f(\Oset\cap V_A \mid G_A)\\
&= f(\Oset) - f(G_A) - \Delta_A \cdot f(\Oset\cap V_A)\enspace,
\end{align*}
where the first inequality follows by submodularity and the second one by monotonicity of $f$. Finally, if we have, as discussed, that~\eqref{eq:optPartInAlice} is close to a worst-case instance in terms of approximability, then $f(G_A)\approx (1-e^{-1})\cdot f(\Oset\cap V_A)$ and $\Delta_A\approx e^{-1}$, and hence,
\begin{equation*}
f(\Oset \cap V_B \mid G_A) \gtrapprox f(\Oset) - f(\Oset\cap V_A)\enspace.
\end{equation*}
Thus, when augmenting $G_A$ with $k_B$ elements of $V_B$ through \Greedy, the increase $f(X_{k_A}\mid G_A)$ in submodular value is at least around $(1-e^{-1})\cdot (f(\Oset)-f(\Oset\cap V_A))$. Together with the fact that $f(G_A)\approx (1-e^{-1})\cdot f(\Oset\cap V_A)$, this implies $f(X_{k_A}) \gtrapprox (1-e^{-1})\cdot f(\Oset)$. Hence, our protocol computed a set $X_{k_A}$ that is even close to a $(1-e^{-1})$-approximation for this case.

The above example highlights that small values for $\Delta_A$---and this includes the worst-case value of $\Delta_A=e^{-1}$ for a classical submodular maximization problem with a cardinality constraint---allow Bob to complement $G_A$ in a strong way. To complete the above intuitive reasoning to a full formal proof, we proceed as follows. We first quantify in \cref{lem:greedyDelta1} the performance of \Greedy on Alice's side depending on the parameter $\Delta_A$. This will in particular imply that $\Delta_A \approx e^{-1}$ is indeed worst-case if the task is for Alice to select $k_A$ elements of highest submodular value. It also quantifies how \Greedy, run on Alice's side, improves for values of $\Delta_A$ bounded away from $e^{-1}$.

We then generalize and formalize in \cref{lem:Xguar} the above discussion, done for $\Delta_A\approx e^{-1}$, to arbitrary $\Delta_A \in [0,1]$. This shows that our protocol has a good approximation guarantee whenever $\Delta_A$ is small, and also covers the case when $f(\Oset \cap V_A)$ is small. Finally, \cref{lem:Yguar} covers the case when both $\Delta_A$ and $f(\Oset\cap V_A)$ are large. In this case, the set $G_A\cup (\Oset\cap V_A)$---which may have up to $2k_A$ elements and is thus not necessarily feasible---has large submodular value. This is useful to show that $f(Y_{k_A})$ is large due to the following. Recall that $Y_{k_A}$ is constructed by first applying \Greedy to Bob's elements to select $k_B$ elements $Q_p$ and then complement $Q_p$ with elements from $\{a_1,\ldots, a_{2k}\}$. Knowing that $G_A \cup (\Oset\cap V_A)$ has large submodular value implies a significant increase in submodular value when adding to $Q_p$ the highest-valued elements of $\{a_1,\ldots, a_{2k}\}$. Notice that this reasoning relies on the fact that Alice considers sets of cardinality larger than $k$.

\begin{lemma}\label{lem:greedyDelta1}
\begin{equation*}
f(G_A) \geq 
\left(1+\Delta_A\ln\Delta_A\right) \cdot f(\Oset\cap V_A)
\enspace,
\end{equation*}
where, for $\Delta_A=0$, we interpret $\Delta_A \ln\Delta_A =0$.
\end{lemma}
\begin{proof}
We first observe that the result clearly holds for $\Delta_A \leq e^{-1}$. Indeed, in this case we have
\begin{align*}
(1+\Delta_A\ln\Delta_A)\cdot f(\Oset\cap V_A) &\leq (1-\Delta_A)\cdot f(\Oset\cap V_A)\\
&= f(\Oset\cap V_A) - f(\Oset\cap V_A \mid G_A)\\
&= f(\Oset\cap V_A) - f(\Oset\cap V_A \cup G_A) + f(G_A)\\
&\leq f(G_A)\enspace,
\end{align*}
where the last inequality uses the monotonicity of $f$. Thus, from now on we assume $\Delta_A > e^{-1}$.

Recall that the classical analysis of the greedy algorithm (see~\eqref{eq:greedyvalue}) shows that if \Greedy is used to select $p\in \mathbb{Z}_{\geq 0}$ elements from $V_A$, then a set $G_A^p$ is obtained that satisfies
\begin{equation}\label{eq:guarGreedy}
f(G_A^p) \geq \left(1-\left(1-\frac{1}{k_A}\right)^{p}\right)\cdot f(\OAset) 
\geq 
\left(1-\left(1-\frac{1}{k_A}\right)^{p}\right)\cdot f(\Oset\cap V_A)
\enspace,
\end{equation}
where
\begin{equation*}%\label{eq:defWA}
\OAset \in \argmax\left\{f(S) : S\subseteq V_A: |S|\leq k_A\right\}\enspace.
\end{equation*}
Moreover, because \Greedy successively adds the element with largest marginal return, we have 
\begin{equation}\label{eq:greedyHighestMarginal}
f(a_i \mid \{a_1, \ldots, a_{i-1}\}) \geq f(a \mid \{a_1,\ldots, a_{i-1}\}) \qquad
\forall i\in [k_A],
\forall a\in V_A\enspace.
\end{equation}
Using the above observation and the submodularity of $f$, we obtain for every $i\in [k_A]$
\begin{align*}
f(a_i \mid \{a_1,\ldots, a_{i-1}\}) &\geq
\frac{1}{k_A} \cdot \sum_{a\in \Oset\cap V_A} f(a \mid \{a_1,\ldots, a_{i-1}\})\\
 &\geq \frac{1}{k_A} \cdot f(\Oset\cap V_A \mid \{a_1,\ldots, a_{i-1}\})\\
 &\geq \frac{1}{k_A}\cdot f(\Oset\cap V_A \mid G_A)\\
 &= \Delta_A \cdot \frac{f(\Oset\cap V_A)}{k_A}\enspace,
\end{align*}
where the first inequality is due to~\eqref{eq:greedyHighestMarginal}, and the other two follow from the submodularity of $f$ and the fact that $\{a_1, \ldots, a_{i-1}\}\subseteq G_A$.

Therefore, we can lower bound $f(G_A)$ by considering the partial solution $G_A^p$, and assuming that all of the other $k_A-p$ elements $\{a_{p+1},\dots, a_{k_A}\}$ each had a marginal contribution of $\Delta_A\cdot f(\Oset\cap V_A)/k_A$, which leads to the following lower bound:
\begin{align}
f(G_A) &\geq \max\left\{
f(G_A^p) + (k_A-p)\cdot \Delta_A\cdot \frac{f(\Oset \cap V_A)}{k_A}
\;\middle\vert\; p\in \{0,\ldots, k_A\}\right\}\notag\\
&\geq f(\Oset \cap V_A)\cdot \max\left\{
1-\left(1-\frac{1}{k_A}\right)^{p} + \frac{k_A-p}{k_A}\cdot\Delta_A \;\middle\vert\; p\in \{0,\ldots, k_A\}
\right\}\enspace,\label{eq:lowerBoundGADisc}
\end{align}
where the second inequality follows from~\eqref{eq:guarGreedy}.
To get a clean and easy way of evaluating the maximum in~\eqref{eq:lowerBoundGADisc}, we would like to allow $p$ to take continuous values within $[0,k_A]$, in which case we could set the derivative to zero to obtain the maximum. However, because we need a lower bound for this maximum, just lifting the integrality requirements on $p$ does not work out, as it would lead to a larger value for the maximum. To work around this problem, we define the function $\eta\colon [0,k_A] \to \mathbb{R}_{\geq 0}$ such that $\eta(p)$ for $p\in [0,k_A]$ is the piece-wise linear interpolation of the values $(1-\frac{1}{k_A})^p$ for $p\in \{0,\ldots, k_A\}$. Formally, 
\begin{align*}
\eta(p) =   (\lfloor p + 1 \rfloor - p)    \left(1-\frac{1}{k_A}\right)^{\lfloor p\rfloor}
       + (p - \lfloor p \rfloor)  \left(1-\frac{1}{k_A}\right)^{\lceil p \rceil}\qquad \forall p\in [0,k_A]\enspace.
\end{align*}
Now, we can replace the maximum in~\eqref{eq:lowerBoundGADisc} with a continuous version using $\eta$ as follows.
\begin{equation}\label{eq:makeMaxContinuous}
\begin{aligned}
\frac{f(G_A)}{f(\Oset \cap V_A)} \geq \max\bigg\{
1-\left(1-\frac{1}{k_A}\right)^{p} &+ \frac{k_A-p}{k_A}\cdot\Delta_A \;\bigg\vert\; p\in \{0,\ldots, k_A\}
\bigg\}\\
&= \max\left\{
1- \eta(p) + \frac{k_A-p}{k_A}\cdot\Delta_A \;\middle\vert\; p\in \{0,\ldots, k_A\}
\right\}\\
&= \max\left\{
1-\eta(x) + \frac{k_A - x}{k_A} \cdot \Delta_A \;\middle\vert\; x \in [0,k_A]\right\}\enspace,
\end{aligned}
\end{equation}
where the first equality follows because $\eta(p) = (1-\frac{1}{k_A})^p$ for $p\in \{0,\ldots, k_A\}$, and the second equality follows from the fact that the expression we maximize over is piece-wise linear with all break-points contained in $\{0,\ldots, k_A\}$; hence, there is a maximizer within $\{0,\ldots, k_A\}$.
The following claim, which we show in \cref{app:poly2players}, allows for replacing $\eta$ by a smooth approximation.

\begin{restatable}{claim}{claimUpperBoundEta} \label{claim:upperBoundEta}
 $\eta(x) \leq e^{-\frac{x}{k_A}}$ for all $x\in [0,k_A]$.
\end{restatable}

Using the claim, we can now further expand~\eqref{eq:makeMaxContinuous} to obtain
\begin{align*}
\frac{f(G_A)}{f(\Oset\cap V_A)}
&\geq \max\left\{
1-\eta(x) + \frac{k_A-x}{k_A}\cdot\Delta_A \;\middle\vert\; x\in [0,k_A]
\right\}\\
&\geq \max\left\{
1- e^{-\frac{x}{k_A}} + \frac{k_A - x}{k_A} \cdot \Delta_A \;\middle\vert\; x \in [0,k_A]\right\}\\
&=1+\Delta_A\ln\Delta_A\enspace, \label{eq:lowerBoundGACont}
\end{align*}
where the equality holds by observing that the maximum of the concave function $1-e^{-\frac{x}{k_A}} + \frac{k_A-x}{k_A}\cdot\Delta_A$ over $x\in \mathbb{R}$ is achieved for $x=-k_A\ln\Delta_A$. (Notice that $-k_A\ln\Delta_A\in [0,k_A]$ because we assumed $\Delta_A > e^{-1}$.) This proves the result.
\end{proof}

We now show how the intuitive discussion for the case $\Delta_A \approx e^{-1}$ can be made formal and be generalized to arbitrary $\Delta_A\in [0,1]$.
\begin{lemma}\label{lem:Xguar}
\begin{align*}
f(X_{k_A}) &\geq
 (1-e^{-1})\cdot f(\Oset) + e^{-1} f(G_A) - (1-e^{-1})\cdot \Delta_A\cdot f(\Oset\cap V_A)\\
&\geq \left(1-e^{-1}\right)\cdot f(\Oset) -\left(
(1-e^{-1})\Delta_A - e^{-1} - e^{-1}\Delta_A\ln\Delta_A
\right)\cdot f(\Oset\cap V_A)\enspace,
\end{align*}
where, as usual, we interpret $\Delta_A \ln\Delta_A = 0$ for $\Delta_A=0$.
\end{lemma}
\begin{proof}
We recall that $X_{k_A} = G_A \cup S_{k_A}$, where $G_A\subseteq V_A$ is a set constructed by Alice and $S_{k_A}$ is a set constructed by Bob as described in \cref{alg:simple2pAlg}. Observe that the following inequality holds.
\begin{equation}\label{eq:basicLBForXkA}
f(X_{k_A}) \geq f(G_A) + (1-e^{-1})\cdot f(\Oset\cap V_B \mid G_A)\enspace.
\end{equation}
Indeed, Bob extends the solution $G_A$ by adding the $k_B$ elements $S_{k_A}$ from $V_B$ using \Greedy. Hence, this corresponds to applying \Greedy to the submodular function maximization problem
\begin{equation*}
\max\left\{f(S \mid  G_A) : S\subseteq V_B, |S|\leq k_B \right\}\enspace.
\end{equation*}
Because \Greedy is a $(1-e^{-1})$-approximation, and $\Oset\cap V_B$ is a feasible solution to the above problem, we have
\begin{equation*}
f(S_{k_A} \mid G_A) \geq \left(1-e^{-1}\right)\cdot f(\Oset \cap V_B \mid G_A)\enspace,
\end{equation*}
thus implying~\eqref{eq:basicLBForXkA} because $f(X_{k_A}) = f(G_A) + f(S_{k_A} \mid G_A)$.

The result now follows due to the following chain of inequalities:
\begin{align*}
f(X_{k_A}) &\geq f(G_A) + \left(1-e^{-1}\right)\cdot f(\Oset\cap V_B \mid G_A)\\
&\geq f(G_A) + (1-e^{-1})\cdot f(\Oset\cap V_B \mid G_A \cup (\Oset\cap V_A))\\
&= f(G_A) + (1-e^{-1})\cdot \left[ f(\Oset\cup G_A) - f(G_A\cup (\Oset\cap V_A)) \right]\\
&\geq f(G_A) + (1-e^{-1})\cdot \left(
f(\Oset) - f(G_A) - \Delta_A\cdot f(\Oset\cap V_A)
\right)\\
&= (1-e^{-1})\cdot f(\Oset) + e^{-1} f(G_A) - (1-e^{-1})\cdot \Delta_A\cdot f(\Oset\cap V_A)\\
&\geq (1-e^{-1})\cdot f(\Oset) - ((1 - e^{-1})\Delta_A - e^{-1} - e^{-1}\Delta_A\ln\Delta_A)\cdot f(\Oset\cap V_A)\enspace,
\end{align*}
where the first inequality comes from~\eqref{eq:basicLBForXkA}, the second one uses the submodularity of $f$, the third inequality uses the definition of $\Delta_A = \frac{f(\Oset\cap V_A \mid G_A)}{f(\Oset \cap V_A)}$ and the monotonicity of $f$, and the last inequality is implied by \cref{lem:greedyDelta1}.
\end{proof}

We now show that if both $\Delta_A$ and $f(\Oset\cap V_A)$ are large, then there are (possibly infeasible) sets on Alice's side of very large submodular value, which can be exploited to complement a greedy solution on Bob's side, leading to a set $Y_{k_A}$ with high submodular value.
\begin{lemma}\label{lem:Yguar}
\begin{equation}\label{eq:Yguar}
f(Y_{k_A}) \geq \frac{1}{2}(1-e^{-1})\cdot f(\Oset)
+ \frac{1}{2}\left(e^{-1} + \Delta_A\ln\Delta_A + (1-e^{-1})\Delta_A\right)\cdot f(\Oset\cap V_A)\enspace,
\end{equation}
where, for $\Delta_A=0$, we interpret $\Delta_A\ln\Delta_A$ as zero.
\end{lemma}
\begin{proof}
First, observe that the set $G_A^{2k_A}$ is obtained by adding $k_A$ elements from $V_A\setminus G_A$ to $G_A$ using \Greedy. Since \Greedy has an approximation guarantee of $1-e^{-1}$ and the set $\Oset\cap V_A$ is a set of cardinality $k_A$, the $k_A$ elements added by \Greedy to $G_A$ increase the submodular value of the set by at least $(1-e^{-1})\cdot f(\Oset\cap V_A \mid G_A)$. Thus,
\begin{equation*}
f(G_A^{2k_A}) \geq f(G_A) + (1-e^{-1})\cdot f(\Oset\cap V_A \mid G_A)\enspace.
\end{equation*}
Together with \cref{lem:greedyDelta1}, the above inequality implies
\begin{equation}\label{eq:lowBoundGA2}
f(G_A^{2k}) \geq (1+\Delta_A\ln\Delta_A + (1-e^{-1})\Delta_A)\cdot f(\Oset\cap V_A)\enspace.
\end{equation}

We now observe that
\begin{equation}\label{eq:YkABasicLB}
f(Y_{k_A}) \geq f(G_B) + \frac{1}{2} f(G_A^{2k_A} \mid G_B)\enspace.
\end{equation}
Indeed, $Y_{k_A}$ was obtained by starting from $G_B$ and adding $k_A$ elements among the $2k_A$ elements of $G_A^{2k_A}$ to it using \Greedy. Adding all $2k_A$ elements of $G_A^{2k_A}$ to $G_B$ would have increased the submodular value by $f(G_A^{2k_A} \mid G_B)$. Adding half of the $2k_A$ elements of $G_A^{2k_A}$ to $G_B$ increases the value by at least half that much due to the following. When running \Greedy for $2k_A$ steps, each element added in the first half has a marginal return at least as large as each element added in the second half, because marginal returns of elements added by \Greedy are non-increasing. Hence, the first half has a total marginal increase at least as large as the second one, which implies~\eqref{eq:YkABasicLB}.

The statement now follows from the following chain of inequalities.
\begin{align*}
f(Y_{k_A}) &\geq f(G_B) + \frac{1}{2}f(G_A^{2k_A} \mid G_B)\\
&\geq f(G_B) + \frac{1}{2}\left[ (1+\Delta_A\ln\Delta_A +(1-e^{-1})\Delta_A) f(\Oset \cap V_A) - f(G_B) \right]\\
&\geq \frac{1}{2} (1-e^{-1}) \cdot f(\Oset\cap V_B) + \frac{1}{2}\left(
1 + \Delta_A\ln\Delta_A + (1-e^{-1})\Delta_A
\right) \cdot f(\Oset\cap V_A)\\
&\geq \frac{1}{2} (1-e^{-1}) \cdot f(\Oset) + \frac{1}{2} \left(
e^{-1} + \Delta_A \ln\Delta_A + (1-e^{-1})\Delta_A
\right)\cdot f(\Oset\cap V_A)\enspace,
\end{align*}
where the first inequality is due to~\eqref{eq:YkABasicLB} , the second one follows from~\eqref{eq:lowBoundGA2} and the monotonicity of $f$, the third one uses the fact that $G_B$ was obtained by \Greedy, and therefore, $f(G_B) \geq (1-e^{-1})\cdot f(\Oset\cap V_B)$, and the last one follows from $f(\Oset\cap V_A) + f(\Oset\cap V_B) \geq f(\Oset)$.
\end{proof}

Finally, the approximation factor claimed by \cref{thm:compRatio1TwoPl} is obtained by combining the lower bounds provided by \cref{lem:Xguar,lem:Yguar} to bound $\max\{f(X_{k_A}),f(Y_{k_A})\}/f(\Oset)$. To this end, we compute the worst-case value of the two lower bounds for all possibilities of $\Delta_A\in [0,1]$ and $\frac{f(\Oset\cap V_A)}{f(\Oset)} \in [0,1]$. This can be captured through the following nonlinear optimization problem, where, for brevity, we use $x$ for $\Delta_A$ and $y$ for $f(\Oset\cap V_A)/f(\Oset)$.
\begin{equation}\label{eq:optProbFor2PLB}
\begingroup
\renewcommand\arraystretch{1.2}
\begin{array}{rr@{\;\;}c@{\;\;}l}
\min & z & &\\
     & z &\geq &1 - e^{-1} - \left[ (1-e^{-1})x - e^{-1} - e^{-1} x \ln x \right]\cdot y \\
     & z &\geq &\frac{1}{2} (1-e^{-1}) + \frac{1}{2}\left[ e^{-1} + x\ln x + (1-e^{-1}) x \right]\cdot y\\
& x,y &\in &[0,1]\\
& z   &\in &\mathbb{R}
\end{array}
\endgroup
\end{equation}
Thus, the optimal value of~\eqref{eq:optProbFor2PLB} is a lower bound on the approximation ratio of \cref{alg:simple2pAlg}. Together with the following statement, this completes the proof of \cref{thm:compRatio1TwoPl}.

\begin{restatable}{lemma}{lemboundOptProbFor2PLP}\label{lem:boundOptProbFor2PLP}
The optimal value $\alpha$ of Problem~\eqref{eq:optProbFor2PLB} satisfies $\alpha \geq 0.514$.
\end{restatable}

One easy way to do a quick sanity check of \cref{lem:boundOptProbFor2PLP} is by solving~\eqref{eq:optProbFor2PLB} via standard numerical optimization methods, which is not difficult due to the fact that the problem has only $3$ variables and only $2$ non-trivial constraints (which are smooth). Nevertheless, we formally prove \cref{lem:boundOptProbFor2PLP} by providing an analytical description of the unique optimal solution of~\eqref{eq:optProbFor2PLB} through the use of the necessary Karush-Kuhn-Tucker optimality conditions. Because this is a standard approach for deriving optima, we defer the proof of \cref{lem:boundOptProbFor2PLP} to \cref{app:poly2players}.

%% file: A5060-applications.tex
\section{Formal Connections to the Applications}

In this section we give the proofs of the theorems from \cref{sec:intro} showing the formal connections between \maxcard and the applications we show for it.
\input{A50-hardness-streaming.tex}
\input{A60-robustness.tex}

%% file: A50-hardness-streaming.tex
\subsection{Proof of \texorpdfstring{\cref{thm:streaming-hardness}}{Theorem~\ref*{thm:streaming-hardness}}}\label{app:streaming-hardness}
\thmstreaminghardness*
\begin{proof}
%The following arguments are standard and included for completeness. 
    Let $\cA$ be a data stream algorithm for \maxcard with an approximation guarantee of $1/2+\varepsilon$, and consider the following $p$-player protocol. 
    The first player feeds to $\cA$ the elements of her input $V_1 \subseteq W$ in any order. 
    She then sends the memory state of $\cA$ to the second player, who feeds her own elements $V_2 \subseteq W$ to $\cA$ before forwarding the resulting memory state of $\cA$ to the next player, and so on. 
    The last player finally feeds $V_p \subseteq W$ to $\cA$ and outputs the same set as $\cA$. It is clear that the output of the last player is the same as that of running $\cA$ on a stream consisting of elements $V_1 \cup V_2 \cup \cdots \cup V_p$.   
    Therefore, for any $p\geq 2$ the protocol satisfies that (i)  its approximation guarantee  is $1/2+\varepsilon$ and, by definition, (ii) the size of any message sent by the players is at most the memory usage of $\cA$.
    Now, by selecting $p=k$ large enough as a function of $\varepsilon$,   \cref{thm:intro_gen_hardness} implies that the memory usage of $\cA$  must be $\Omega\rb{\varepsilon N/ p^3} = \Omega\rb{\varepsilon s/k^3}$, where we used the equality $p=k$ and the fact that $s = |V_1| + |V_2| + \ldots +  |V_p|$ is at most  $ |W| = N$ since the sets $V_1, V_2, \ldots, V_p$ are disjoint.
\end{proof}

%% file: A60-robustness.tex
\subsection{Proof of \texorpdfstring{\cref{thm:robust_reduction_simplified}}{Theorem~\ref*{thm:robust_reduction}}} \label{app:robustness}

We begin this section we a formal restatement of \cref{thm:robust_reduction_simplified}.

\let\oldthetheorem\thetheorem
\renewcommand{\thetheorem}{\getrefnumber{thm:robust_reduction_simplified}}
\begin{theorem} \label{thm:robust_reduction}
Assume we are given a protocol $\cP$ for \maxcard in the two-player model which has the properties that
\begin{enumerate}[label=(\roman*),itemsep=0em]
	\item Alice does not access in any way elements that do not belong to the set $V_A$. In particular, she does not query $f$ on subsets including such elements and neither includes such elements in the message to Bob.
	\item Let $M$ be the set of elements that explicitly appear in the message of Alice. Then, Bob does not access in any way elements that do not belong to the set $M \cup V_B$. In particular, he does not query $f$ on subsets including such elements and neither includes such elements in the output set he generates.
\end{enumerate}
Then, there exists an algorithm $\cA$ for robust submodular maximization whose approximation guarantee is at least as good as the approximation guarantee of $\cP$, and the number of elements in the summary of $\cA$ is $O(d \cdot \text{\{communication complexity of $\cP$ in elements\}})$. Furthermore, if $\cP$ runs in polynomial time, then so does $\cA$.
\end{theorem}
\let\thetheorem\oldthetheorem \addtocounter{theorem}{-1}

Observe that the properties required from $\cP$ by \cref{thm:robust_reduction} imply that $\cP$ completely ignores the sets $W$, $W_A$, and $W_B$ even though these sets are part of the global information in the formal description of the two-player model in \cref{sec:prelim_model}. Thus, the algorithm $\cA$ we design may use $\cP$ without specifying these sets. We give below, as \cref{alg:robustness}, the algorithm $\cA$. The design of this algorithm is based on a technique of~\cite{pmlr-v70-mirzasoleiman17a}. In a nutshell, the algorithm uses $d + 1$ independent copies of the protocol $\cP$ and manages to guarantee that one of them gets exactly the elements that have not been deleted.

\begin{algorithm}
\caption{Algorithm $\cA$ of \cref{thm:robust_reduction}} \label{alg:robustness}
\textbf{Summary Procedure}
\begin{algorithmic}[1]
\State Let $V$ be the ground set.
\For{$i = 1$ to $d + 1$} 
    \State Initialize a new independent copy $\cP_i$ of $\cP$.
    \State Pass $V \setminus (\bigcup_{j = 1}^{i - 1} M_j)$ as the input $V_A$ of Alice of $\cP_i$.
    \State Let $M_i$ be the set of elements explicitly appearing in the message sent by Alice in $\cP_i$.
\EndFor
\State The summary is the set of the messages sent by the Alices of all the protocols $\cP_1, \cP_2, \dotsc, \cP_{d + 1}$.
\end{algorithmic}
\textbf{Query Procedure}
\begin{algorithmic}[1]
\State Let $D$ be the set of elements that have been deleted.
\State Let $\ell$ be a value in $[d + 1]$ such that $M_{\ell} \cap D = \varnothing$ (we prove in the main text that such a value exists).
\State Pass the set $(\bigcup_{j = 1}^{\ell - 1} M_j) \setminus D$ as the input $V_B$ to Bob of $\cP_\ell$.\\
\Return{The set produced by Bob of $\cP_\ell$}.
\end{algorithmic}
\end{algorithm}

We begin the analysis of \cref{alg:robustness} with the following technical observation.
\begin{observation}
The sets $M_1, M_2, \dotsc, M_{d + 1}$ are pairwise disjoint.
\end{observation}
\begin{proof}
For every $i \in [d + 1]$, the set $M_i$ includes only elements that appear in the message generated by Alice of $\cP_i$. By the assumptions on $\cP$ in \cref{thm:robust_reduction}, $M_i$ must be a subset of the set $V_A$ passed to Alice of $\cP_i$, which implies
\[
    M_i \subseteq V \setminus \left(\bigcup_{j = 1}^{i - 1} M_j\right)
    \implies
    M_i \cap M_j = \varnothing \quad \forall\; j \in [i - 1]
    \enspace.
    \qedhere
\]
\end{proof}
\begin{corollary}
At least one of the sets $M_1, M_2, \dotsc, M_{d + 1}$ is disjoint from $D$ because $|D| = d$.
\end{corollary}

The last corollary implies that \cref{alg:robustness} can indeed find a value $\ell$ with the promised properties, and therefore, it is well defined. One can also observe that the summary produced by \cref{alg:robustness} consists of $d + 1$ messages of $\cP$, and thus, the number of elements in this summary is $d + 1 = O(d)$ times the communication complexity in elements of $\cP$. Finally, it is clear from the description of \cref{alg:robustness} that this algorithm runs in polynomial time when the algorithms of Alice and Bob in $\cP$ run in polynomial time. Hence, to prove \cref{thm:robust_reduction}, it only remains to argue that \cref{alg:robustness} always outputs a feasible solution and that it inherits the approximation guarantee of $\cP$.

From this point on, we denote by $\widehat{S}$ the output set of \cref{alg:robustness} and by $V_A$ and $V_B$ the input sets of Alice and Bob of $\cP_\ell$, respectively. Using this notation, we get the following observation.
\begin{observation}
$\widehat{S} \subseteq V \setminus D$, and therefore, \cref{alg:robustness} is guaranteed to produce a feasible solution.
\end{observation}
\begin{proof}
By the assumptions about $\cP$ in \cref{thm:robust_reduction}, the output set of $\cP_\ell$ is a subset of $M_\ell \cup V_B$. Since this output set becomes the output set $\widehat{S}$ of \cref{alg:robustness}, we get
\[
    \widehat{S}
    \subseteq
    M_\ell \cup V_B
    =
    M_\ell \cup \left[\left(\bigcup_{j = 1}^{\ell - 1} M_j\right) \setminus D\right]
    \subseteq
    V \setminus D
    \enspace,
\]
where the last inclusion holds because $M_\ell$ does not include any elements of $D$ by the choice of $\ell$.
\end{proof}

We now get to analyzing the approximation ratio of \cref{alg:robustness}.
\begin{lemma}
Let $\alpha$ denote the approximation guarantee of $\cP$, then the approximation guarantee of \cref{alg:robustness} is at least $\alpha$.
\end{lemma}
\begin{proof}
We begin by observing that
\[
    V_A \cup V_B
    =
    \left[V \setminus \left(\bigcup_{j = 1}^{\ell - 1} M_j\right)\right] \cup \left[\left(\bigcup_{j = 1}^{\ell - 1} M_j\right) \setminus D\right]
    =
    V \setminus \left[D \cap \left(\bigcup_{j = 1}^{\ell - 1} M_j\right)\right]
    \supseteq
    V \setminus D
    \enspace.
\]
Thus, by the approximation guarantee of $\cP$,
\[
    f(\hat{S})
    \geq
    \alpha \cdot \max_{\substack{S \subseteq V_A \cup V_B \\ |S| \leq k}} \mspace{-9mu} f(S)
    \geq
    \alpha \cdot \max_{\substack{S \subseteq V \setminus D \\ |S| \leq k}} \mspace{-9mu} f(S)
    \enspace.
    \qedhere
\]
\end{proof}

%% file: A40-alternative-model.tex
\section{A Sketch of an Alternative Model} \label{app:alternative_model}

As mentioned in \cref{sec:prelim_model}, there are multiple natural models that can be used to formulate our problem. In this section we sketch one such model which appears quite different from the model used throughout the rest of the paper. Nevertheless, our results can be extended to this model (up to minor changes in the proved bounds). For brevity, we present the model directly for the $p$-player case instead of presenting first its two-player version.

The global information in an instance of this model consists of a ground set $W$, a non-negative set function $f\colon 2^W \to \nnR$ and a partition of $W$ into $p$ disjoint sets $W_1, W_2, \dotsc, W_p$, where as usual, one should think of $W_i$ as the set of elements player number $i$ might get. In addition, each player $i \in [p]$ has access to private information consisting of a set $V_i \subseteq W_i$ of elements that this player actually gets. It is also guaranteed that $f$ is a monotone and submodular function when restricted to the domain $2^{\bigcup_{i = 1}^p V_i}$. Like in our regular model, the objective of the players in this model is to find a set $S \subseteq \bigcup_{i = 1}^p V_i$ of size $k$ maximizing $f$ among all such sets; and they can do that via a one-way communication protocol in which player $1$ sends a message to player $2$, player $2$ sends a message to player $3$ and so on, until player $p$ gets a message from player $p - 1$, and based on this message produces the output set.

There are two main differences between this model and our regular model.
\begin{itemize}
    \item In our regular model, the objective function is part of the private information, and each player $i \in [p]$ has access only to the restriction of this function to $\bigcup_{j = 1}^i W_i$. In contrast, in the model suggested here, the objective function is a global information accessible to all players, which simplifies the model and creates a nice symmetry between the players.
    \item Since the objective function is now available to all the players from the very beginning, to avoid leaking information to the early players about the parts of the instance that should be revealed only to the later players, the sets $W_i$ corresponding to these late players should include many elements that might end up in $V_i$ in different scenarios. In particular, since $W_i$ is large, many of the elements in $W_i$ might never end up in $V_i$ together, and thus, there is no reason to require all these elements to form together one large submodular function. Accordingly, the model does not require $f$ to be monotone and submodular on all of $W$. Instead, it only requires $f$ to be monotone and submodular over the elements that really arrive to some player.
\end{itemize}

%% file: A20-kIndex.tex
\section{Hardness of the \texorpdfstring{\chainP}{CHAINp} Problem}
\label{app:pindex}

For completeness, we adapt  one of the many hardness proofs for the standard INDEX problem  to get the mentioned hardness result for the \chainP. See \cref{sec:index-model} for the definition of the \chainP problem. For convenience, we recall first the theorem statement.
\thmpindex*
The proof that we present is both based on insightful discussions with Michael Kapralov and on lecture notes from his course ``Sublinear Algorithms for Big Data Analysis''. We start with some preliminaries. We then introduce a distribution of instances $D^p$ and,  finally, we show that no protocol can have a ``good'' success probability on instances sampled from  $D^p$ without having a ``large'' communication complexity. Throughout this section, we  simply refer to the \chainPn problem as the \chainP problem because $n$ will be clear from context.

\paragraph{Basic notation and Fano's inequality.}
For discrete random variables $X,Y$, and $Z$, define
\begin{itemize}
    \item the entropy $H(X) = \sum_{x}  \Pr[X = x] \log_2(1/\Pr[X=x])$;
    \item the conditional entropy $H(X \mid Y) = \sum_y \Pr[Y= y] H(X \mid Y= y) = \E_Y H(X \mid Y = y)$; 
    \item the mutual information $I(X; Y) = H(X)- H(X\mid Y)$; and 
    \item the conditional mutual information $I(X;Y \mid Z) = H(X \mid Z) - H(X \mid Y, Z)$.
\end{itemize}  
We also use the following well-known relations:
\begin{itemize}
    \item symmetry of mutual information: $I(X;Y) = I(Y;X)$; and
    \item chain rule for mutual information: $I(X; Y, Z) = I(X;Z) + I(X;Y \mid Z)$.
\end{itemize}
The proof relies on Fano's inequality.
\begin{theorem}
    Let $X$ and $Y$ be discrete random variables and $g$ an estimator (based on $Y$) of $X$ such that 
    $\Pr[g(Y) \neq  X] \leq \delta$ and $g(Y)$ only takes values in the support $\supp(X)$ of $X$. Then,
    \begin{align*}
        H(X \mid Y) \leq H_2(\delta)  + \delta\log_2(|\supp(X)| -1)\enspace,
    \end{align*}
    where $H_2(\delta) = \delta \log_2(1/\delta) + (1-\delta) \log_2(1/(1-\delta))$ is the binary entropy at $\delta$.
\end{theorem}

\paragraph{Distribution $D^p$ of \chainP instances.}

Our hardness uses Yao's minimax principle. Specifically, we give a distribution over \chainP instances so that any deterministic protocol with a ``good'' success probability must have  a ``large'' communication complexity.
To define a \chainP instance, we use the superscript to indicate the input of each player. For example, $x^i \in \{0,1\}^n, t^i \in [n]$ denotes the $n$-bit string $x^i$ and index $t^i$ given to the $i$-th  player.   A \chainP instance is thus defined by $x^1, t^2, x^2, t^3, \ldots, x^{p-1}, t^p$ where $x^i \in \{0,1\}^n$ and $t^{i+1} \in [n]$ for $i\in [p-1]$. 

We now define a distribution $D^p$ over such instances and let $X^i, T^{i+1}$ be the discrete random variables corresponding to $x^i$ and $t^{i+1}$, respectively. 
For $z\in \{0, 1\}$, let $B_z = \{ (x, t) \in \{0,1\}^n \times [n] : x_t = z\}$. The joint distribution $D^p$ over the 
the discrete random variables $X^1, T^2, X^2, T^3,  \ldots, X^{p-1}, T^p$ is now defined by the following sampling procedure.
\begin{itemize}[itemsep=0.1em, parsep=0em, topsep=0.2em]
    \item Select $z \in \{0,1\}$ uniformly at random.
    \item For $i \in [p-1]$ select $(x^i, t^{i+1}) \in B_z$ uniformly at random.
    \item Output  $x^1, t^2, \ldots, x^{p-1}, t^p$. 
\end{itemize}
In other words, $\Pr[X^1 = x^1, T^2 = t^2 ,\ldots, X^{p-1}= x^{p-1}, T^p = t^p]$ equals the probability that the above procedure outputs $x^1,t^2, \ldots, x^{p-1}, t^p$.

Note that an alternative equivalent procedure for obtaining a sample from $D^p$ is
    \begin{itemize}[itemsep=0.1em, parsep=0em, topsep=0.2em]
        \item Select $x^1 \in \{0,1\}^n$ uniformly at random.
        \item Select $t^2 \in [n]$ uniformly at random.
        \item Set $z = x^1_{t^2}$.
        \item For $i = 2, \ldots, p-1$ select $(x^i, t^{i+1}) \in B_z$ uniformly at random.
        \item Output $x^1, t^2, \ldots, x^{p-1}, t^p$. 
    \end{itemize}
This immediately implies the following observation:
\begin{observation}
    The random variables $X^1, T^2$ are equivalently distributed in $D^p$  as in $D^2$ for any $p\geq 2$, i.e., they are uniformly distributed in $\{0,1\}^n \times [n]$.
    \label{obs:pindex_uniform}
\end{observation}
We also have the following 
\begin{lemma}
    Let $M$ be a function of $X^1$ and $m$ one of its possible values. Further fix some $t^2 \in [n]$. Then the total variation distance between the distribution of $X^2, T^3, \ldots, X^{p-1}, T^p$ conditioned on $M = m, T^2 = t^2$ and the unconditional distribution of $X^2, T^3, \ldots, X^{p-1}, T^p $ equals
    \begin{align*}
        \left| 1/2 - \Pr[X^1_{t^2} = 0 \mid M = m] \right|\enspace.
    \end{align*}
    \label{lem:pindex_tvd}
\end{lemma}
\begin{proof}

Consider the aforementioned alternative sampling procedure for $D^p$ where we first sample $x^1, t^2$, and then let $z = x^1_{t^2}$.
    As $M$ is a function of $X^1$,  we have that the distribution of $X^2, T^3, \ldots, X^{p-1}, T^p$ conditioned on $M = m$ and $T^2 = t^2$ can be defined by the following sampling procedure.
    \begin{itemize}[itemsep=0.1em, parsep=0em, topsep=0.2em]
        \item Select $x^1 \in \{0,1\}^n$ at random from the conditional distribution $M = m$.
%        \item Select $t^2 \in [n]$ uniformly at random.
        \item Set $z = x^1_{t^2}$.
        \item For $i = 2, \ldots, p-1$ select $(x^i, t^{i+1}) \in B_z$ uniformly at random.
        \item Output the outcome $x^2, t^3, \ldots, x^{p-1}, t^p$. 
    \end{itemize}
    In other words, $z$ is selected to be $0$ with probability $\Pr[X^1_{t^2} = 0  \mid M = m]$  and 
                                        $1$ with probability $\Pr[X^1_{t^2} = 1  \mid M = m]$. 
                                        
    The unconditional distribution is defined in the same way, except that $x^1$ and $t_2$ are selected uniformly at random. It follows that the total variation distance between the conditional distribution and unconditional distribution of $X^2, T^3, \ldots, X^{p-1}, T^p$ equals
    \begin{align*}
         \left| 1/2 - \Pr[X^1_{t^2} = 0 \mid M = m] \right|\enspace.
         \tag*{\qedhere}
    \end{align*}
\end{proof}

%%    We thus have that $t^2$ is selected uniformly at random  from $[n]$ even after conditioning on the outcome of  $x^{1}$.
\paragraph{Hardness proof of \chainP.}
We are now ready to prove the mentioned hardness result for the \chainP problem. As already noted, we show that 
any deterministic protocol with a ``good'' success probability over instances sampled from $D^p$ must have  a ``large'' communication complexity.
By Yao's minimax principle, this then implies \cref{thm:pindex_hardness}.

For a deterministic $p$-player protocol $\cP$, let $\mbox{Z}_{\cP}(x^1, t^2, \ldots, x^{p-1}, t^{p})$ denote the indicator function that $\cP$ outputs the correct prediction on instance $x^1, t^2, \ldots, x^{p-1}, t^{p}$. 
We start by analyzing the $2$-player case, which we later use as a building block in the more general $p$-player case.
\begin{lemma}
    Let $\delta \in (0, 1/2)$.
    Any protocol $\cP$ for the \chainOnlyP{2} problem with communication complexity at most $c = (1-H_2(\delta))\cdot n - 1$ satisfies
    \begin{align*}
       \Pr_{x^1, t^2 \sim D^2}[ \cP \mbox{ outputs correct prediction on instance $x^1, t^2$}] = \E[ \mbox{Z}_{\cP}(X^1, T^2)] \leq 1-\delta\enspace.
    \end{align*}
    \label{lem:index}
\end{lemma}
\begin{proof}
    We abbreviate $X^1$ by $X$ to simplify notation.
    The message $M$ that the first player, Alice, sends is a discrete random variable  that is a function of $X$. Note that
    \begin{align*}
        c + 1 \geq H(M) \geq I(M; X)
        \enspace
        .
    \end{align*}
    The first inequality holds because Alice sends at most $c$ bits and so the entropy of $M$ is at most $\log_2(2^c + 2^{c-1} + \ldots + 1) \leq c+1$. The second inequality follows from the definition of mutual information and the non-negativity of entropy: $I(M;X) = H(M) - H(M \mid X) \leq H(M)$.

    Letting, for any $i\in [n]$, $X_{<i}$ denote the vector $(X_1, \ldots, X_{i-1})$, we get
{\def\bw{48mm}
    \begin{align*}
        I(M; X) &=  \parbox{\bw}{$I(X;M)$} \mbox{(symmetry of mutual information)}   \\
        & = \parbox{\bw}{$\displaystyle\sum_{i=1}^n I(X_i; M \mid X_{<i})$} \mbox{(chain rule for mutual information)} \\
        & = \displaystyle\sum_{i=1}^n \left(H(X_i \mid X_{<i}) - H(X_i \mid M , X_{<i})\right) \\
        & \geq \parbox{\bw}{$\displaystyle\sum_{i=1}^n \left( H(X_i) - H(X_i \mid M)\right)$} \mbox{($X_i$'s are iid, and conditioning does not increase entropy)} \\
        & = \parbox{\bw}{$\displaystyle\sum_{i=1}^n \left( 1- H(X_i \mid M)\right)\enspace.$} \mbox{($X_i$ is a uniform binary random variable)} 
    \end{align*}}%
Now, if we let $1-\delta_i$ denote the probability that $\cP$ outputs the correct prediction when $t^2 = i$, then by Fano's  inequality  $H(X_i \mid M) \leq H_2(\delta_i)$, and hence,
    \begin{align*}
         \sum_{i=1}^n \left( 1- H(X_i \mid M)\right) \geq \sum_{i=1}^n \left( 1- H_2(\delta_i)\right)\enspace.
    \end{align*}
    Further, let $\delta' = \frac{1}{n}\sum_{i=1}^n \delta_i$. Note that $\delta'$ equals the probability that the protocol $\cP$ makes the incorrect prediction. Hence, all that remains to be shown is $\delta' \geq \delta$. To this end, first observe that if $\delta'\geq 1/2$, then this trivially holds because $\delta\in (0,1/2)$. Otherwise, observe that the concavity of the binary entropy function implies $\frac{1}{n} \sum_{i=1}^n H_2(\delta_i) \leq H_2(\delta')$, and hence
    \begin{align*}
        \sum_{i=1}^n \left( 1- H_2(\delta_i)\right) &= n - n\cdot \left( \frac{1}{n} \sum_{i=1}^n H_2(\delta_i)\right)  
         \geq n  \left(1- H_2(\delta')\right) \enspace.
    %    & \geq n  \left(1- H_2(\delta)\right)\enspace,
    \end{align*}
    Thus, $n(1-H_2(\delta)) - 1 = c \geq I(M; X) - 1 \geq n(1-H_2(\delta')) - 1$. This implies $\delta \leq \delta'$ since $n(1 - H_2(y))$ is a strictly decreasing function of $y$ within the range $[0, 1/2]$.
\end{proof}
We now use the above lemma and induction to prove the hardness result for \chainP. 

\begin{lemma}
    Let $\delta \in (0, 1/2)$ and $s = 1/2 - \delta$.
    For any integer $p\geq 2$, any protocol $\cP$ for the \chainP problem with communication complexity at most $(1- H_2(\delta)) \cdot n -  1$ satisfies 
    \begin{align*}
        \E[ Z_{\cP}(X^1, T^2, \ldots, X^{p-1}, T^p)]  \leq  1/2 + p\cdot s\enspace.
    \end{align*}
    \label{lem:pindex_hardness}
\end{lemma}
\begin{proof}
    We prove the statement by induction on $p$. For $p=2$, the statement is implied by \cref{lem:index}. Now assume it holds for $2, \ldots, p-1$.

    Consider a \chainP protocol $\cP$, and  let $M$ denote the random variable corresponding to the message sent by the first player. Note that $M$ is a function of the random variable $X^1$.
    Fix a message $m$ sent by the first player, and let $t^2$ be the index received by the second player.  
    Denote by $\cP(m, t^2)$ the $(p-1)$-player protocol that proceeds in the same way as $\cP$ proceeds for players $2, \ldots, p$ after the first player sent message $m$ and the second player received the index $T^2 = t^2$.
    Thus, with this notation, we can write $\E[ Z_{\cP}(X^1, T^2, \ldots, X^{p-1}, T^p)]$ as
    \begin{align*}
         \sum_{m, t^2} \Pr[M = m, T^2 = t^2] \cdot \E  \left[ Z_{\cP(m, t^2)}(X^2, T^3, \ldots, X^{p-1}, T^p) \mid M = m, T^2 = t^2 \right] \enspace,
    \end{align*}
    or more concisely as
    \begin{align*}
        \E_{M, T^2} \left[ \E  \left[ Z_{\cP(m, t^2)}(X^2, T^3, \ldots, X^{p-1}, T^p) \mid M = m, T^2 = t^2 \right] \right]\enspace.
    \end{align*}
    Recall that the total variation distance between two distributions is the largest difference the two distributions can assign to an event. It follows (by considering the event ``$\cP(m, t^2)$ gives the correct prediction'') that
    \begin{align*}
        \E  \left[ Z_{\cP(m, t^2)}(X^2, T^3, \ldots, X^p , T^{p-1}) \mid M = m, T^2 = t^2 \right] 
    \end{align*}
    is upper bounded by 
    \begin{align*}
         \E  \left[ Z_{\cP(m, t^2)}(X^2, T^3, \ldots, X^p , T^{p-1}) \right] + \mbox{TVD}(m, t^2)\enspace,
    \end{align*}
     where $\mbox{TVD}(m, t^2)$ denotes the total variation distance between the unconditional distribution of $X^2, T^3, \ldots, X^{p-1}, T^p$ and the distribution of these random variables conditioned on $M= m$ and $T^2 = t^2$.

    Note that the unconditional distribution is equivalent to taking a random $(p-1)$-player instance from $D^{p-1}$.
    By the induction hypothesis, we have that any $(p-1)$-player protocol succeeds with probability at most $1/2 + s \cdot (p-1) $ over this distribution of instances, and  so
    \begin{align*}
        \E  \left[ Z_{\cP(m, t^2)}(X^2, T^3, \ldots, X^p , T^{p-1}) \right] \leq 1/2 + s \cdot (p-1)\enspace.
    \end{align*}
   
   We proceed to analyze the total variation distance $\mbox{TVD}(m, t^2)$.  By \cref{lem:pindex_tvd},
    \begin{align*}
        \mbox{TVD}(m, t^2) =  \left| 1/2 - \Pr[X^1_{t^2} = 0 \mid M = m] \right|\enspace.
    \end{align*}
    Consider now the protocol $\cP'$ for the \chainOnlyP{2} problem where the first player is identical to the first player in the $p$-player protocol $\cP$ and the second player, given the message $m$ and the index $t^2$, predicts that $x^1_{t^2}$ equals $0$ if and only if  $\Pr[X^1_{t^2} = 0 \mid M = m] \geq 1/2$.
    Conditioned on $M=m$ and $T^2=t^2$, we thus have that $\cP'$ succeeds with probability $1/2+\mbox{TVD}(m, t^2)$. 
    Furthermore, the probability that $\cP'$ succeeds on a random instance is
    \begin{align*}
        \E_{X^1, T^2}\left[ \E_{M} [ 1/2 + \mbox{TVD}(M, T^2) \mid X^1]\right] & =  \E_{M, T^2} [1/2 + \mbox{TVD}(M, T^2)]\enspace.
    \end{align*}
    As no player is a allowed to send messages consisting of more than $(1-H_2(\delta))n -1$ bits and $X^1, T^2$ are uniformly distributed (\cref{obs:pindex_uniform}), \cref{lem:index} says that the success probability is at most $1-\delta $ and so 
    \begin{align*}
        \E_{M, T^2}[ \mbox{TVD}(M, T^2)] \leq 1/2- \delta = s\enspace.
    \end{align*}
    This concludes the inductive step and the proof  of the lemma since we upper bounded 
    \begin{align*}
        \E[ Z_{\cP}(X^1, T^2, \ldots, X^p, T^{p-1})] & = \E_{M, T^2} \left[ \E  \left[ Z_{\cP(m, t^2)}(X^2, T^3, \ldots, X^p , T^{p-1}) \mid M = m, T^2 = t^2 \right] \right]
    \end{align*}
     by 
    \begin{align*}
        1/2 + (p-1)\cdot s + \E_{M, T^2}[ \mbox{TVD}(M, T^2)] \leq 1/2 +  p \cdot s\enspace. 
        \tag*{\qedhere}
    \end{align*}

\end{proof}

We now finalize the proof of \cref{thm:pindex_hardness}.

\begin{proof}[Proof of \cref{thm:pindex_hardness}]
    First note that we can assume that $n\geq 36 p^2$ since any protocol needs communication complexity at least $1$ to have a success probability above $1/2$. 
    Now, let $s = 1/(6p)$ and $\delta = 1/2 - s$.
    \cref{lem:pindex_hardness} says that any (potentially randomized) protocol  for the \chainP problem with success probability at least $1/2 + p\cdot s$  has communication complexity at least $(1- H_2(\delta))n - 1$.
    
    With these parameters we have $1/2 + p \cdot s = 2/3$. So any protocol with success probability $2/3$ must have communication complexity at least $(1- H_2(\delta))n -1$. By the definition of the binary entropy,
    \begin{align*}
        1 - H_2\left( \frac{1 - 2s}{2} \right) &  =    \frac{1-2s}{2} \log_2(1-2s) + \frac{1+2s}{2} \log_2(1+2s) \enspace, 
    \end{align*}
    which, by using Taylor series, can be bounded by  
    \begin{align*}
        \frac{1-2s}{2} \left(- \sum_{i=1}^\infty \frac{(2s)^i}{i}\right)  + \frac{1+2s}{2} \left( \sum_{i=1}^\infty (-1)^{i+1} \frac{(2s)^i}{i} \right) & =  2s \cdot\left( \sum_{i=1}^\infty \frac{(2s)^{2i-1}}{2i-1}\right) - \left( \sum_{i=1}^\infty \frac{(2s)^{2i}}{2i} \right) \\
        & = \sum_{i=1}^\infty (2s)^{2i} \left(\frac{1}{2i-1} - \frac{1}{2i} \right) \\
        & \geq (2s)^2 \left( \frac{1}{2-1} - \frac{1}{2}\right) = 2s^2\enspace.
    \end{align*}

    We thus have that any protocol for the \chainP problem with success probability at least $2/3$ has communication complexity at least
    \begin{align}
        (1- H_2(\delta))n -1 & \geq (2s^2) n -1  = \frac{n}{18p^2} -1 \geq \frac{n}{36 p^2} \enspace, 
        \label{eq:gen_hardness_complexity}
    \end{align}
    where the last inequality holds because $n \geq 36 p^2$.
%    
%    Now suppose toward contradiction that there is a protocol $\cP$ for the $p$-INDEX problem with success probability at least $1+\varepsilon$ and communication complexity  at most $\frac{n\varepsilon}{200 p^2}$. Using $\cP$ we construct a new protocol $\cP_{3/4}$ with success probability $3/4$ and communication complexity less than $\frac{n}{200p^2}$ which contradicts the minimum communication complexity of any such protocol. 
%    
%    Protocol $\cP_{3/4}$ is obtained by running $r$ independent copies of $\cP$ in parallel. 
%    The last player then outputs the majority vote of the outcome of each of these copies.  
%    Let $Y_i$ be the indicator random variable that the $i$-th copy makes the correct prediction. 
%    Further, let $Y = \sum_{i=1}^r Y_i$. 
%    Note that $\cP_{3/4}$ outputs the correct prediction if  $Y > r/2$ and $\E[Y] = (1+\varepsilon) r$. We thus have 
%    \begin{align*}
%        \Pr[\cP_{3/4} \mbox{ makes wrong prediction }] \leq \Pr[Y \leq (1-\varepsilon/2) \E[Y]]\,.
%    \end{align*}
%    As $Y$ is a sum of $r$ independent  Bernoulli random variables, a standard application of Chernoff bounds implies
%    \begin{align*}
%        \Pr[Y \leq (1-\varepsilon/2) \E[Y]] \leq e^{-\frac{(\varepsilon/2)^2 \E[Y]}{2}}
%    \end{align*}
    \end{proof}

%% file: A30-geometric-grouping.tex
\section{Proof of the First Part of \texorpdfstring{\cref{thm:uppertwoplayersub}}{Theorem~\ref*{thm:uppertwoplayersub}}} \label{app:geometric_grouping}

\cref{alg:many_sizes_exponential} is a variant of \cref{alg:many_sizes} modified to use exponential grouping as described in the end of \cref{ssc:two_players_algs}.

\begin{protocol}
\caption{Repeated Solving with Varying Sizes and Exponential Grouping} \label{alg:many_sizes_exponential}
\textbf{Alice's Algorithm}
\begin{algorithmic}[1]
\State $I \gets \{0\} \cup \{\lfloor (1 + \eps)^j \rfloor, 2\lfloor (1 + \eps)^j \rfloor \mid j\in \{0,\ldots, \lfloor \log_{1+\eps}k \rfloor\} \}$.
\For{every $i \in I$} 
    \State	Let $S_i$ be the set maximizing $f$ among all subsets of $V_A$ of size at most $i$.
\EndFor
\State Send all the sets $\{S_i\}_{i \in I}$ as the message to Bob.
\end{algorithmic}
\textbf{Bob's Algorithm}
\begin{algorithmic}[1]
\State Let $\widehat{S}$ be the subset of $V_B \cup \bigcup_{i\in I} S_{i}$ of size at most $k$ maximizing $f$ ($I$ is not part of the message sent by Alice, but it can easily be recomputed by Bob).\\
\Return{$\widehat{S}$}.
\end{algorithmic}
\end{protocol}

One can observe that, like \cref{alg:many_sizes}, \cref{alg:many_sizes_exponential} is also guaranteed to output a feasible set. Furthermore, the number of elements Alice sends to Bob under this protocol is upper bounded by
\begin{align*}
    \sum_{i \in I} |S_i|
    \leq{} &
    2k \cdot |I|
    \leq
    2k \cdot (3 + 2 \cdot \log_{1 + \eps} k)\\
    ={} &
    2k \cdot \left(3 + \frac{2 \cdot \ln k}{\ln(1 + \eps)}\right)
    \leq
    2k \cdot \left(3 + \frac{2 \cdot \ln k}{\eps / 2}\right)
    =
    O(k\eps^{-1} \cdot \log k)
    \enspace.
\end{align*}
Thus, to complete the proof of the first part of \cref{thm:uppertwoplayersub}, we are only left to show that \cref{alg:many_sizes_exponential} produces a $(\nicefrac{2}{3} - \eps)$-approximation, which is our objective in the rest of this section. We use towards this goal the following well-known lemma.
\begin{lemma}[A rephrased version of Lemma~2.2 of~\cite{feige2011maximizing}] \label{lem:sampling}
Let $g\colon 2^X \to \mathbb{R}$ be a submodular function. Denote by $A(p)$ a random subset of $A$ in which each element appears with probability p (not necessarily independently). Then, $\expected{g(A(p))} \geq (1−p) \cdot g(\varnothing) + p \cdot g(A)$.
\end{lemma}

Recall that by $\Oset$ we denote a subset of $V_A \cup V_B$ of size $k$ maximizing $f$ among all such subsets, and $\Oval$ denotes the value of this set. Also, let $M = V_B \cup \bigcup_{i \in I} S_{i}$ be the set of elements that Bob gets either from Alice or directly, which is also the set of elements in which Bob looks for $\widehat{S}$. Finally, let $\hi$ be the maximum value in $\{0\} \cup \{\lfloor (1 + \eps)^j \rfloor \mid j\in \{0,\ldots, \lfloor \log_{1+\eps} k \rfloor\}\}$ that is not larger than $k - |\Oset \cap M|$. Observe that the definition of $\hi$ guarantees that $\hi \in I$ and $(1 + \eps)\hi \geq k - |\Oset \cap M|$. Furthermore, the definition of $I$ implies $2\hi \in I$.

We are now ready to state the following claims, which correspond to \cref{obs:remaining_OPT_approximation}, \cref{lem:large_set_guarantee} and \cref{lem:secondary_relationship} from \cref{ssc:two_players_algs}, respectively.
\begin{observation} \label{obs:remaining_OPT_approximation_exponential}
$f(S_{\hi}) \geq (1 - \eps) \cdot f(\Oset \setminus M)$.
\end{observation}
\begin{proof}
Let $T$ be a uniformly random subset of $\Oset \setminus M$ of size $\hi$ (if $|\Oset \setminus M| < \hi$, we set $T$ to be deterministically equal to $\Oset \setminus M$). Since every element of $\Oset \setminus M$ belongs to $T$ with some probability $p \geq \hi / (k - |\Oset \cap M|) \geq 1 / (1 + \eps)$, we get by \cref{lem:sampling} that $\expected{f(T)} \geq f(\Oset \setminus M) / (1 + \eps) \geq (1 - \eps) \cdot f(\Oset \setminus M)$. Thus, there exists some realization $T'$ of $T$ obeying $f(T') \geq (1 - \eps) \cdot f(\Oset \setminus M)$. The observation now follows from the choice of $S_{\hi}$ by \cref{alg:many_sizes_exponential} since the set $T'$ is a subset of $(V_A \cup V_B) \setminus M \subseteq V_A$ of size at most $\hi$.
\end{proof}

\begin{lemma} \label{lem:large_set_guarantee_exponential}
$f(S_{2\hi}) \geq (1 - \eps) \cdot [\Oval + f(S_{\hi}) - f(\widehat{S})]$.
\end{lemma}
\begin{proof}
To prove the lemma, we have to show that $V_A$ includes a set of size at most $2\hi$ whose value is at least $(1 - \eps) \cdot [\Oval + f(S_{\hi}) - f(\widehat{S})]$. Let $T$ be a uniformly random subset of $\Oset \setminus M$ of size $\hi$ (like in the previous proof, if $|\Oset \setminus M| < \hi$, then we define $T$ to be deterministically equal to $\Oset \setminus M$), and consider the set $T \cup S_{\hi}$. First, we observe that this set is a subset of $V_A$ because $V_B \subseteq M$. Second, the size of this set is at most $|T| + |S_{\hi}| \leq 2\hi$. Thus, to prove the lemma it remains to show that the value of this set is at least $(1 - \eps) \cdot [\Oval + f(S_{\hi}) - f(\widehat{S})]$ for some realization of $T$, which we do by showing that the expected value of this set is at least that large. Note that $(\Oset \cap M) \cup S_{\hi}$ is a subset of $M$ of size
\[
    |(\Oset \cap M) \cup S_{\hi}|
    \leq
    |\Oset \cap M| + \hi
    \leq
    (k - \hi) + \hi
    =
    k
    \enspace,
\]
and thus, $f(\widehat{S}) \geq f((\Oset \cap M) \cup S_{\hi})$ by the definition of $\widehat{S}$. Using the last inequality, we get
\begin{align*}
	\frac{\expected{f(T \cup S_{\hi})}}{1 - \eps}
	\geq{} &
	\frac{1}{(1 + \eps)(1 - \eps)} \cdot f((\Oset \setminus M) \cup S_{\hi})
	\geq
	f((\Oset \setminus M) \cup S_{\hi})\\
	\geq{} &
	f((\Oset \setminus M) \cup S_{\hi}) + f((\Oset \cap M) \cup S_{\hi}) - f(\widehat{S})\\
	\geq{} &
	f(\Oset \cup S_{\hi}) + f(S_{\hi}) - f(\widehat{S})
	\geq
	f(\Oset) + f(S_{\hi}) - f(\widehat{S}) \enspace,
\end{align*}
where the second inequality follows from the non-negativity of $f$ and the last two inequalities follow from the submodularity and monotonicity of $f$, respectively. The first inequality follows from \cref{lem:sampling} by defining $g(S) = f(S \cup S_{\hi})$ since $T$ includes every element of $\Oset \setminus M$ with some probability $p \geq \hi / (k - |\Oset \cap M|) \geq 1 / (1 + \eps)$.
\end{proof}

\begin{lemma}\label{lem:secondary_relationship_exponential}
$2 \cdot f(\widehat{S}) \geq f(\Oset \cap M) + f(S_{2\hi})$.
\end{lemma}
We omit the proof of the last lemma since it is identical to the proof of \cref{lem:secondary_relationship} up to a replacement of every occurrence of the expression $k - |\Oset \cap M|$ with $\hi$.

We are now ready to prove the approximation guarateee of \cref{alg:many_sizes_exponential} (and thus, also complete the proof of the first part of \cref{thm:uppertwoplayersub}).
\begin{corollary}
\cref{alg:many_sizes_exponential} is a $(\nicefrac{2}{3} - \eps)$-approximation protocol.
\end{corollary}
\begin{proof}
Combining \cref{lem:large_set_guarantee_exponential,lem:secondary_relationship_exponential}, we get
\[
	2 \cdot f(\widehat{S})
	\geq
	f(\Oset \cap M) + f(S_{2\hi})
	\geq
	f(\Oset \cap M) + (1 - \eps) \cdot [\Oval + f(S_{\hi}) - f(\widehat{S})]
	\enspace.
\]
Rearranging this inequality, and plugging into it the lower bound on $f(S_{\hi})$ given by \cref{obs:remaining_OPT_approximation_exponential}, yields
\begin{align*}
	f(\widehat{S})
	\geq{} &
	\frac{f(\Oset \cap M) + (1 - \eps) \cdot [\Oval + (1 - \eps) \cdot f(\Oset \setminus M)]}{3}\\
	\geq{} &
	\frac{(1 - \eps) \cdot \Oval + (1 - 2\eps) \cdot [f(\Oset \cap M) + f(\Oset \setminus M)]}{3}
	\geq
	(\nicefrac{2}{3} - \eps) \cdot \Oval
	\enspace,
\end{align*}
where the third inequality follows from the submodularity and non-negativity of $f$. Because $\widehat{S}$ is the output of \cref{alg:many_sizes_exponential}, this concludes the proof.
\end{proof}

%% file: A-missing-proofs.tex
\section{Missing Proofs}
\input{A00-missing-proofs.tex}

%% file: A00-missing-proofs.tex
\subsection{Missing Proof of \texorpdfstring{\cref{sec:two-player-sub-hardness}}{Section~\ref*{sec:two-player-sub-hardness}}} \label{app:missing_proofs}

\lemTwoSidesReductionProperties*
\begin{proof}
It is clear from its definition that $g_i$ is non-negative. To prove that $g_i$ is also monotone and submodular we need to show 
\begin{enumerate}[label=(\roman*),itemsep=0em]
\item\label{item:giMonotone} \parbox{40mm}{$g_i(v \mid S) \geq 0$} for every $S\subseteq W'$ and $v\in W'\setminus S$, and

\item\label{item:giSubmodular} \parbox{40mm}{$g_i(v \mid S_1) \geq g_i(v \mid S_2)$} for every $S_1\subseteq S_2 \subseteq W'$ and $v\in W'\setminus S_2$.
\end{enumerate}
We show that by considering a few cases. For $v = w$, we get
\[
    g_i(w \mid S)
    =
    \begin{cases*}
        \frac{1}{3} & \text{if $S \subseteq \{v_i\}$} \enspace,\\
        0 & \text{otherwise} \enspace.
    \end{cases*}
\]
For $v = v_i \not\in S$, we get
\[
    g_i(v_i \mid S)
    =
    \begin{cases*}
        \frac{2}{3} & \text{if $S \subseteq \{w\}$} \enspace,\\
        \frac{1}{3} & \text{if $|S \setminus \{w\}| = 1$} \enspace,\\
        0 & \text{otherwise} \enspace.
    \end{cases*}
\]
Finally, for $j \neq i$ and $v = v_j\not\in S$,
\[
    g_i(v_j \mid S)
    =
    \begin{cases*}
        \frac{2}{3} & \text{if $S = \varnothing$} \enspace,\\
        \frac{1}{3} & \text{if $S = \{w\}$} \enspace,\\
        \frac{1}{3} & \text{if $|S \setminus \{w\}| = 1$ and $\{v_i, w\} \not \subseteq S$} \enspace,\\
        0 & \text{otherwise} \enspace.
    \end{cases*}
\]
One can verify that in all the above cases $g_i(v \mid S)$ indeed fulfills~\ref{item:giMonotone} and~\ref{item:giSubmodular}.

It remains to show that $g_i$ being non-negative, monotone and submodular implies that $f_i$ has these properties as well. Recall that $f_i(S) = G_i(y^S)$. Thus, the fact that $g_i$ is non-negative (and therefore, so is its multilinear extension $G_i$) directly implies non-negativity of $f_i$. Consider now two sets $S_1 \subseteq S_2 \subseteq W$. One can observe that the definition of $y^S$ implies $y^{S_1} \leq y^{S_2}$ component-wise. Hence, by the monotonicity of $g_i$ we obtain
\begin{equation*}
    f_i(S_1)
    =
    G_i\left(y^{S_1}\right)
    =
    \ee{g_i\left(\RSet\left(y^{S_1}\right)\right)}
    \leq
    \ee{g_i\left(\RSet\left(y^{S_2}\right)\right)}
    =
    G_i\left(y^{S_2}\right)
    =
    f_i(S_2)
    \enspace,
\end{equation*}
which implies that $f_i$ is also monotone.

We now check the submodularity of $f_i$. For that purpose, let $v$ be an arbitrary element of $W \setminus S_2$. If $v = w$, then
\begin{align*}
    f_i(w \mid S_1)
    ={} &
    G_i\left(y^{S_1 \cup \{w\}}\right) - G_i\left(y^{S_1}\right)
    =
    \ee{g_i\left(w \mid \RSet\left(y^{S_1}\right)\right)}\\
    \geq{} &
    \ee{g_i\left(w \mid \RSet\left(y^{S_2}\right)\right)}
    =
    G_i(y^{S_2 \cup \{w\}}) - G_i(y^{S_2})
    =
    f_i(w \mid S_2)
    \enspace,
\end{align*}
where the second and third equalities hold by linearity of expectation, and the inequality follows from the submodularity of $g_i$ and the inequality $y^{S_1} \leq y^{S_2}$. Similarly, if $v =u_i^j$ for some $i\in [n]$ and $j\in [k-1]$, then
\begin{align*}
    f_i(u_i^j \mid S_1)
    &=
    G_i\left(y^{S_1} + \characteristic_{v_i} / (k - 1)\right) - G_i\left(y^{S_1}\right)\\
    ={} &
    \frac{G_i\left(y^{S_1 \cup \{u_i^1, u_i^2, \dotsc, u_i^{k - 1}\}}\right) - G_i\left(y^{S_1 \setminus \{u_i^1, u_i^2, \dotsc, u_i^{k - 1}\}}\right)}{k - 1}
    =
    \frac{\ee{g_i\left(v_i \mid \RSet\left(y^{S_1}\right) \setminus \{v_i\}\right)}}{k-1}\\
    \geq{} &
    \frac{\ee{g_i\left(v_i \mid \RSet\left(y^{S_2}\right) \setminus \{v_i\}\right)}}{k-1}
    =
    \frac{G_i\left(y^{S_2 \cup \{u_i^1, u_i^2, \dotsc, u_i^{k - 1}\}}\right) - G_i\left(y^{S_2 \setminus \{u_i^1, u_i^2, \dotsc, u_i^{k - 1}\}}\right)}{k - 1}\\
    ={} &
    G_i\left(y^{S_2} + \characteristic_{v_i} / (k - 1)\right) - G_i\left(y^{S_2}\right)
    =
    f_i(u \mid S_2)
    \enspace,
\end{align*}
where the second and penultimate equalities hold by the multilinearity of $G_i$. This completes the proof that $f_i$ is submodular.
\end{proof}

%% file: A25-fractional-coverage.tex
\subsection{Missing Proofs of \texorpdfstring{\cref{sec:gen_hardness}}{Section~\ref*{sec:gen_hardness}}} 

\subsubsection{Weighted Fractional Coverage Functions are Weighted Coverage Functions}
\label{app:frac-cover}
The following lemma was stated informally in the section. Here, we state it formally, and provide a proof for it.

\begin{lemma}
Every weighted fractional coverage function is a weighted coverage function.
\end{lemma}
\begin{proof}
Consider a weighted fractional coverage function $f\colon 2^V \to \mathbb{R}_{\geq 0}$. Hence, $f$ is of the form
\begin{equation*}
f(S) = \sum_{u\in U} a_u \cdot \left( 1- \prod_{v\in S} (1-p_v(u)) \right) \qquad \forall S\subseteq V\enspace,
\end{equation*}
where $U$ is a finite universe with non-negative weights $a\colon U\to \mathbb{R}_{\geq 0}$, each element $v\in V$ is a subset of $U$, i.e., $V\subseteq 2^{U}$, and $p_v\colon U\to [0,1]$ for $v\in V$.

To show that $f$ is a weighted coverage function, we interpret each element $v\in V$ as a subset $\overline{v}$ of a new universe $\overline{U}$ with non-negative weights $\overline{a}\colon U\to \mathbb{R}_{\geq 0}$, such that
\begin{equation}\label{eq:fIsCoverage}
f(S) = \sum_{\overline{u}\in \bigcup\limits_{v\in S}\overline{v}}\overline{a}(\overline{u})\qquad \forall S\subseteq V\enspace.
\end{equation}
We now define a universe $\overline{U}$ with weights $\overline{a}:U\to\mathbb{R}_{\geq 0}$ and the mapping from $v\in V$ to $\overline{v} \subseteq \overline{U}$ such that~\eqref{eq:fIsCoverage} holds. The universe $\overline{U}$ is 
\begin{equation*}
\overline{U} = 2^V\enspace,
\end{equation*}
and the weight $\overline{a}(\overline{u})\in \mathbb{R}_{\geq 0}$ of an element $\overline{u} \in \overline{U}$ is set to
\begin{equation*}
\overline{a}(\overline{u}) = \sum_{u\in U} \left(a_u \prod_{v\in \overline{u}} p_v(u) \prod_{v\in V\setminus \overline{u}} (1-p_v(u))\right)\enspace.
\end{equation*} 
Moreover, an element $v\in V$, which is a subset of $U$, gets mapped to the subset $\overline{v}$ of $\overline{U}$ given by
\begin{equation*}
\overline{v} = \left\{ \overline{u} \in \overline{U} : v\in \overline{u} \right\}\enspace.
\end{equation*}
Notice that this implies (in particular) that, for any set $S\subseteq V$, 
\begin{equation}\label{eq:Sbar}
\bigcup_{v\in S} \overline{v} = \left\{\overline{u}\in \overline{U}: S\cap \overline{u} \neq \varnothing \right\}\enspace.
\end{equation}
We now show that~\eqref{eq:fIsCoverage} holds. Hence, let $S\subseteq V$. We have
\begin{equation}\label{eq:expandABarValue}
\begin{aligned}
&\sum_{\overline{u}\in \bigcup\limits_{v\in S}\overline{v}}\overline{a}(\overline{u})
= \sum_{\substack{\overline{u}\in \overline{U}:\\ S\cap \overline{u}\neq\varnothing}} \overline{a}(\overline{u})
= \sum_{\substack{X\subseteq V:\\ S\cap X\neq\varnothing}} \overline{a}(X)
=\sum_{\substack{X_1\subseteq S:\\ X_1\neq\varnothing}}\; \sum_{X_2\subseteq V\setminus S} \overline{a}(X_1\cup X_2)\\
&=\sum_{u\in U} \left[a_u
\left(\sum_{\substack{X_1\subseteq S:\\X_1\neq\varnothing}}\;
\prod_{v\in X_1} p_v(u) \prod_{v\in S\setminus X_1} (1-p_v(u))\right)
\left(\sum_{X_2\subseteq V\setminus S}\; \prod_{v\in X_2} p_v(u) \prod_{v\in (V\setminus S)\setminus X_2}\mspace{-18mu}(1-p_v(u))\right)\right]\enspace,
\end{aligned}
\end{equation}
where the first equality follows from~\eqref{eq:Sbar}. To further expand~\eqref{eq:expandABarValue}, we observe the following basic fact.
\begin{claim}\label{claim:totalProbIsOne}
Let $Z$ be a finite set, and let $p\colon Z\to [0,1]$. Then,
\begin{equation*}
\sum_{X\subseteq Z} \left(\prod_{z\in X} p_z \prod_{z\in Z\setminus X} (1-p_z)\right) = 1\enspace.
\end{equation*}
\end{claim}
\begin{proof}
One can interpret the values $p\colon Z\to [0,1]$ as probabilities. Let $Q$ be a random subset of $Z$ containing element $z\in Z$ with probability $\Pr[z\in Q] = p_z$, independently of the other elements. Then, for any fixed $X\subseteq Z$,
\begin{equation*}
\Pr[Q = X] = \prod_{z\in X} p_z \prod_{z\in Z\setminus X} (1-p_z)\enspace.
\end{equation*}
The claim now follows by observing that
\begin{equation*}
\sum_{X\subseteq Z}\Pr[Q=X] = 1\enspace,
\end{equation*}
because $Q$ realizes to some subset of $Z$ with probability $1$.
\end{proof}

Claim~\ref{claim:totalProbIsOne} allows for the following simplifications of terms from~\eqref{eq:expandABarValue}:
\begin{align*}
\sum_{\substack{X_1\subseteq S:\\ X_1\neq \varnothing}}
\left(\prod_{v\in X_1} p_v(u) \prod_{v\in S\setminus X_1} (1-p_v(u))\right)
&= 1 - \prod_{v\in S} (1-p_v(u))\enspace, \text{ and}\\
\sum_{X_2\subseteq V\setminus S}
\left(\prod_{v\in X_2} p_v(u) \prod_{v\in (V\setminus S)\setminus X_2} (1-p_v(u))\right)
&= 1\enspace,
\end{align*}
thus leading to
\begin{align*}
\sum_{\overline{u}\in \bigcup\limits_{v\in S}\overline{v}}\overline{a}(\overline{u}) 
 &= \sum_{u\in U} a_u \left(1 - \prod_{v\in S} p_v(u) \right) = f(S) \enspace,
\end{align*}
which shows~\eqref{eq:fIsCoverage} as desired.
\end{proof}

\subsubsection{Bounding the Partial Derivatives of \texorpdfstring{\boldmath{$\widehat{F}$}}{F}}
\label{app:bounding_partial_derivatives}
In this section we perform the rather mechanical calculations that bound the partial derivatives of $\widehat{F}$ at $\vec{1} = (1,1, \ldots, 1)$. For convenience, recall the definition of $\hat F$ is
\begin{align*}
    \widehat{F}(s_1, \ldots, s_p) =a_p +  \sum_{j=1}^{p-1} a_j \cdot  \left( 1 - \prod_{i=1}^j \left(1 - \frac{a_i}{A_{\geq i}} \right)^{s_i} \right)\,.
\end{align*}

We start by showing the following identity, which follows by first taking the partial derivative, and then applying the identities of  \cref{lem:gen_hardness_select_a}.
\begin{claim*}
    For $\ell = 1, \ldots, p-1$,
    \begin{align*}
        \frac{\partial \hat F}{\partial s_\ell}(\vec{1}) & =  -\ln\left( 1- \frac{a_\ell}{A_{\geq \ell}}\right) \left( \frac{A_{\geq \ell}}{a_\ell} - (p+1-\ell) \right)\,.
    \end{align*}
\end{claim*}
\begin{proof}
We have
    \begin{align*}
        \frac{\partial \hat F}{\partial s_\ell}(\vec{s}) & =  \frac{\partial }{\partial s_\ell} \left( \sum_{j=1}^{p-1} a_j \cdot  \left( 1 - \prod_{i=1}^j \left(1 - \frac{a_i}{A_{\geq i}} \right)^{s_i} \right)+ a_p\right)\\
         & = - \ln\left( 1- \frac{a_\ell}{A_{\geq \ell}}\right) \sum_{j=\ell}^{p-1} a_j \prod_{i=1}^j \left(1 - \frac{a_i}{A_{\geq i}} \right)^{s_i}\enspace.
    \end{align*}
    Using  \cref{lem:gen_hardness_select_a}, we can simplify this expression for $\vec{s} = \vec{1}$ as follows.
    \begin{align*}
        \frac{\partial \hat F}{\partial s_\ell}(\vec{1})  & = - \ln\left( 1- \frac{a_\ell}{A_{\geq \ell}}\right) \sum_{j=\ell}^{p-1} a_j \prod_{i=1}^j \left(1 - \frac{a_i}{A_{\geq i}} \right) \\
        & = - \ln\left( 1- \frac{a_\ell}{A_{\geq \ell}}\right) \sum_{j=\ell}^{p-1}  \left(1 - \frac{a_j}{A_{\geq j}} \right) & \mbox{\small $a_j\prod_{i=1}^{j-1} \left(1 - \frac{a_i}{A_{\geq i}}\right) = 1$ by \cref{lem:gen_hardness_select_a}} \\
        & = -\ln\left( 1- \frac{a_\ell}{A_{\geq \ell}}\right) \sum_{j=\ell}^{p}  \left(1 - \frac{a_j}{A_{\geq j}} \right) & \mbox{\small $a_p/A_{\geq p} = a_p/a_p = 1$}\\
%        & = -\ln\left( 1- \frac{a_\ell}{A_\ell}\right) \sum_{j=\ell}^{p}  \left(2 - \frac{a_j}{A_{\geq j}} \right) - (p+1-\ell)\\
        & =-\ln\left( 1- \frac{a_\ell}{A_{\geq \ell}}\right) \left( \frac{A_{\geq \ell}}{a_\ell} - (p+1-\ell) \right)\enspace. & \mbox{\small $\sum_{j=\ell}^p  \left(2 - \frac{a_j}{A_{\geq j}} \right) = \frac{A_{\geq \ell}}{a_\ell}$ by  \cref{lem:gen_hardness_select_a}}
        \tag*{\qedhere}
    \end{align*}
\end{proof}
The bound on the partial derivatives, which we restate here for convenience, now follows from the above identity, the fact that $\ln(1-x) = -\sum_{i=1}^\infty \frac{x^i}{i}$ for $|x| < 1$, and the bounds of \cref{lem:gen_hardness_select_a}.
\claimpartialderivatives*
\begin{proof}
    Note that for $\ell = p$ the statement trivially holds because $\frac{\partial \widehat{F}}{\partial s_p} (\vec{1}) = 0$. 
    For $\ell = 1, \ldots, p-1$, we use the identity of the previous claim together with the fact that $\ln(1-x) = -\sum_{i=1}^\infty \frac{x^i}{i}$ for $|x| < 1$. 
    For brevity, we let $x = a_\ell/A_{\geq \ell}$. Then,
    \begin{align*}
        \frac{\partial \hat F}{\partial s_\ell}(\vec{1}) &= -\ln\left( 1- x\right) \left( \frac{1}{x} - (p+1-\ell) \right) \\
        & = \sum_{i=1}^\infty \left(\frac{x^i}{i}\right) \left( \frac{1}{x} - (p+1-\ell) \right) \\
        & = \sum_{i=1}^\infty \left(\frac{x^{i-1}}{i}\right) \left( 1 - x\cdot (p+1-\ell) \right)\enspace.
    \end{align*}
    By \cref{lem:gen_hardness_select_a}, we have $x = a_{\ell}/A_{\geq \ell} \leq \left( 2(p+1-\ell) - H_{p+1-\ell}\right)^{-1}$. Since $\sum_{i=1}^\infty \left(\frac{x^{i-1}}{i}\right)\geq \frac{x^0}{1} = 1$, this gives the lower bound
    \begin{align*}
        \frac{\partial \hat F}{\partial s_\ell}(\vec{1}) &\geq 1 - x\cdot (p+1-\ell) \\
        &\geq 1 - \frac{(p+1-\ell)}{2(p+1-\ell) - H_{p+1-\ell}}\\
        &= \frac{1}{2} - \frac{1}{2}\cdot \frac{H_{p+1-\ell}}{2(p+1-\ell) - H_{p+1-\ell}}\\
 	&\geq \frac{1}{2} - \frac{1}{2} \cdot \frac{H_{p+1-\ell}}{p+1-\ell} \\
 	&\geq \frac{1}{2} - \frac{H_{p+1-\ell}}{p+1-\ell}\enspace. 
    \end{align*}
    For the upper bound, we have
    \begin{align*}
        \frac{\partial \hat F}{\partial s_\ell}(\vec{1})  & = \sum_{i=1}^\infty \left(\frac{x^{i-1}}{i}\right) \left( 1 - x\cdot (p+1-\ell) \right) \\
        & \leq \left(1 + \frac{1}{2} \sum_{i=1}^\infty \left(x^{i}\right)\right) \left( 1 -x \cdot (p+1-\ell)\right) \\
        & = \left(1 + \frac{1}{2}\cdot \frac{x}{1 -x}\right) \left( 1 - x\cdot (p+1-\ell)\right) \\
        & = 1 - x\cdot (p+1-\ell)  +\frac{1}{2} x \cdot \frac{1 - x\cdot (p+1-\ell)}{1-x} \\
        & \leq 1 - x\cdot\left( (p+1-\ell) - \frac{1}{2} \right)\\
        & \leq 1/2\enspace,
    \end{align*}
where the last inequality is due to the bound $x = a_\ell/A_{\geq \ell} \geq \left(2(p+1 - \ell) - 1\right)^{-1}$ of \cref{lem:gen_hardness_select_a}.
\end{proof}

%% file: A10-poly2players.tex
\subsection{Missing Proofs of \texorpdfstring{\cref{sec:poly2players}}{Section~\ref*{sec:poly2players}}} \label{app:poly2players}

\claimUpperBoundEta*
\begin{proof}
Notice that the claim clearly holds for $x\in \{0,\ldots, k_A\}$ as $(1-\sfrac{1}{k_A})^x \leq e^{-\frac{x}{k_A}}$ holds because $1+y \leq e^y$ for all $y\in \mathbb{R}$.
\smallskip

To prove the claim also for fractional values of $x$, we fix an integer $p\in \{0,\ldots, k_A-1\}$ and show that $\eta(x)\leq e^{-\frac{x}{k_A}}$ holds for all $x\in [p,p+1]$. Hence, let $x=p+\lambda$ with $\lambda\in [0,1]$.
By the construction of $\eta$, we have
\begin{equation*}
\eta(x) =   (1-\lambda)\cdot \left(1 - \frac{1}{k_A}\right)^{p}
       + \lambda \cdot \left(1 - \frac{1}{k_A}\right)^{p+1}
   \triangleq g(\lambda)
\enspace.
\end{equation*}
Hence, our goal is to show $e^{-\frac{p+\lambda}{k_A}} - g(\lambda) \geq 0$ for all $\lambda\in [0,1]$. To find a minimizer of
\begin{equation*}
h(\lambda) = e^{- \frac{p+\lambda}{k_A}} - g(\lambda)
\end{equation*}
over $\lambda\in [0,1]$, we consider the border values $\lambda \in \{0,1\}$, and points in between for which $h$ has a derivative of zero. Because we already proved the statement for integer $x$, the border values $\lambda\in \{0,1\}$ fulfill $h(\lambda)\geq 0$. Moreover,
\begin{equation*}
h'(\lambda) = - \frac{1}{k_A} \cdot e^{-\frac{p+\lambda}{k_A}}
  + \frac{1}{k_A} \left( 1 - \frac{1}{k_A} \right)^p\enspace.
\end{equation*}
Setting this derivative to zero leads to a value $\overline{\lambda}$ satisfying
\begin{equation*}
e^{-\frac{p+\overline{\lambda}}{k_A}} = \left(1 - \frac{1}{k_A}\right)^p\enspace.
\end{equation*}
Notice that such a $\overline{\lambda}$ must satisfy $\overline{\lambda}\geq 0$ because $e^{-\frac{p+\lambda}{k_A}}$ is strictly decreasing in lambda and $e^{-\frac{p}{k_A}} \geq (1-\frac{1}{k_A})^p$.
We now obtain $h(\overline{\lambda}) \geq 0$ as desired because $g(\overline{\lambda}) \leq g(0) \leq (1-\frac{1}{k_A})^p = e^{-\frac{p + \bar{\lambda}}{k_A}}$, where the first inequality holds because $g(\lambda)$ is a non-increasing function, and the equality follows from the choice of $\bar{\lambda}$. This completes the proof of the claim.
\end{proof}

We now would like to prove \cref{lem:boundOptProbFor2PLP}. However, since it is more convenient for our proof technique, we prove a slightly stronger version of this lemma, stated below, where the variables $y$ can take values within $[0,10]$ instead of just $[0,1]$.

%\begingroup
%\def\thelemma{\ref{lem:boundOptProbFor2PLP}}
\begin{lemma}
The optimal value $\alpha$ of the nonlinear program
\begin{equation}\label{eq:optProbFor2PLBApp}
\begingroup
\renewcommand\arraystretch{1.2}
\begin{array}{rr@{\;\;}c@{\;\;}l}
\min & z & &\\
     & z &\geq &1 - e^{-1} - \left[ (1-e^{-1})x - e^{-1} - e^{-1} x \ln x \right]\cdot y \\
     & z &\geq &\frac{1}{2} (1-e^{-1}) + \frac{1}{2}\left[ e^{-1} + x\ln x + (1-e^{-1}) x \right]\cdot y\\
& x &\in &[0,1] \\
& y &\in &[0,10] \\
&   z &\in &\mathbb{R}
\end{array}
\endgroup
\end{equation}
satisfies $\alpha\geq 0.514$.
\end{lemma}
%\addtocounter{lemma}{-1}
%\endgroup

\begin{proof}
We define the two functions $f_1,f_2\colon [0,1]\times [0,10]\to \mathbb{R}$ as the right-hand sides of the two non-trivial equations of~\eqref{eq:optProbFor2PLBApp}, i.e.,
\begin{align*}
f_1(x,y) &= 1 - e^{-1} - \left[ (1-e^{-1})x - e^{-1} - e^{-1} x \ln x \right]\cdot y\enspace, \text{ and} \\
f_2(x,y) &= \frac{1}{2} (1-e^{-1}) + \frac{1}{2}\left[ e^{-1} + x\ln x + (1-e^{-1}) x \right]\cdot y\enspace.\\
\end{align*}
We show below that~\eqref{eq:optProbFor2PLBApp} has a unique minimizer $(x^*, y^*)$, defined as follows.
\begin{enumerate}[label=(\roman*)]
\item $x^*$ is the unique root of $\ln x + 3 - (e+1) x$ in the interval $[0.5,1]$, i.e., 
\begin{equation*}
x^* \approx 0.7175647\enspace.
\end{equation*}
\item $y^*$ is the unique value of $y\in \mathbb{R}$ for which $f_1(x^*,y) = f_2(x^*,y)$, which is
\begin{equation*}
y^* = \frac{1-e}{1 + x^* \left(
3 - 3e + (2-e)\ln x^*
\right)}  \approx 0.6797341\enspace.
\end{equation*}
\end{enumerate}
By plugging these values into~\eqref{eq:optProbFor2PLBApp}, one can easily check that the corresponding optimal $z$ value, denoted by $z^*$ and satisfing $z^* = \max\{f_1(x^*, y^*), f_2(x^*, y^*)\}$, fulfills
\begin{equation*}
z^* \in [0.514, 0.515)\enspace.
\end{equation*}

To show that the minimizer of~\eqref{eq:optProbFor2PLBApp} is indeed the tuple $(x^*, y^*)$ described above, we show the following.
\begin{enumerate}
\item Problem~\eqref{eq:optProbFor2PLBApp} does not have a minimizer $x,y$ at the boundary of the area $[0,1]\times [0,10]$, i.e., any minimizer satisfies $x\in (0,1)$ and $y\in (0,10)$.

\item We then apply the (necessary) Karush-Kuhn-Tucker conditions to a modified version of Problem~\eqref{eq:optProbFor2PLBApp}, where we drop the requirements $x\in [0,1]$ and $y\in [0,10]$, i.e., we only consider the remaining two constraints, described by the right-hand sides given by $f_1$ and $f_2$, to show that $(x^*, y^*)$ is the unique minimizer.
\end{enumerate}

For the first point, we start by observing that
\begin{equation*}
\begingroup
\renewcommand\arraystretch{2.0}
\begin{array}{r@{\;}c@{\;}>{\displaystyle}ll}
f_1(0,y) &\geq &1 - \frac{1}{e} \geq 0.63  &\forall y\in [0,10]\enspace,\text{ and}\\
f_1(x,0) &= &1 - \frac{1}{e} \geq 0.63  &\forall x\in [0,1]\enspace.
\end{array}
\endgroup
\end{equation*}
Consequently, there is no minimizer $(x,y)$ for~\eqref{eq:optProbFor2PLBApp} with either $x=0$ or $y=0$.
Moreover, for $x=1$ and $y\in [0,10]$ we obtain
\begin{align*}
f_1(1,y) &= 1 - \frac{1}{e} - y \left(1 - \frac{2}{e}\right)\enspace,\text{ and} \\
f_2(1,y) &= \frac{1}{2}\left(1 - \frac{1}{e} +y\right)\enspace.
\end{align*}
The minimizer of $\max\{f_1(1,y), f_2(1,y)\}$ for $y\in \mathbb{R}$ is achieved for $\overline{y}$ such that $f_1(1,\overline{y})=f_2(1,\overline{y})$, which is
\begin{equation*}
\overline{y} = \frac{1-e}{4-3e}\enspace.
\end{equation*}
However, when setting $x=1$ and $y=\overline{y}$, the smallest value that $z$ can take in~\eqref{eq:optProbFor2PLBApp} is $f_1(1,\overline{y}) = f_2(1,\overline{y}) \geq 5.2$, which is larger than the value we obtain with $(x^*, y^*)$. Finally, there is also no minimizer of~\eqref{eq:optProbFor2PLBApp} with $y=10$. Indeed, in this case, the objective value of~\eqref{eq:optProbFor2PLBApp} must be at least
\begin{align*}
f_2(x,10) &= \frac{1}{2}  + \frac{9}{2} e^{-1} + 5 (x\ln x + (1-e^{-1})x)\\
&\geq \frac{1}{2} + \frac{9}{2} e^{-1} - 5 e^{e^{-1}-2}\\
&\geq 1\enspace,
\end{align*}
where the first inequality follows by observing that $x\ln x + (1-e^{-1})x$ is a convex function with minimizer at $e^{e^{-1} -2}$; plugging in this minimizer leads to the inequality.

Hence, the minimizer $(x^*, y^*)$ of~\eqref{eq:optProbFor2PLBApp} satisfies $x^*\in (0,1)$ and $y^*\in (0,10)$.

Consequently, it suffices to write the necessary Karush-Kuhn-Tucker conditions for an optimal solution with respect to the two constraints corresponding to $f_1$ and $f_2$. Thus, an optimal solution $(x^*, y^*)$ to~\eqref{eq:optProbFor2PLBApp} must satisfy that there are two multipliers $\lambda_1, \lambda_2 \in \mathbb{R}_{\geq 0}$ such that
\begin{align*}
\lambda_1 + \lambda_2 &= 1  &&\text{(constraint for variable $z$)}\\
\lambda_1 \cdot \nabla_{x} f_1(x^*, y^*) + \lambda_2 \cdot \nabla_{x} f_2(x^*, y^*) &= 0 &&\text{(constraint for variable $x$)}\\
\lambda_1 \cdot \nabla_{y} f_1(x^*, y^*) + \lambda_2 \cdot \nabla_{y} f_2(x^*, y^*) &= 0 &&\text{(constraint for variable $y$)}\enspace.\\
\end{align*}

Evaluating the above derivatives leads to the following system of equations.
\begin{align}
\lambda_1 + \lambda_2 &= 1 \label{eq:kkt_z}\\
-\lambda_1\left(1 - 2 e^{-1} - e^{-1}\ln x^*\right)
-\frac{1}{2}\lambda_2 \left( e^{-1} -2 - \ln x^* \right) &=0\label{eq:kkt_x}\\
-\lambda_1 \left((1-e^{-1}) x^* - e^{-1} - e^{-1}x^*\ln x^*\right)
+\frac{1}{2} \lambda_2 \left( e^{-1} + x^* \ln x^* + (1-e^{-1}) x^* \right) &= 0\label{eq:kkt_y}
\end{align}
To derive~\eqref{eq:kkt_x}, we used the fact that $y^* >0$, which allowed us to divide both right-hand side and left-hand side by $y^*$.

Multiplying~\eqref{eq:kkt_x} by $-x^*$ and adding it to~\eqref{eq:kkt_y} leads to the equation
\begin{align}
\lambda_1 e^{-1} \left( x^* - 1 \right) &= \frac{1}{2}\lambda_2 (e^{-1} - x^*)\enspace,\notag
\intertext{which implies}
\lambda_1 &= \frac{\lambda_2 (1 - e x^*)}{2(x^*-1)}\label{eq:lam1lam2Rel}\enspace.
\end{align}
Finally, we use~\eqref{eq:lam1lam2Rel} to substitute $\lambda_1$ in~\eqref{eq:kkt_y}, and simplify the expression (by multiplying by $2(x^*-1)/\lambda_2$, expanding and then dividing by $(1/e - 1)x$) to obtain
\begin{equation}\label{eq:optCondOnX}
\ln x^* + 3 - (e+1) x^* = 0\enspace,
\end{equation}
as desired. Hence, any optimal solution $(x^*,y^*)$ to~\eqref{eq:optProbFor2PLBApp} must satisfy~\eqref{eq:optCondOnX} due to necessity of the Karush-Kuhn-Tucker conditions.

It remains to observe that~\eqref{eq:optCondOnX} has two solutions. One has $x$-value below $0.07$, and clearly does not lead to a minimizer because $f_1(x,y) \geq 1 - e^{-1}$ for any $x\leq 0.07$ and $y\in \mathbb{R}_{\geq 0}$. Thus, the $x^*$-value of the minimizer is unique and corresponds to the second solution of~\eqref{eq:optCondOnX}, which is $x^* \approx 0.7175647$, i.e., the value stated at the beginning of the proof. Finally, the minimizing value $y^*$ for $y$ can be obtained by computing the unique minimizer of $\max\{f_1(x^*,y), f_2(x^*, y)\}$, which is a maximum of two linear function, one with strictly positive and one with strictly negative derivative. Hence, the unique minimizer $y^*$ is the $y$-value that solves $f_1(x^*,y) = f_2(x^*, y)$, which leads to the expression for $y^*$ highlighted at the beginning of the proof.
\end{proof}